\documentclass[CRMECA,Unicode,manuscript]{cedram}

\usepackage{slashbox} 
\usepackage{subfig}
\usepackage{gensymb}
\usepackage{color}
\usepackage{tcolorbox} %
\usepackage{stmaryrd} 
\usepackage{cases,empheq,algorithm,enumitem}
\usepackage{ulem}
\usepackage{todonotes} 

\usepackage[scr=boondox,cal=euler]{mathalfa}

\graphicspath{{Figures/}}

\newtheorem{Lemma}{Lemma}

\newcommand{\ds}{\displaystyle}
\newcommand{\ts}{\textstyle}
\newcommand{\jump}[1]{\left\llbracket #1 \right\rrbracket}

\newcommand{\dmean}[1]{\llangle #1 \rrangle}

\newcommand{\llangle}{{\langle\!\langle}}

\newcommand{\rrangle}{{\rangle\!\rangle}}

\newcommand{\tmmathbf}[1]{\ensuremath{\boldsymbol{#1}}}
\newcommand{\tmop}[1]{\ensuremath{\operatorname{#1}}}

\definecolor{rouge}{rgb}{1,0,0}
\definecolor{bleu}{rgb}{0,0,1}
\definecolor{vert}{rgb}{0,0.5,0}

 \newcommand{\Na}{\mathfrak{a}}
 \newcommand{\Nb}{\mathfrak{b}}
 \newcommand{\Nh}{\mathfrak{h}}
 \newcommand{\Ng}{\mathfrak{g}}

\title{Scattering of transient waves by an interface with time-modulated jump conditions}

\author{\firstname{Micha\"el} \lastname{Darche}
\CDRorcid{0000-0001-5731-4551}}
\address{Aix Marseille Univ, CNRS, Centrale Marseille, LMA UMR 7031, Marseille, France}
\email[M. Darche]{darche@lma.cnrs-mrs.fr}

\author{\firstname{Rapha\"el} \lastname{Assier}}
\address{Department of Mathematics, University of Manchester, Oxford Road, Manchester M13 9PL, UK}
\email[R. Assier]{raphael.assier@manchester.ac.uk}

\author{\firstname{S\'ebastien} \lastname{Guenneau}}
\address{UMI 2004 Abraham de Moivre-CNRS, Imperial College London, London SW7 2AZ, UK}
\address{Department of Physics, The Blackett Laboratory, Imperial College London, London SW7 2AZ, UK}
\email[S. Guenneau]{s.guenneau@imperial.ac.uk}

\author{\firstname{Bruno} \lastname{Lombard}\IsCorresp}
\addressSameAs{1}{Aix Marseille Univ, CNRS, Centrale Marseille, LMA UMR 7031, Marseille, France}
\email[B. Lombard]{lombard@lma.cnrs-mrs.fr}

\author{\firstname{Marie} \lastname{Touboul}}
\addressSameAs{4}{UMI 2004 Abraham de Moivre-CNRS, Imperial College London, London SW7 2AZ, UK}
\address{Department of Mathematics, Imperial College London, Huxley Building, Queen's Gate, London SW7 2AZ, UK}
\address{POEMS, ENSTA, CNRS, INRIA, Institut Polytechnique de Paris, 91120, Palaiseau, France}
\email[M. Touboul]{marie.touboul@ensta.fr}
\thanks{The first author is funded as a post-doctoral researcher by the Institut M\'ecanique et Ing\'enierie (Marseille, France). S. G. is funded by UK Research and Innvovation (UKRI) under the UK government's Horizon Europe funding guarantee (grant number 10033143).} 

\ESM{Supplementary material for this article is supplied as a separate 
archive available from  the journal's
website under 
article's URL 
or from the author.}


\begin{abstract} 
Time modulation of the physical parameters offers interesting new possibilities for wave control. Examples include amplification of waves, harmonic generation and non-reciprocity, without resorting to non-linear mechanisms. Most of the recent studies focus on the time-modulation of the bulk physical properties. However, as the temporal modulation of these properties is difficult to achieve experimentally, we will concentrate here on the special case of an interface with time-varying jump conditions, which is simpler to implement. This work is focused on wave propagation in a one-dimensional medium containing one modulated interface. Properties of the scattered waves are investigated theoretically: energy balance, generation of harmonics, impedance matching and non-reciprocity. A fourth-order numerical method is also developed to simulate  transient scattering. Numerical experiments are conducted to validate the numerical scheme and to illustrate the theoretical findings.
\end{abstract}

\alttitle{Propagation d'ondes à travers une interface avec des conditions de saut modulées en temps}

\altkeywords{Ondes élastiques, Conditions de sauts imparfaites, Milieux modulés en temps, Non-réciprocité, Méthodes numériques pour les équations hyperboliques}

\begin{altabstract} 
  La modulation temporelle des propriétés d'un milieu de propagation offre de nouvelles possibilités intéressantes pour le contrôle des ondes. Les exemples incluent l'amplification des ondes, la génération d'harmoniques et la non-réciprocité, sans avoir recours à des mécanismes non linéaires.
  La plupart des études récentes se concentrent sur la modulation temporelle des propriétés physiques volumiques d'un milieu de propagation. Cependant, cette modulation étant difficile à réaliser expérimentalement, on se concentre ici sur le cas particulier d'une interface avec des conditions de saut variables en temps, plus simple à mettre en œuvre. Ce travail se concentre sur la propagation des ondes dans un milieu monodimensionnel contenant une interface modulée. On étudie théoriquement de nombreuses propriétés des ondes diffusées : bilan d'énergie, génération d'harmoniques, adaptation d'impédance et non-réciprocité. Une méthode numérique d'ordre 4 est développée pour simuler les ondes transitoires diffractées. Des expériences numériques sont menées pour valider le schéma numérique et illustrer les résultats théoriques attendus.
  
\end{altabstract}

\begin{document}

\maketitle

\section{Introduction}\label{SecIntro}

Since the early 2000s, periodic modulation of physical properties in space has led to significant developments in wave control. Based on homogenization theory and Floquet-Bloch analysis, these so-called {\it metamaterials} have made it possible to achieve cloaking, negative refraction or perfect lensing, to cite a few of the exotic effects. We refer the interested reader to \cite{craster2013} for an overview of this abundant subject in the case of acoustics.

In recent years, the technical possibility of modulating physical properties in time has opened new perspectives in the field of metamaterials \cite{caloz2019spacetime1,caloz2019spacetime2,galiffiA2022}. Time-varying metamaterials exhibit unusual phenomena such as time reflection and time refraction\cite{wang_temporal_2025}, nontrivial topology \cite{Swinteck2015,Lustig2018}, frequency conversion \cite{Salehi2022,yi_frequency_2017} unidirectional and parametric amplification \cite{torrentPRB2018,kimPRE2023,kiorpelidisPRB2024}. Numerous works have emerged since the late fifties, investigating this physics both theoretically \cite{cullen1958travelling,morgenthaler1958velocity,tien1958parametric,Cassedy1963,Cassedy1967,fante1971transmission,chu1972wave,weekes2001numerical,maestre2007spatio,lurie2007introduction,maestre2008dynamic,jensen2009space,to2009homogenization,jensen2010optimization,sanguinet2011homogenized,fang2012realizing,yuan2016photonic,milton2017field,movchanPTRSMPES2022,ammariJoMP2023,Ammari2022,Nassar2017,Huidobro2019,farhat2021spacetime,Huidobro2021,Pham2022,li2023space,Touboul2024} and experimentally \cite{ashkin1958parametric,couder2005walking,lira2012electrically,taravati2016mixer,fink2016loschmidt,mallejac2023scattering,tessierbrothelandeAPL2023,Tirole2022,Moussa2023,Lustig2023,Harwood2024,Goldsberry2025}. In particular, systems that are periodically modulated in both space and time have the ability to alleviate some of the constraints of static media, such as the breaking of reciprocity \cite{Goldsberry2022}.  As a result, exciting wave control possibilities have been identified, such as unidirectional amplification \cite{wenCP2022}, coherent perfect absorption \cite{Galiffi2024} and non-reciprocity \cite{Croenne2019,Nassar2020,tessierbrothelandeAPL2023,yi_reflection_2018}. 

Nevertheless, a current limitation relies on the difficulty to achieve time modulation of bulk parameters experimentally (e.g.\ density or Young's modulus of elasticity). It is easier to modify properties at discrete points, for example, by modifying the stiffness of a membrane or a surface impedance \cite{zhuPRB2020,zhuAPL2020}, or by vertical oscillations of a submerged plate within a water tank \cite{KOUKOURAKI2025103530}. 
Recent work has focused on the theoretical analysis of wave propagation in systems with time modulation at discrete points using multiple scattering \cite{puJoSaV2024}, transmission-line theory \cite{mallejacPRA2023}, and capacitance matrix \cite{ammari_scattering_2024,ammari_spacetime_2025}.

The aim of this paper is then to study theoretically and numerically the propagation of waves across a time-modulated interface. For this purpose, the article is structured as follows.
Section \ref{SecContinu} describes the time-dependent jump conditions, which generalize the static jump conditions widely used to describe imperfect contacts \cite{assierPRSMPES2020}. Interface dissipation is also incorporated to be closer to experimental devices. An energy balance is conducted. The generation of harmonics is studied through a harmonic balance analysis. The particular case of reflectionless modulated interface is discussed. Section \ref{SecNum} describes the time-domain numerical methods used to simulate transient wave phenomena. Integration of the momentum equation and of the constitutive law is done by a fourth-order finite-difference ADER scheme. The  time-varying jump conditions are discretized by the Explicit Simplified Interface Method (ESIM) \cite{lombardSJSC2003,lombardSJSC2006}, requiring new developments of this method. Validation of the numerical results is done by comparisons with a semi-analytical solution. Section \ref{SecSimus} presents different numerical experiments, illustrating the theoretical findings: amplification of waves, generation of harmonics, impedance matching, and non-reciprocity. Lastly, Section \ref{SecConclu} draws future lines of research such as homogenization of a network of such modulated interfaces.


\section{Physical modeling}\label{SecContinu}

\subsection{Problem statement}\label{SecContinutPb}

Let us consider a one-dimensional linear elastic medium with density $\rho(x)$ and Young's modulus $E(x)=\rho(x)\,c^2(x)$, where $c$ is the sound speed (the analysis is however relevant to other wave physics problems such as acoustics or s-polarization and p-polarization in electromagnetism, see Remark 1). These parameters may be discontinuous at $x=x_0$. The conservation of momentum and the Hooke's law are formulated as follows:
\begin{subnumcases}{\label{EDP-Hooke}}
\ds \partial_t(\rho(x)\,\partial_{t}u(x,t)) =\partial_x\sigma(x,t)+F(x,t),\label{EDP-Hooke-a}\\ [8pt]
\ds \sigma(x,t) = E(x)\,\partial_x u(x,t),\label{EDP-Hooke-b}
\end{subnumcases}
with $u$ being the displacement field, $\sigma$ the stress field, and $F$ a body force with compact support that does not contain $x_0$. The velocity field is denoted by $v=\partial_t u$. For any function $g(x)$, we define the jump and mean operators $\jump{.}_{x_0}$ and $\dmean{.}_{x_0}$ as 
\begin{equation}
\jump{g}_{x_0}=g^+(x_0)-g^-(x_0),\qquad \dmean{g}_{x_0}=\frac{1}{2}\left(g^+(x_0)+g^-(x_0)\right),
\label{jump_mean}
\end{equation}
where 
\begin{equation}
g^\pm(x_0)=\lim_{\eta \rightarrow 0^+}g(x_0\pm \eta).
\label{Gpm}
\end{equation}
To simplify the notations, the indices and arguments denoting space and time will be omitted when no risk of ambiguity occurs; in particular, $x_0$ is omitted in the jump and mean operators from now on. 
As can be seen from direct computations, for any two functions $g$ and $h$, these operators satisfy the property
\begin{equation}
\jump{g\,h}=\jump{g}\,\dmean{h}+\dmean{g}\,\jump{h}.
\label{PropJumpMean}
\end{equation}
Introducing interface parameters of stiffness $\mathscr{K}(t)>0$, compliance  $\mathscr{C}(t)=1/\mathscr{K}(t)$, inertia $\mathscr{M}(t)\geq 0$ and dissipation $\mathscr{Q}_{C,M}(t)\geq 0$, the imperfect interface is modelled by the time-dependent jump conditions:
\begin{subnumcases}{\label{JC1interf}}
\ds \jump{v(\cdot,t)}=\partial_t\left(\mathscr{C}(t)\,\dmean{\sigma(\cdot,t)}\right)+\mathscr{Q}_C(t)\,\dmean{\sigma(\cdot,t)},\label{JC1interf-K}\\ [8pt]
\ds \jump{\sigma(\cdot,t)}=\partial_t\left(\mathscr{M}(t)\,\dmean{v(\cdot,t)}\right)+\mathscr{Q}_M(t)\,\dmean{v(\cdot,t)}.\label{JC1interf-M}
\end{subnumcases} 
Five remarks arise from \eqref{JC1interf}:
\begin{itemize}
\item The modulated jump conditions are assumed to result from an external stimulus. Such a modulation of $\jump{v}$ and $\jump{\sigma}$ can be obtained by piezoelectric components \cite{tessierbrothelandeAPL2023} and time-modulated membranes \cite{zhuAPL2020}, respectively; 
\item As seen further in Section \ref{SecContinuNRJ}, the compliance $\mathscr{C}(t)$ and the inertia $\mathscr{M}(t)$ must be {\it inside} the time derivative to ensure a sound energy balance;
\item  The terms in \eqref{JC1interf-K} are reminiscent of the Maxwell model of viscoelasticity in the case of a static dissipation parameter $\mathscr{Q}_C(t)=\mathscr{Q}_C$; 
\item When $\mathscr{C}=0$, $\mathscr{M}=0$, $\mathscr{Q}_C=0$ and $\mathscr{Q}_M=0$, the conditions of perfect bonding are recovered; 
\item If $\mathscr{Q}_C=0$, then \eqref{JC1interf-K} can be recast in the more usual form with displacement:
$$
\ds \jump{u(\cdot,t)}=\mathscr{C}(t)\,\dmean{\sigma(\cdot,t)};$$
\end{itemize}

\begin{Remark}
The problem studied in this article is generic to various wave physics problems. For instance, the case of acoustics is obtained by changing $u$ into the acoustic pressure, $E$ into the inverse of mass density and $\rho$ into the compressibility. Similarly, the case of s-polarization (resp. p-polarization) in electromagnetism can be obtained by changing $u$ into the transverse electric (resp. magnetic) field, $\rho$ into the permittivity (resp. permeability) and $E$ into the inverse of the permeability (resp. inverse of the permittivity).
\end{Remark}


\subsection{Energy balance}\label{SecContinuNRJ}

We take the product of the momentum equation \eqref{EDP-Hooke-a} by $v=\partial_t u$ and integrate over $(-\infty,x_0)$. Integrating by parts and using \eqref{EDP-Hooke-b} give
\begin{equation}
\begin{array}{lll}
\ds \int_{-\infty}^{x_0}\rho\,v\,\partial_t v\,dx &=& \ds \int_{-\infty}^{x_0}v\,\partial_x\sigma\,dx+\int_{-\infty}^{x_0}F\,v\,dx,\\ [10pt]
&=& \ds -\int_{-\infty}^{x_0}\sigma\,\partial_xv\,dx+[v\,\sigma]_{-\infty}^{x_0}+\int_{-\infty}^{x_0}F\,v\,dx,\\ [10pt]
&=& \ds -\int_{-\infty}^{x_0}\frac{1}{E}\,\sigma\,\partial_t\sigma\,dx+v^-(x_0,t)\,\sigma^-(x_0,t)+\int_{-\infty}^{x_0}F\,v\,dx.
\end{array}
\label{IPP1}
\end{equation}
Similarly, integrating over $(x_0,+\infty)$ yields
\begin{equation}
\int_{x_0}^{+\infty}\rho\,v\,\partial_t v\,dx= -\int_{x_0}^{+\infty}\frac{1}{E}\,\sigma\,\partial_t\sigma\,dx-v^+(x_0,t)\,\sigma^+(x_0,t)+\int^{+\infty}_{x_0}F\,v\,dx,
\label{IPP2}
\end{equation}
where we have used the fact that $F$ is compactly supported and that the fields vanish at infinity. Summing \eqref{IPP1} and \eqref{IPP2} leads to
\begin{equation}
\frac{d}{dt} \left(\frac{1}{2}\int_{\mathbb{R}}\left(\rho\,v^2+\frac{1}{E}\,\sigma^2\right)\,dx\right)=-\jump{v\,\sigma}+\int_{\mathbb{R}}F\,v\,dx.
\label{IPP3}
\end{equation}
The property \eqref{PropJumpMean} and the jump conditions \eqref{JC1interf} lead to
\begin{equation}
\begin{array}{lll}
\ds \jump{v\,\sigma} &=& \ds\jump{v}\dmean{\sigma}+\dmean{v}\jump{\sigma}, \\ [8pt]
&=& \ds \left(\mathscr{C}'(t)\,\dmean{\sigma}+\mathscr{C}(t)\,\dmean{\partial_t \sigma}+\mathscr{Q}_C(t)\,\dmean{\sigma}\right)\dmean{\sigma}+\left(\mathscr{M}'(t)\,\dmean{v}+\mathscr{M}(t)\,\dmean{\partial_t v}+\mathscr{Q}_M(t)\,\dmean{v}\right)\dmean{v},\\ [8pt]
&=& \ds \frac{d}{dt} \left(\frac{1}{2}\left(\mathscr{C}(t)\,\dmean{\sigma}^2 + \mathscr{M}(t)\,\dmean{v}^2\right)\right)+\frac{1}{2}\left(\mathscr{C}'(t)\dmean{\sigma}^2+\mathscr{M}'(t)\dmean{v}^2\right)+\mathscr{Q}_C(t)\,\dmean{\sigma}^2+\mathscr{Q}_M(t)\,\dmean{v}^2.
\end{array}
\end{equation}
It leads to the following energy balance.

\begin{Proposition}
Let ${\mathcal E}_m={\mathcal E}_b+{\mathcal E}_i$ be the total energy, with the bulk energy
\begin{equation}
{\mathcal E}_b(t)=\frac{1}{2}\int_{\mathbb{R}}\left(\rho\,v^2+\frac{1}{E}\,\sigma^2\right)\,dx
\label{NRJv}
\end{equation}
and the interface energy
\begin{equation}
{\mathcal E}_i(t)=\frac{1}{2}\mathscr{M}\,\dmean{v}^2+\frac{1}{2}\mathscr{C}\,\dmean{\sigma}^2.
\label{NRJi}
\end{equation}
Introducing the power of external forces due to body forcing
\begin{equation}
{\mathcal P}(t)=\int_{\mathbb{R}}F\,v\,dx,
\end{equation}
one has
\begin{equation}
\frac{d}{dt}{\mathcal E}_m(t)={\mathcal P}(t)-\frac{1}{2}\left( \mathscr{M}'(t)\,\dmean{v}^2+ \mathscr{C}'(t)\,\dmean{\sigma}^2\right)-\left(\mathscr{Q}_C(t)\,\dmean{\sigma}^2+\mathscr{Q}_M(t)\,\dmean{v}^2\right).
\label{dEdT}
\end{equation}
\label{PropNRJ}
\end{Proposition}

\noindent
Five remarks arise from Proposition \ref{PropNRJ} and positivity of $\mathscr{M}$, $\mathscr{C}$, $\mathscr{Q}_C$ and $\mathscr{Q}_M$:
\begin{itemize}
\item both ${\mathcal E}_b$ and ${\mathcal E}_i$ are positive, hence ${\mathcal E}_m \geq 0$;
\item in ${\mathcal E}_i$ one recognizes the kinetic energy and the potential energy of a spring-mass system when static parameters are considered;
\item if there are no body force ($F=0$), no time-modulation (both $\mathscr{M}$ and $\mathscr{C}$ are constant in time) and no dissipation term in the jump conditions (${\mathscr{Q}_C}$=0 and ${\mathscr{Q}_M}$=0), then the total mechanical energy is conserved;
\item if no body force is applied ($F=0$) and both $\mathscr{M}$ and $\mathscr{C}$ are constant in time, but ${\mathscr{Q}_C}$ or ${\mathscr{Q}_M}$ is non-zero, then the total mechanical energy is dissipated;
\item if $F=0$ but either $\mathscr{M}$ or $\mathscr{C}$ varies with time, then the total mechanical energy may vary, due to the power of external forces required to modify $\mathscr{M}$ and $\mathscr{C}$. The sign of the right-hand-side in \eqref{dEdT} is arbitrary, so that ${\mathcal E}_m$ may increase or decrease.
\end{itemize}

\subsection{Some additional properties}
The latter remark about the amplification of energy raises the question of possible parametric amplification. Under suitable assumptions, the next Proposition states that such an amplification is not possible, and that the solution remains bounded. The proof is given in Appendix \ref{AppAmpli}.

\begin{Proposition}
Let us assume:
\begin{description}
    \item[(i)] a sinusoidal modulation of the interface compliance $\mathscr{C}(t)$ and the dissipation term $\mathscr{Q}_C(t)$;
    \item[(ii)] no jump of stress (ie $\mathscr{M}(t)=0$ and $\mathscr{Q}_M(t)=0$);
    \item[(iii)] a body force in \eqref{EDP-Hooke-a} in the form $F(x,t)=\delta(x-x_s)\,S(t)$, where $S$ is bounded.
\end{description}
Then the scattered fields $v$ and $\sigma$ remain bounded for all $t$.
\label{PropAmpli}
\end{Proposition}

This result is counter-intuitive, since modulated systems often exhibit resonance (e.g.\ a modulated spring-mass system). An interpretation is that the energy may be evacuated on both sides of the domain, preventing from an unbounded increase. A similar conclusion holds by interchanging the roles of $\mathscr{C}$ and $\mathscr{M}$. The case where both interface compliance and interface inertia are modulated remains open but this limitation seems, a priori, purely technical. Indeed, our numerical experiments suggest that the same result will hold in this case.

\begin{Proposition}
When only $\mathscr{Q}_C(t)\neq 0$, \textit{i.e} $\mathscr{M}(t)=0$, $\mathscr{C}(t)=0$ and $\mathscr{Q}_M(t)=0$, the limits between which the solutions evolve are determined analytically (Appendix \ref{AppAnalytic}) and correspond to the static cases with the extreme values of $\mathscr{Q}_C(t)$. The limits for the case with only $\mathscr{Q}_M(t)\neq0$ can be obtained using the same reasoning.
\label{PropEnvelop}
\end{Proposition}


\subsection{Generation of harmonics}\label{SecContinuHarm}

The interaction of a single-frequency wave with a modulated interface generates an infinite number of harmonics. Here we shortly describe how to compute them. For simplicity, constant physical parameters are considered around the interface at $x_0=0$. In this part, a sinusoidal modulation for $\mathscr{C}$,   $\mathscr{M}$, $\mathscr{Q}_C$ and $\mathscr{Q}_M$ is chosen: 
\begin{equation}
\begin{array}{l}
\ds \mathscr{C}(t)=\mathscr{C}_0\left(1+\varepsilon_C\,\sin(\Omega t)\right),\\ [8pt]
\ds \mathscr{M}(t)=\mathscr{M}_0\left(1+\varepsilon_M\,\sin(\Omega t)\right),\\ [8pt]
\ds \mathscr{Q}_C(t)=\mathscr{Q}_{C_0}\left(1+\varepsilon_{Q_C}\,\sin(\Omega t)\right),\\ [8pt]
\ds \mathscr{Q}_M(t)=\mathscr{Q}_{M_0}\left(1+\varepsilon_{Q_M}\,\sin(\Omega t)\right),
\end{array}
\label{ModulSinus}
\end{equation}
with $\Omega=2\,\pi\,f_m$, where $f_m$ is the modulation frequency, and $\mathscr{C}_0\geq0$, $\mathscr{M}_0\geq0$ and $-1<\varepsilon_X<1$, with $X=C,M,Q_C,Q_M$.
 A harmonic incident wave $v^{\tmop{in}}$ and $\sigma^{\tmop{in}}$ of the form 
\begin{eqnarray*}
  v^{\tmop{in}} (x, t) & = & e^{i \omega (t - x / c)};\qquad \sigma^{\tmop{in}} (x, t)  =  -\rho c \,e^{i \omega (t - x / c)}
\end{eqnarray*}
impacts the interface from the left. Thus, we can write
\begin{eqnarray*}
  v^{\tmop{tot}} (x, t) & = & \left\{ \begin{array}{ccc}
    v^{\tmop{in}} (x, t) + v^{\tmop{ref}} (x, t) & \tmop{for} & x < 0\\
    v^{\tmop{trans}} (x, t) & \tmop{for} & x > 0,
  \end{array} \right. 
\end{eqnarray*}
and \begin{eqnarray*}
  \sigma^{\tmop{tot}} (x, t) & = & \left\{ \begin{array}{ccc}
    \sigma^{\tmop{in}} (x, t) + \sigma^{\tmop{ref}} (x, t) & \tmop{for} & x < 0\\
    \sigma^{\tmop{trans}} (x, t) & \tmop{for} & x > 0,
  \end{array} \right. 
\end{eqnarray*}
for some reflected and transmitted wave fields $v^{\tmop{ref}}$, $\sigma^{\tmop{ref}}$, $v^{\tmop{trans}}$ and $\sigma^{\tmop{trans}}$ that also satisfy the governing wave equation. Floquet theorem states that these fields write
\begin{eqnarray}{\label{AnsatzRT}}
  v^{\tmop{ref}} (x, t) = \sum_{k \in \mathbb{Z}} R_k e^{ i \omega_k (t +
  x / c)}, &\qquad &\sigma^{\tmop{ref}} (x, t)  = \rho c\sum_{k \in \mathbb{Z}} R_k  e^{ i \omega_k (t +
  x / c)}, \label{eq:ansatz-ref}\\
  v^{\tmop{trans}} (x, t)  =  \sum_{k \in \mathbb{Z}} T_k e^{ i \omega_k (t
  - x / c)}, &\qquad &\sigma^{\tmop{trans}} (x, t)  = - \rho c \sum_{k \in \mathbb{Z}} T_k e^{ i \omega_k (t
  - x / c)},  \label{eq:ansatz-trans}
\end{eqnarray}
for some reflection and transmission coefficients $R_k$ and $T_k$ to be determined, and where we denote $\omega_k \equiv  \omega + k \Omega$. One introduces the auxiliary quantities $\Psi_k = T_k - R_k$ and $\Phi_k  =  T_k + R_k$.



\subsubsection{Dealing with the first jump equation}

\noindent 
Using \eqref{eq:ansatz-ref} and \eqref{eq:ansatz-trans}, together with the definition of the jump operator, we find that
\begin{eqnarray*}
  \llbracket v^{\tmop{tot}} \rrbracket  =  e^{ i \omega t} \left( - 1 +
  \sum_{k \in \mathbb{Z}} \Psi_k e^{ik \Omega t} \right) & \text{ and } &   \llangle \sigma^{\text{tot}} \rrangle = - \frac{\rho\,c}{2} e^{ i \omega t} \left(1+\sum_{k \in \mathbb{Z}} \Psi_k e^{ik \Omega t} \right).
\end{eqnarray*}
Similarly, though it is slightly longer, we find that
\begin{eqnarray*}
  \partial_t \llangle \mathscr{C}(t)\, \sigma^{\tmop{tot}} \rrangle  &= & e^{i
  \omega t} \left( \frac{\mathcal{C}  \varepsilon_C}{4}\omega_{-1} e^{- i \Omega t} -
  \frac{i\mathcal{C} }{2}\omega - \frac{\mathcal{C}  \varepsilon_C}{4} \omega_1e^{+ i
  \Omega t}\right) \\
  & + & e^{i \omega t}\sum_{k \in \mathbb{Z}} \left( - \frac{\mathcal{C}
  \varepsilon_C}{4} \omega_{k} \Psi_{k - 1} - \frac{i\mathcal{C}}{2} 
 \omega_{k}  \Psi_k + \frac{\mathcal{C} \varepsilon_C}{4} \omega_{k}  \Psi_{k + 1} \right)
  e^{ ik \Omega t},
\end{eqnarray*}
where we have introduced the reduced quantity $\mathcal{C} = \rho\,c\,\mathscr{C}_0$. Therefore the first jump condition \eqref{JC1interf-K}, after dividing through by $e^{i \omega t}$ and rearranging slightly becomes
\begin{equation}
\begin{array}{l}
\ds \sum_{k \in \mathbb{Z}} \left(\frac{\mathcal{C} \varepsilon_C\omega_{k}-i\mathcal{Q}_C\varepsilon_{Q_C}}{4}  \Psi_{k - 1} + \left( 1 +\frac{\mathcal{Q}_C}{2}+ i\frac{\mathcal{C}}{2} \omega_k  \right) \Psi_k -\frac{\mathcal{C} \varepsilon_C\omega_{k}-i\mathcal{Q}_C\varepsilon_{Q_C}}{4}  \Psi_{k +1}  \right) e^{ik\Omega t} \\ [12pt]
\ds \hspace{0.7cm} =  \frac{\mathcal{C}  \varepsilon_C\omega_{-1}-i\mathcal{Q}_C\varepsilon_{Q_C}}{4}\,e^{- i \Omega t} + \left(1-\frac{\mathcal{Q}_C}{2} - i\frac{\mathcal{C} \omega}{2} \right) - \frac{\mathcal{C}  \varepsilon_C\omega_1-i\mathcal{Q}_C\varepsilon_{Q_C}}{4}\,e^{+ i \Omega t},
\end{array}
\end{equation}
where $\mathcal{Q}_C = \rho\,c\,\mathscr{Q}_{C_0}$.
Since the $e^{ ik \Omega t}$ functions are linearly independent, we obtain this system in matrix form:
\begin{equation}
  \mathbb{A} (\mathcal{C}, \varepsilon_C, \mathcal{Q}_C, \varepsilon_{Q_C}, \omega, \Omega) \,\tmmathbf{\Psi} = 
  {\bf V} (\mathcal{C}, \varepsilon_C, \mathcal{Q}_C, \varepsilon_{Q_C}, \omega, \Omega),
\label{SysPsi}
\end{equation}
where 
\begin{equation}
\begin{array}{l}
\ds \tmmathbf{\Psi} \equiv \left(\hdots\, \Psi_{-N},\, \hdots\, \Psi_{-1},\, \Psi_0,\, \Psi_{ 1},\, \hdots\, \Psi_{ N},\, \hdots
  \right)^\top,\\ [8pt]
\ds {\bf V} (\mathcal{C}, \varepsilon_C, \mathcal{Q}_C, \varepsilon_{\mathscr{Q}_C},
  \omega,\Omega) = \left( \hdots\, 0\, \hdots\, \frac{\mathcal{C} \varepsilon_C}{4}\omega_{-1}-i\frac{\mathcal{Q}_C\varepsilon_{Q_C}}{4} ,\, 1-\frac{\mathcal{Q}_C}{2} - i\frac{\mathcal{C} }{2}\omega,\, -\frac{\mathcal{C}  \varepsilon_C}{4}\omega_1+i\frac{\mathcal{Q}_C\varepsilon_{Q_C}}{4},\, \hdots\, 0\, \hdots \right)^\top,
  \end{array}
\end{equation}
and $\mathbb{A}$ is a tri-banded antidiagonal matrix. Now we can truncate this system so that $\mathbb{A}$ becomes a
$(2 N + 1) \times (2 N + 1)$ matrix, and the truncated vector
$\tmmathbf{\Psi}$ can be recovered by simple inversion.



\subsubsection{Dealing with the second jump equation}

\noindent
One obtains directly 
\begin{eqnarray*}
   \llbracket \sigma^{\tmop{tot}} \rrbracket  =   \rho c\,e^{i \omega t} \left(1 -  \sum_{k \in \mathbb{Z}} \Phi_k e^{
  ik \Omega t} \right) & \text{ and } & \llangle v^{\text{tot}} \rrangle=\frac{1}{2} e^{i\omega t} \left( 1+\sum_{k\in\mathbb{Z}} \Phi_k e^{
  ik \Omega t} \right).
\end{eqnarray*}
Similarly, we find that
\begin{eqnarray*}
  \partial_t  \llangle \mathscr{M} (t) {v^{\tmop{tot}}} \rrangle & = & e^{i
  \omega t} \left(- \frac{\mathscr{M}_0 \,\varepsilon_M}{4}\omega_{-1} e^{- i \Omega t} +
  \frac{i\mathscr{M}_0 }{2}\omega + \frac{\mathscr{M}_0 \, \varepsilon_M}{4}\omega_1 e^{+ i
  \Omega t}\right) \\
  & + & e^{i \omega t}\sum_{k \in \mathbb{Z}} \left(  \frac{\mathscr{M}_0\,
  \varepsilon_M}{4} \omega_{k} \Phi_{k - 1} + \frac{i\mathscr{M}_0}{2} 
 \omega_{k}  \Phi_k - \frac{\mathscr{M}_0\,\varepsilon_M}{4} \omega_{k}  \Phi_{k + 1} \right)
  e^{ ik \Omega t}.
\end{eqnarray*}
Upon introducing $\mathcal{M} = \mathscr{M}_0/(\rho c)$ and $\mathcal{Q}_M=\mathscr{Q}_{M_0}/(\rho c)$, the second jump condition (\ref{JC1interf-M}) leads to
\begin{equation}
\begin{array}{l}
\ds \sum_{k \in \mathbb{Z}} \left(\frac{\mathcal{M}\,\varepsilon_M\omega_{k}-i\mathcal{Q}_M\varepsilon_{Q_M}}{4}  \Phi_{k - 1} + \left( 1 +\frac{\mathcal{Q}_M}{2}+ i\frac{\mathcal{M}}{2} \omega_k \right) \Phi_k -\frac{\mathcal{M}\,\varepsilon_M\omega_{k}-i\mathcal{Q}_M\varepsilon_{Q_M}}{4}  \Phi_{k +1}  \right) e^{ik\Omega t} \\ [12pt]
\ds \hspace{0.7cm} =  \frac{\mathcal{M}\, \varepsilon_M\omega_{-1}-i\mathcal{Q}_M\varepsilon_{Q_M}}{4}\,e^{- i \Omega t} + \left(1 -\frac{\mathcal{Q}_M}{2} - i\frac{\mathcal{M} \omega}{2}\right) - \frac{\mathcal{M}\,\varepsilon_M\omega_1-i\mathcal{Q}_M\varepsilon_{Q_M}}{4}\,e^{+ i \Omega t}.
\end{array}
\end{equation}
We end up with a system of the form
\begin{equation}
\mathbb{B} (\mathcal{M}, \varepsilon_{M},\mathcal{Q}_M, \varepsilon_{Q_M}, \omega, \Omega) \,\tmmathbf{\Phi} = {\bf W} (\mathcal{M}, \varepsilon_M,\mathcal{Q}_M, \varepsilon_{Q_M},\omega,\Omega).
\label{SysPhi}
\end{equation}
The matrix $\mathbb{B}$ and the vector ${\bf W}$ are obtained  by replacing $(\mathcal{C}, \varepsilon_C, \mathcal{Q}_C, \varepsilon_{Q_C}, \omega, \Omega)$ by $(\mathcal{M}, \varepsilon_M , \mathcal{Q}_M, \varepsilon_{Q_M},\omega,\Omega)$ in $\mathbb{A}$ and ${\bf V}$. The vector $\tmmathbf{\Phi}$ is
$$
(\cdots, \Phi_{-N}, \cdots, \Phi_{-1}, \Phi_0, \Phi_{ 1}, \cdots, \Phi_{ N},\cdots)^\top.
$$
As before, this system can be solved for any truncation by simple inversion.



\subsubsection{Final reflection and transmission coefficients}

\noindent
Once $\Psi_k$ and $\Phi_k$ are computed, the reflection and transmission coefficients are given by:
\begin{equation}
 R_k  =  \frac{1}{2} (\Phi_k - \Psi_k),\qquad T_k  =  \frac{1}{2} (\Psi_k + \Phi_k).
\label{RkTk}
\end{equation}
Two limit-cases are highlighted:
\begin{itemize}
    \item Without modulation, the only non-zero coefficients are $R_0$ and $T_0$. These coefficients are given in Appendix \ref{AppRT};
    \item If $\mathcal{M}=\mathcal{C}$, $\varepsilon_M = \varepsilon_C$, $\mathcal{Q}_C=\mathcal{Q}_M$ and $\varepsilon_{Q_C} = \varepsilon_{Q_M}$, then the two matrix systems \eqref{SysPsi} and \eqref{SysPhi} are exactly the same, and so $\tmmathbf{\Phi}=\tmmathbf{\Psi}$, which leads to $R_k = 0$ for all $k$. This property of {\it impedance matching} is summarized in the following proposition.
\end{itemize}

\begin{Proposition}
Let us consider a homogeneous medium of impedance $Z=\rho\,c$. Under the following equalities
\begin{equation}
\mathscr{M}_0=Z^2\,\mathscr{C}_0, \qquad \mathscr{Q}_{M_0}=Z^2\,\mathscr{Q}_{C_0}, \qquad
\varepsilon_C=\varepsilon_M, \qquad
\varepsilon_{\mathcal{Q}_C} = \varepsilon_{\mathcal{Q}_M},
\label{Z_KMT}
\end{equation}
then no reflected wave is generated at the modulated interface.
\label{PropImpedance}
\end{Proposition}

\noindent
Proposition \ref{PropImpedance} generalizes the impedance matching condition already obtained with static jump conditions and no dissipation in \cite{assierPRSMPES2020}. 

\begin{Remark}
Proposition \ref{PropImpedance} can be extended to any periodic modulation, while respecting $\mathscr{M}(t)=Z^2\mathscr{C}(t)$ and $\mathscr{Q}_M(t)=Z^2\mathscr{Q}_C(t)$, using the linearity of the equations and the fact that periodic functions can be decomposed into Fourier series.
\label{RqueImpedance}
\end{Remark}


\section{Numerical modeling}\label{SecNum}

\subsection{Numerical scheme}\label{SecNumScheme}

\paragraph{\bf First-order formulation}A velocity-stress formulation of the evolution equations is used. We introduce the compact notations
\begin{equation}
\begin{array}{lll}
{\bf U}(x,t)=\left( 
\begin{array}{c}
v(x,t)\\
\sigma(x,t)
\end{array} 
\right), \hspace{1cm} &
{\bf F}(x,t)=\left( 
\begin{array}{c}
\ds \frac{F(x,t)}{\rho(x)}\\
0
\end{array} 
\right), \hspace{0.5cm} &\\ [12pt]
{\bf A}(x)=\left(
\begin{array}{cc}
0 & \ds -\frac{1}{\rho(x)}\\
-E(x) & 0
\end{array}
\right), & 
{\bf B}(t)=\left(
\begin{array}{cc}
0 & \ds \mathscr{C}(t)\\
\mathscr{M}(t) & 0
\end{array}
\right), &
{\bf E}(t)=\left(
\begin{array}{cc}
0 & \ds \mathscr{Q}_C(t)\\
\mathscr{Q}_M(t) & 0
\end{array}
\right),
\end{array}
\label{SysOrdre1}
\end{equation}
where ${\bf A}$ and ${\bf E}$ can be discontinuous at the interface.
From \eqref{EDP-Hooke} and from \eqref{JC1interf}, one obtains the boundary-value problem
\begin{subnumcases}{\label{BVP}}
\ds \partial_t {\bf U}+{\bf A}\,\partial_x{\bf U}={\bf F},\qquad  x\neq x_0,\label{BVP-EDP}\\ [8pt]
\jump{{\bf U}(\cdot,t)}=\partial_t\left({\bf B}(t)\,\dmean{{\bf U}(\cdot,t)}\right)+{\bf E}(t)\,\dmean{{\bf U}(\cdot,t)}.
 \label{BVP-JC}
\end{subnumcases}
It is assumed that the body force is zero at the interface: ${\bf F}(x_0,t)={\bf 0}$, which is not restrictive. 


\paragraph{\bf ADER-4 scheme}A uniform grid with space step $\Delta x$ and time step $\Delta t$ is introduced, to determine an approximation ${\bf U}_j^n$ of ${\bf U}(x_j=j\,\Delta x, t_n=n\,\Delta t)$. Integration of \eqref{BVP-EDP} is done by a fourth-order ADER scheme. This explicit two-step finite-difference scheme is fourth-order accurate in space and time, using a stencil of 2 neighboring points on the right and on the left. It is stable under the usual CFL condition 
\begin{equation}
\zeta=\max(c(x))\,\frac{\Delta t}{\Delta x}\leq 1,
\label{CFL}
\end{equation}
where $c(x)=\sqrt{\frac{E(x)}{\rho(x)}}$ is the sound speed.
The reader is referred to \cite{benjaziaWM2014} for implementation details.


\subsection{Immersed interface method for time-modulated jump conditions}\label{SecNumIIM}

Discretization of the modulated jump conditions \eqref{BVP-JC} is done by implementing the Explicit Simplified Interface Method (ESIM). The philosophy of this numerical method is as follows: based on the jump conditions and the interface geometry, smooth extensions of the solution ${\bf U}$ on either side of the interface are constructed at each time step (Figure \ref{fig:IIM}). Evaluation of these extensions leads to modified values ${\bf U}^\star_j$ of the solution. At a point $x_j$ on one side of the interface, the integration scheme (ADER 4 here) then uses the usual values ${\bf U}_j^n$ if $x_j$ is on the same side, and ${\bf U}_j^\star$ if $x_j$ is on the other side of the interface. This method has been applied to various interface problems, notably to constant imperfect interfaces in 1D \cite{lombardSJSC2003} and 2D \cite{lombardSJSC2006}. 
\begin{figure}
    \centering
    \includegraphics[width=0.65\linewidth]{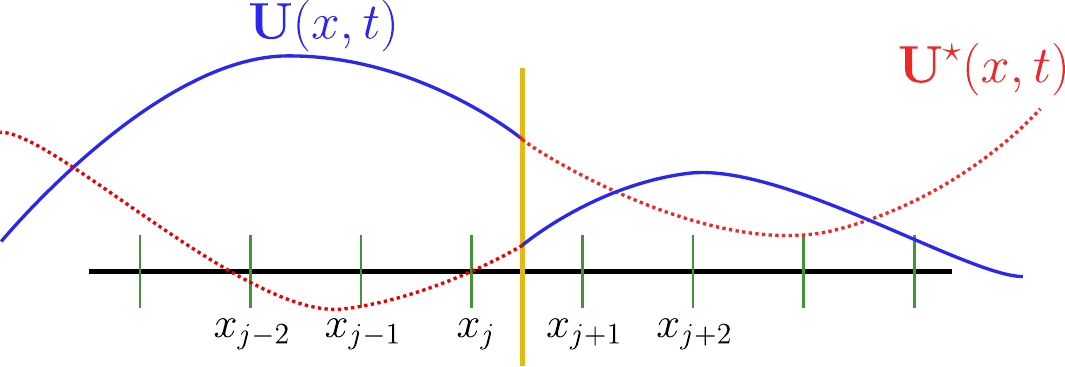}
    \caption{Scheme of the principle of ESIM at an interface with a discontinuous solution ${\bf U}(x,t)$ (blue line) and its smooth extensions ${\bf U}^\star(x,t)$ (red dotted lines).}
    \label{fig:IIM} 
\end{figure}

An important ingredient of ESIM is to determine the jump conditions satisfied by the solution and its successive spatial derivatives. As the jump conditions here vary as a function of time, an important modification has to be made to the method, as detailed now. From \eqref{jump_mean} and \eqref{BVP-JC}, we obtain
\begin{equation}
{\bf U}^+-\frac{1}{2}\,\partial_t({\bf B\,U}^+)-\frac{1}{2}\,{\bf E}\,{\bf U}^+={\bf U}^-+\frac{1}{2}\,\partial_t({\bf B\,U}^-)+\frac{1}{2}\,{\bf E}\,{\bf U}^-,
\label{Esim1}
\end{equation}
where the superscripts $\pm$ refer to the limit values, as in \eqref{Gpm}. Using \eqref{BVP-EDP}, we then have
\begin{equation}
\left({\bf I}-\frac{1}{2}\left(\partial_t{\bf B}+{\bf E}\right)\right)\,{\bf U}^++\frac{1}{2}{\bf B\,A}^+\,\partial_x{\bf U}^+=\left({\bf I}+\frac{1}{2}\left(\partial_t{\bf B}+{\bf E}\right)\right)\,{\bf U}^--{\frac{1}{2}}{\bf B\,A}^-\,\partial_x{\bf U}^-,
\label{Esim2}
\end{equation}
where ${\bf I}$ is the $2\times 2$ identity matrix. Successive time-derivatives of \eqref{Esim2} up to an order $k$ yields the matrix relation
\begin{equation}
{\bf C}_k^+(t)\,{\bf U}_k^++{\boldsymbol \Gamma}_{k+1}^+={\bf C}_k^-(t)\,{\bf U}_k^-+{\boldsymbol \Gamma}_{k+1}^-,
\label{Esim3}
\end{equation}
with the $2\,(k+1)$ vectors of limit values
\begin{equation}
\begin{array}{l}
\ds {\bf U}_k^{\pm}=\left({\bf U}^{\pm},\,\partial_x{\bf U}^{\pm},\cdots,\,\partial^k_x{\bf U}^{\pm}\right)^\top, \\ [8pt]
\ds {\boldsymbol \Gamma}_{k+1}^\pm=\left({\bf 0},\,{\bf 0},\cdots,\,{\bf 0},\,\mp\frac{(-1)^{k+1}}{2}\,{\bf B}\left({\bf A}^\pm\right)^{k+1}\,\partial_x^{k+1}{\bf U}^\pm\right)^\top,
\end{array}
\end{equation}
where ${\bf C}_k^\pm$ are $(2(k+1))\times (2(k+1))$ time-dependent matrices, ${\bf 0}$ is the $2 \times 2$ zero matrix, and ${\boldsymbol \Gamma}_{k+1}^\pm$ are $2(k+1)$ vectors. The matrices ${\bf C}_k^\pm$ depend on the bulk parameters (through the matrices ${\bf A}^\pm$) and the successive time derivatives of the jump conditions (through the matrix ${\bf B}(t)$ and ${\bf E}(t)$ in \eqref{SysOrdre1}). Moreover, the matrices ${\bf C}^\pm$ have a structure of block lower diagonal matrix; the $2 \times 2$ blocks below the diagonal depend on the successive time derivatives of $\mathscr{C}(t)$, $\mathscr{M}(t)$, $\mathscr{Q}_C(t)$ and $\mathscr{Q}_M(t)$.

Neglecting the $(k+1)$-th spatial derivatives (${\boldsymbol \Gamma}_{k+1}^\pm={\bf 0}_{2(k+1)}$) and inverting \eqref{Esim3} yields
\begin{equation}
\begin{array}{lll}
{\bf U}_k^+ &=&\left({\bf C}_k^+(t)\right)^{-1}\,{\bf C}_k^-(t)\,{\bf U}_k^-,\\ [8pt]
&:=& {\bf D}_k(t)\,{\bf U}_k^-.
\end{array}
\label{Esim4}
\end{equation}
The transfer matrices ${\bf D}_k(t)$ of order $k$ \eqref{Esim4} are the building-block of the ESIM. These are used to calculate extrapolation matrices, which are used to determine the modified values ${\bf U}^\star_i$ at each time step. Once computed, the rest of the method follows the usual lines. The reader is referred e.g.\ to \cite{touboulJoCP2020} for practical details. Four remarks can be made concerning the implementation and the properties of the ESIM:
\begin{itemize}
\item The transfer matrices ${\bf D}_k$ in \eqref{Esim4} depend on time. Consequently, the extrapolation matrices must be recalculated at each time step. This is an important difference with the cases previously dealt with by ESIM, where the extrapolation matrices were calculated once during  a preprocessing step; 
\item The calculation of ${\bf D}_k$ is tedious, especially for large values of $k$. A naive idea to simplify this calculation is to {\it freeze} the jump conditions at the instant under consideration, so that $\partial_t {\bf B}={\bf 0}$ in \eqref{Esim2}, and ${\bf D}_k(t)$ depends only on ${\bf B}$  and ${\bf E}$ and not on their successive derivatives. This is a bad idea: the time evolution of the modulated interface conditions is then poorly described, leading to inaccurate numerical results;
\item The problem under study has similarities with other time-varying interface problems. But here we are in a {\it kinematic} case, where the jump conditions are known and imposed, unlike {\it dynamic} cases such as \cite{assierJFM2014}, where the jump conditions are unknown and depend on the solution in a non-trivial way;
\item The optimal choice of the derivation order $k$ is important. In the case of static jump conditions (see e.g.\ \cite{lombardSJSC2003}), a numerical analysis of the ESIM has been carried out, and we will reuse these results here. In the limit case where the interface disappears ($\mathscr{C}_0= 0$, $\mathscr{M}_0=0$,  $\mathscr{Q}_{C_0}=0$ and  $\mathscr{Q}_{M_0}=0$), ${\bf U}_j^\star={\bf U}_j^n$ is needed to recover the scheme in homogeneous medium. This consistency property is obtained if $k$ is odd and $k \geq 2s-1$, where $s$ is the width of the stencil. Additionally, the local truncation error of a $r$-th order accurate scheme is maintained at the interface if $k\geq r$. Putting together these two conditions yields
\begin{equation}
k=\max\left(r,2s-1\right)\,\mbox{and $k$ odd}.
\label{Esim-AN}
\end{equation}

Here $r=4$ and $s=2$, so that $k=5$ is the optimal value, which will be used in the numerical experiments. We will examine numerically the convergence properties using this value in Section \ref{ValidNum} since we have no proof that it still holds in the current modulated case.
\end{itemize}


\subsection{Numerical set-up}\label{SecSimusSetup}

We consider a medium of length $L=400$ m,  discretized on $N_x=400$ grid nodes. The bulk parameters are constant and correspond to those of a bar of Plexiglass: $\rho=1200$ kg/m$^3$ and $c=2800$ m/s. The CFL number in \eqref{CFL} is $\zeta=0.95$.

A modulated interface lies at $x_0=200$ m. The values of the interface parameters used further have no physical origin, and they are only chosen to illustrate clearly the expected phenomena. Typically, modulation magnitudes of the order of 0.75 or 0.9 are used, whereas real systems lead to magnitudes of the order of 0.2 \cite{mallejacPRA2023}.

\begin{figure}[htbp]
\begin{center}
\begin{tabular}{cc}
(a) & (b) \\
\hspace{-0.3cm}
\includegraphics[width=0.48\linewidth]{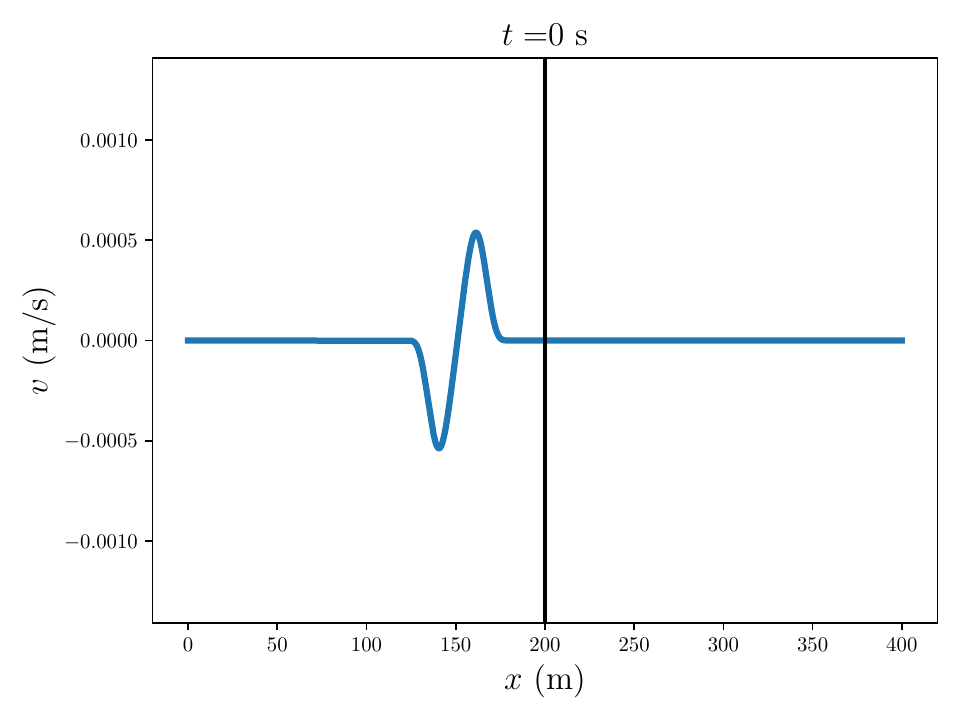} & 
\hspace{-0.3cm}
\includegraphics[width=0.48\linewidth]{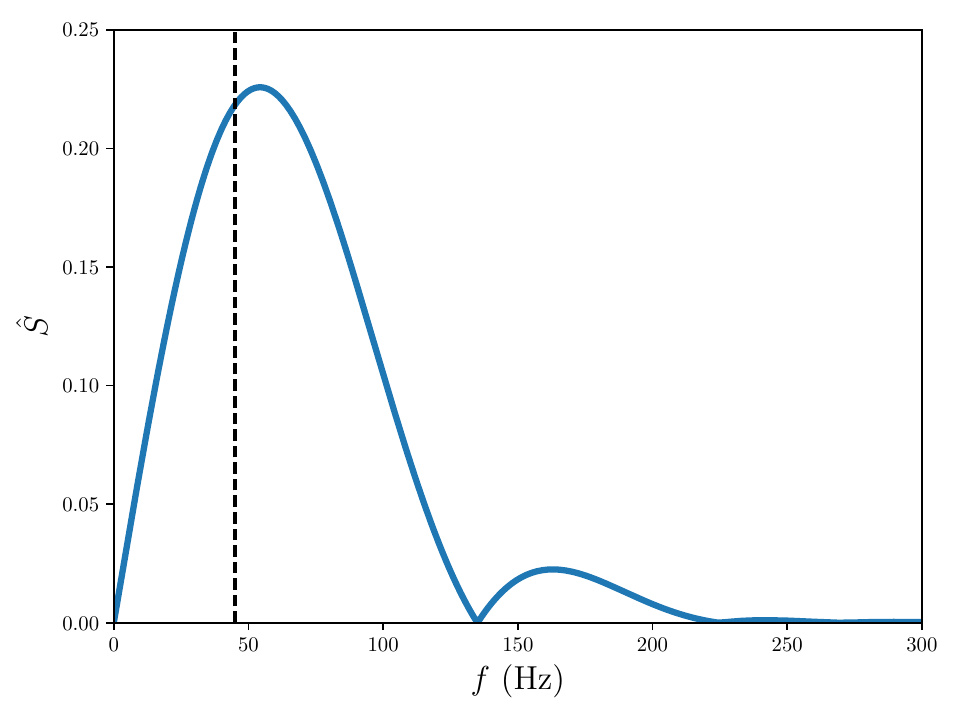}
\end{tabular}
\end{center}
\vspace{-0.5cm}
\caption{\label{FigSource}Forcing of simulations. (a) Cauchy problem \eqref{Cauchy} with $f_c=45$ Hz. (b) Spectrum of the source \eqref{Source}; the vertical dotted line denotes the central frequency $f_c$.}
\end{figure}

Two types of forcing are considered: i) a Cauchy problem with an initial data associated with a right-going wave
\begin{equation}
{\bf U}(x,0)
=
\left(
\begin{array}{c}
\rho\\
-1/c
\end{array}
\right)
\,
S\left(t_0-\frac{x}{c}\right),
\label{Cauchy}
\end{equation}
where $t_0$ allows to position the initial wave on the left of the interface, or ii) zero initial conditions with a time-dependent forcing at a Dirac source point 
\begin{equation}
F(x,t)=S(t)\,\delta(x-x_s).
\label{Dirac} 
\end{equation}
Both forcings involve the source function $S$, chosen as a combination of truncated sinusoids:
\begin{equation}
S(\xi)=\left\{
\begin{aligned}
& \sum_{m=1}^{4}a_m\sin(b_m\,\omega_c\,\xi) && \text{if } -0<\xi <1/f_c \\
& 0 && \text{otherwise,}
\end{aligned}\right.
\label{Source}
\end{equation}
where $f_c=\omega_c/(2\pi)$ is the central frequency, $b_m=2^{m-1}$, the coefficients $a_m$ are $a_1=1$, $a_2=-{21}/{32}$, $a_3 = {63}/{768}$, and $a_4=-{1}/{512}$. This source entails $C^6([0,+\infty[)$ smoothness. Figure \ref{FigSource} illustrates a Cauchy problem (a) and the spectral content of $S$ (b).

The time dependence of $\mathscr{C}$ and $\mathscr{M}$ writes
\begin{equation}
\begin{array}{rcl}
\ds \mathscr{C}(t)&=&\mathscr{C}_0\left(1+\varepsilon_C\,\phi_C(t)\right),\\ [8pt]
\ds \mathscr{M}(t)&=&\mathscr{M}_0\left(1+\varepsilon_M\,\phi_M(t)\right),\\ [8pt]
\ds \mathscr{Q}_C(t)&=&\mathscr{Q}_{C_0}\left(1+\varepsilon_{Q_C}\,\phi_{Q_C}(t)\right),\\ [8pt]
\ds \mathscr{Q}_M(t)&=&\mathscr{Q}_{M_0}\left(1+\varepsilon_{Q_M}\,\phi_{Q_M}(t)\right),
\end{array}
\label{KM-T}
\end{equation}
where $\mathscr{C}_0\geq0$, $\mathscr{M}_0\geq0$, $\mathscr{Q}_{C_0}\geq0$, $\mathscr{Q}_{M_0}\geq0$, and $\varepsilon_X\,\phi_X >-1$, with $X=C,M,Q_C,Q_M$. The functions $\phi_C$, $\phi_M$, $\phi_{Q_C}$ and $\phi_{Q_M}$ may differ from each other. If $\phi_X=0$ or $\varepsilon_X=0$, then $\mathscr{C}$, $\mathscr{M}$, $\mathscr{Q}_C$ and $\mathscr{Q}_M$ do not depend on $t$, and one recovers the static conditions of imperfect contact investigated in \cite{lombardSJSC2003}. In addition to the sinusoidal modulation \eqref{ModulSinus}, two modulation functions are introduced for the numerical simulations: a quasi-periodic modulation and a rectangular one. It follows
\begin{equation}
\phi_{X}(t)=\left\{
\begin{array}{ll}
\ds \sin(\Omega t), & (\mbox{sinusoidal}),\\ [8pt]
\ds \sin(\Omega t)+\sin\left(\sqrt{2}\,\Omega t\right), & (\mbox{quasi-periodic}),\\ [8pt]
\ds (-1)^n,\quad \mbox{with } n=\left \lfloor\left( \frac{\Omega \,t}{2\pi} \mod 1\right) - \nu  \right\rfloor \quad & (\mbox{rectangular}),
\end{array}
\right.
\label{FuncModul}
\end{equation}
where $\lfloor g \rfloor$ denotes the floor function of $g$, $0>\nu>1$, and $\Omega=2\pi\,f_m$, where $f_m$ is the \textit{modulation frequency}. In this work, we assume $\phi_C=\phi_M=\phi_{Q_C}=\phi_{Q_M}$ (which implies that the modulation frequencies are the same), but this is only for simplicity. 
The numerical scheme presented above does not depend on this assumption. Therefore, this scheme and our implementation of it remain valid for functions that differ from one another.


\subsection{Validation}\label{ValidNum}

In this part, a validation of the numerical scheme is proposed.
\begin{figure}[htpb]
\begin{center}
\begin{tabular}{cc}
(a) & (b) \\
\hspace{-0.3cm}
\includegraphics[width=0.39\linewidth]{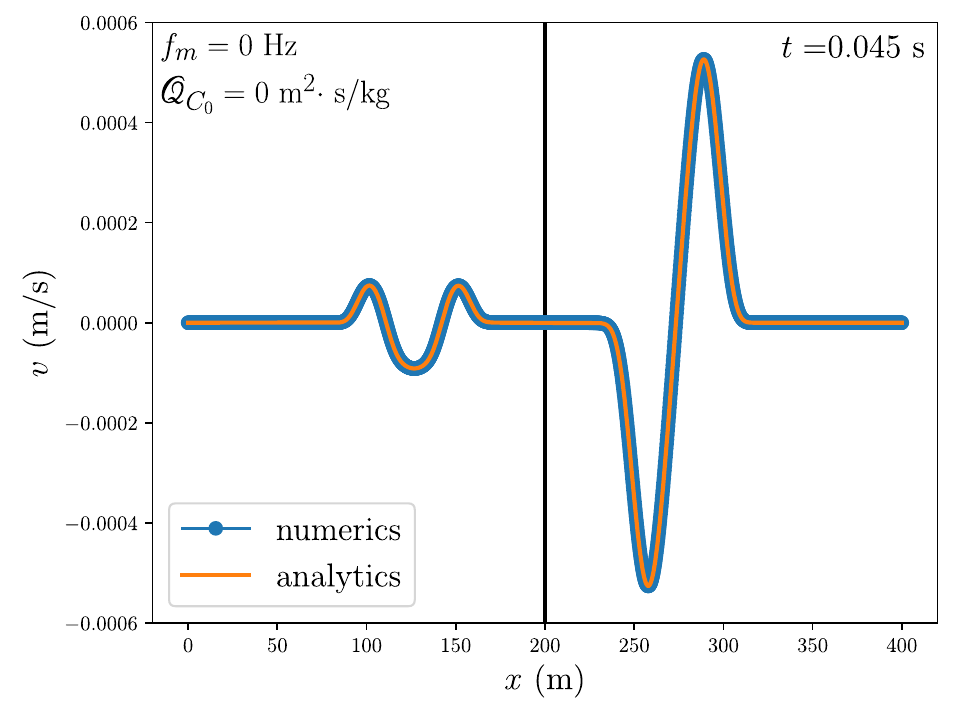} & 
\hspace{-0.3cm}
\includegraphics[width=0.39\linewidth]{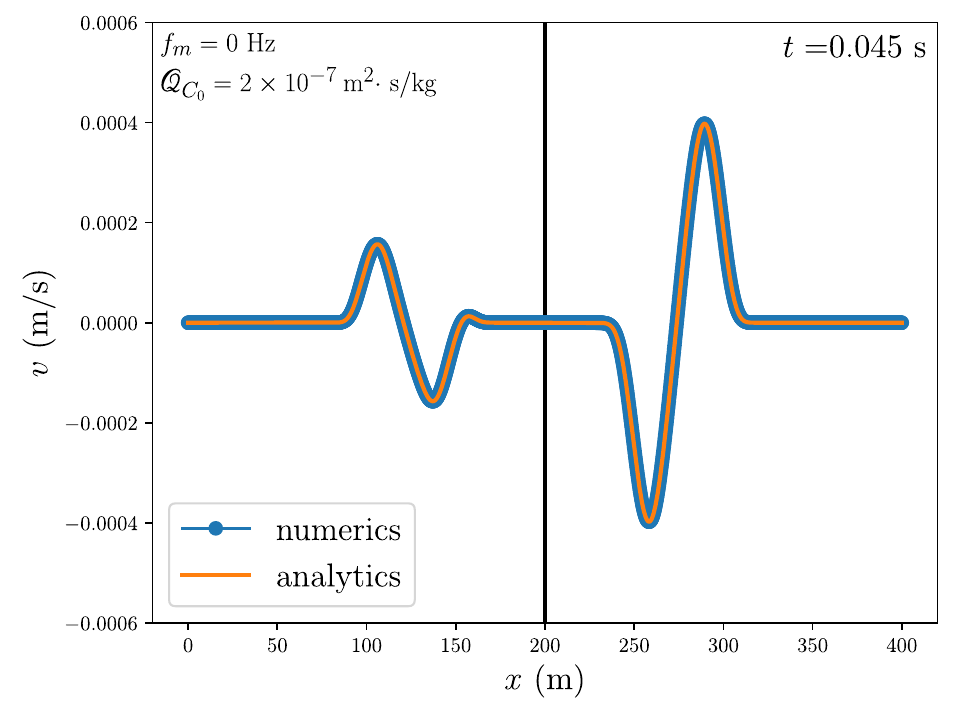}\\
(c) & (d) \\
\hspace{-0.3cm}
\includegraphics[width=0.39\linewidth]{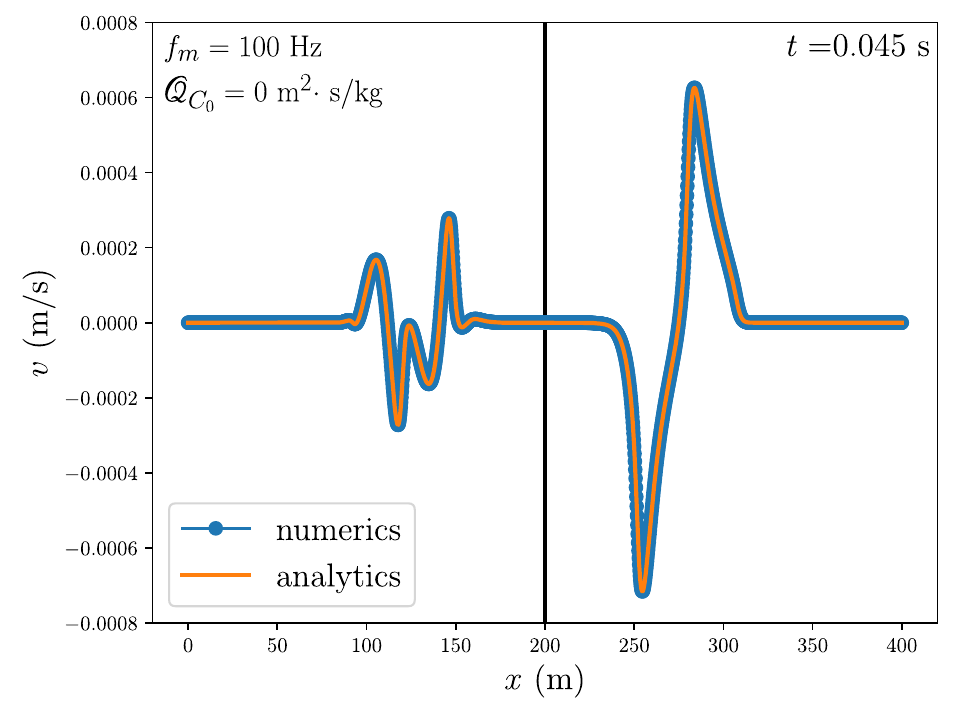} & 
\hspace{-0.3cm}
\includegraphics[width=0.39\linewidth]{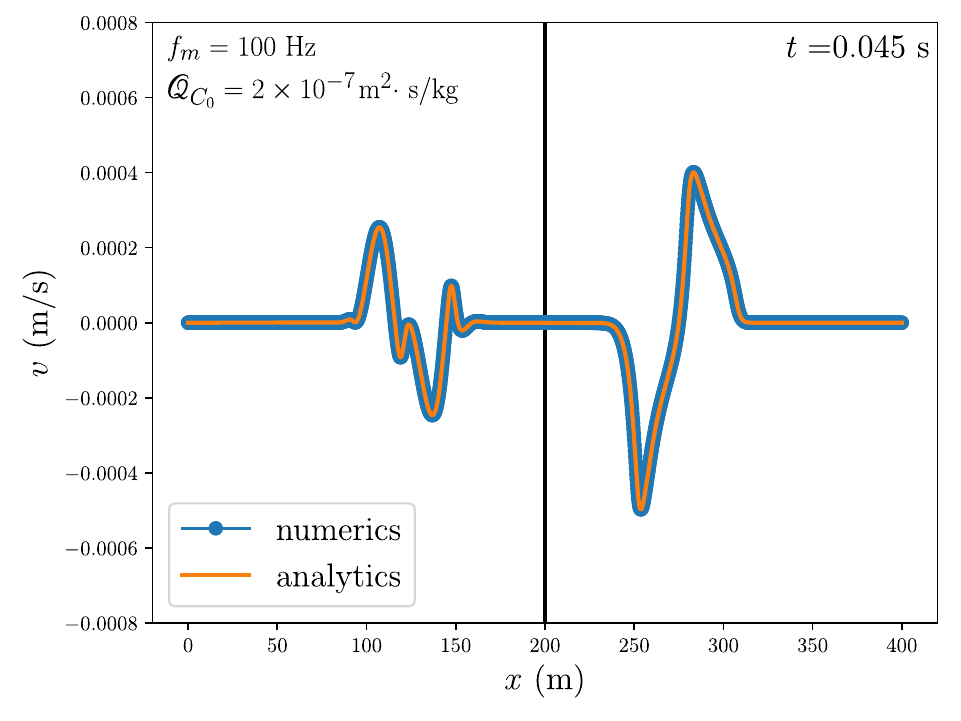}\\
(e) & (f) \\
\hspace{-0.3cm}
\includegraphics[width=0.39\linewidth]{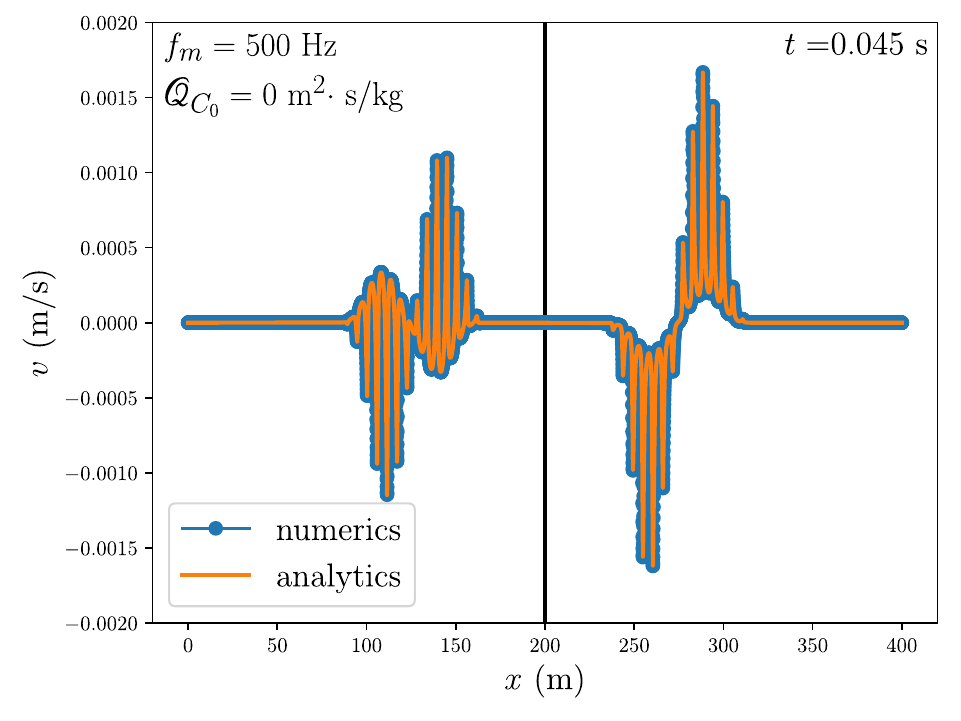} & 
\hspace{-0.3cm}
\includegraphics[width=0.39\linewidth]{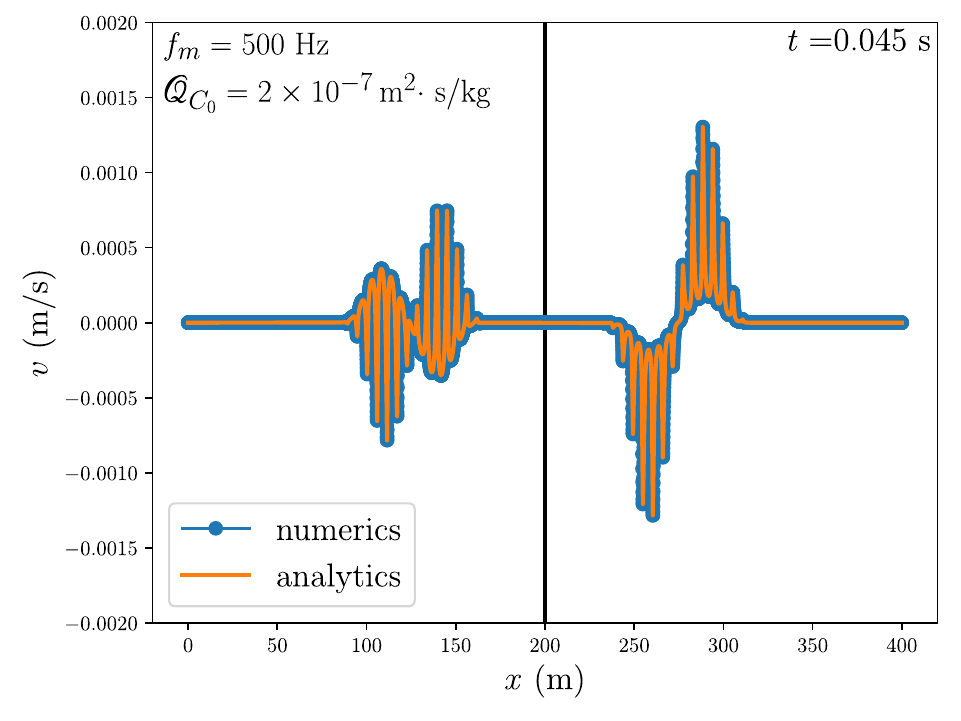}\\
(g) & (h) \\
\hspace{-0.3cm}
\includegraphics[width=0.39\linewidth]{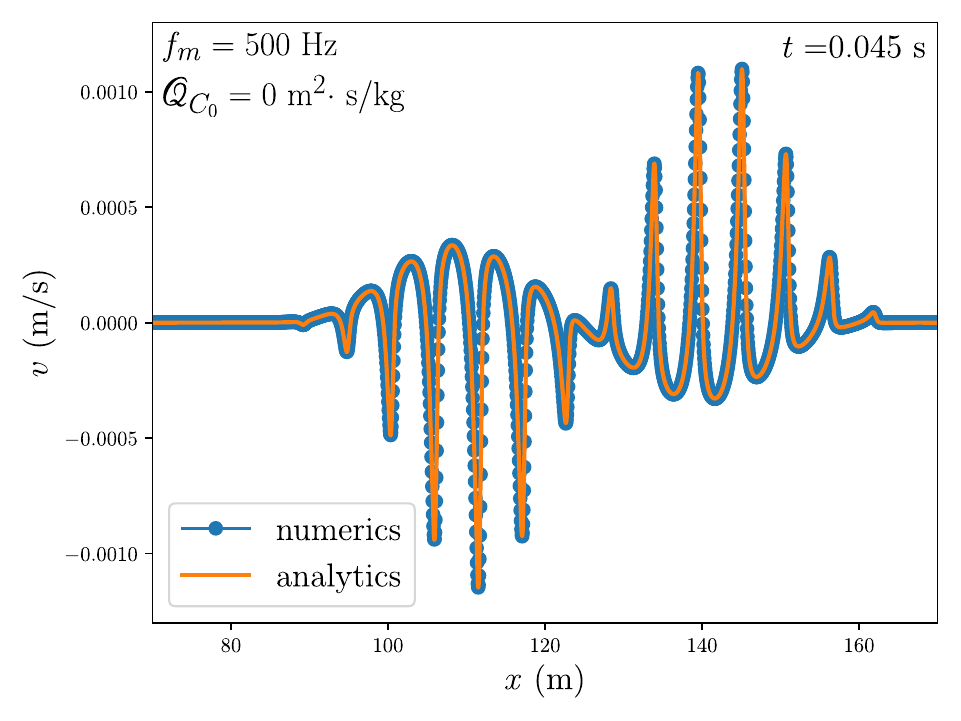} & 
\hspace{-0.3cm}
\includegraphics[width=0.39\linewidth]{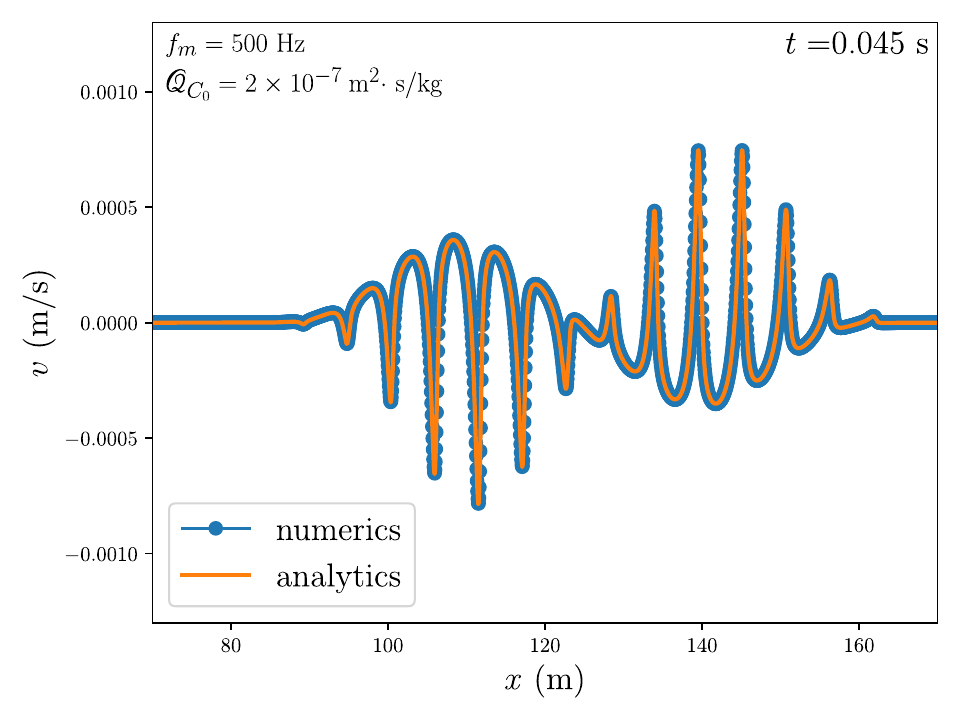}\\
\end{tabular}
\end{center}
\vspace{-0.5cm}
\caption{\label{FigModulK}Scattering by an interface with modulated stiffness, for various frequencies of sinusoidal modulation $f_m$. Snapshot of the solution at $t=0.045$ s with (a-b): no modulation, (c-d): $f_m=100$ Hz, (e-f) $f_m=500$ Hz, and (g-h) is a zoom on the reflected wave when $f_m=500$ Hz. Left: cases without dissipation. Right: $\mathscr{Q}_{C_0}=2\,\times 10^{-7}$ m$^2\cdot$s/kg and $\mathscr{Q}_{M_0}=0$.}
\end{figure}
A pulse is initially located on the left of the interface \eqref{Cauchy}. The interface parameters are $\mathscr{K}_0=2.45$ GPa/m, $\mathscr{M}_0=0$. The central frequency if $f_c=30$ Hz. A semi-analytic solution of \eqref{BVP} can be obtained by the method of characteristics when $\mathscr{M}_0=0$ or $\mathscr{C}_0=0$. Technical details are given in Appendix \ref{AppAnalytic} in the case where only the stiffness is modulated.  

Figure \ref{FigModulK}-(a) and (b) illustrates the case without modulation: $\varepsilon_C=0$ . The other subfigures are computed using $\varepsilon_C=\varepsilon_{Q_C}= 0.9$ and various values of the modulation frequency $f_m=\frac{\Omega}{2\pi}$. Cases without dissipation in the interface conditions ($\mathscr{Q}_{C_0}=0$ and $\mathscr{Q}_{M_0}=0$) are presented on the left column and cases with modulated dissipation ($\mathscr{Q}_{C_0}=2\,\times 10^{-7}$ m$^2\cdot$s/kg and $\mathscr{Q}_{M_0}=0$) are on the right side. The value of $\mathscr{Q}_{C_0}$ have been chosen to have an effect of the same order as the compliance effect. One observes an excellent agreement between the exact and numerical values, which validates the numerical method detailed in Section \ref{SecNum}. There is also an enrichment of peaks in the scattered waves when $|f_m-f_c|$ increases.

The particular case where only $\mathscr{Q}_{C}$ is non-zero and modulated is presented on Figure \ref{FigModulQ_AN}. Here again, the agreement between the analytical solution and the result of the numerical model is excellent. On the left side, the envelope curves has been added to highlight the Proposition \ref{PropEnvelop}. In this case with a sinusoidal modulation, the extreme values of $\mathscr{Q}_{C}(t)$, used to determine the dotted curves, are equal to $\mathscr{Q}_{C_{0}}(1\pm\varepsilon_{Q_C})$.
\begin{figure}[htpb]
\begin{center}
\begin{tabular}{cc}
(a) & (b) \\
\hspace{-0.3cm}
\includegraphics[width=0.48\linewidth]{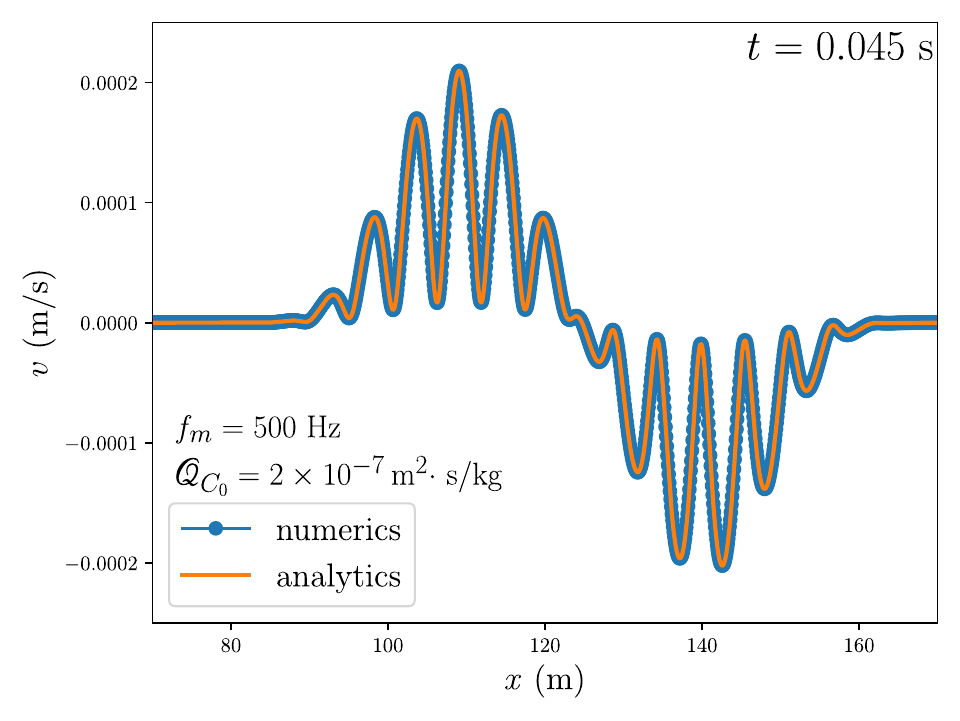} & 
\hspace{-0.3cm}
\includegraphics[width=0.48\linewidth]{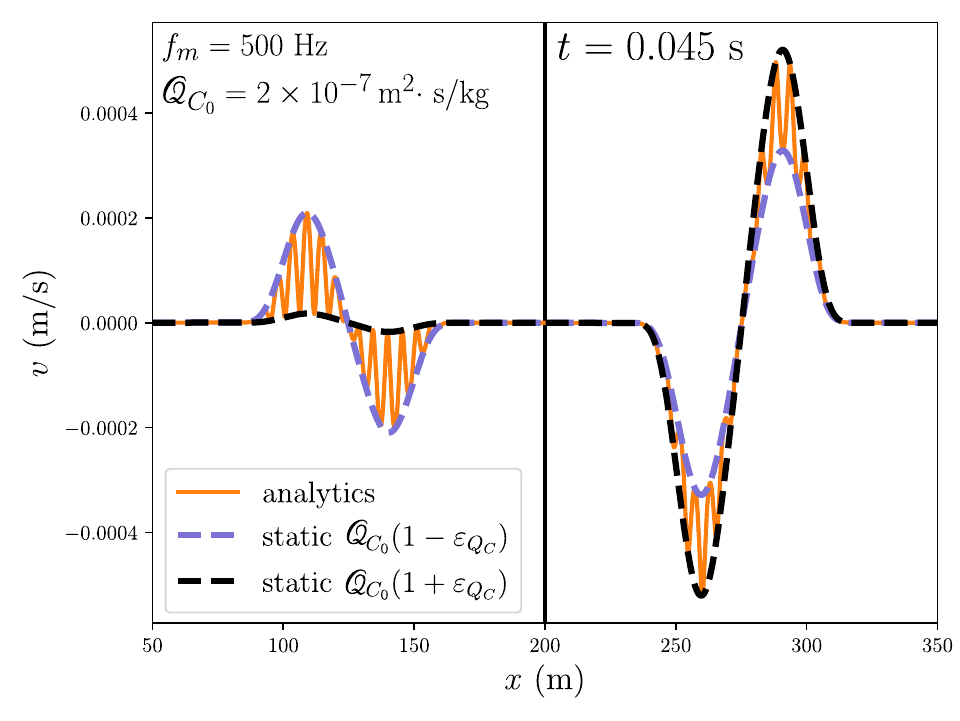}\\
\end{tabular}
\end{center}
\vspace{-0.5cm}
\caption{\label{FigModulQ_AN}Scattering by an interface with modulated dissipation ($\mathscr{Q}_{C_0}=2\times10^{-7}$ m$^2\cdot$s/kg and the other parameters are zero) at the frequency $f_m=500$ Hz at $t=0.045$ s. (a): Snapshot of the solution with a zoom on the reflected wave, (b): analytic solution with envelope curves (dashed lines).}
\end{figure}

In the rest of this work, we only consider cases with $\mathscr{Q}_{C_0}=0$ and $\mathscr{Q}_{M_0}=0$.
\begin{figure}[htpb]
\begin{center}
\begin{tabular}{cc}
(a) & (b) \\
\hspace{-0.3cm}
\includegraphics[width=0.48\linewidth]{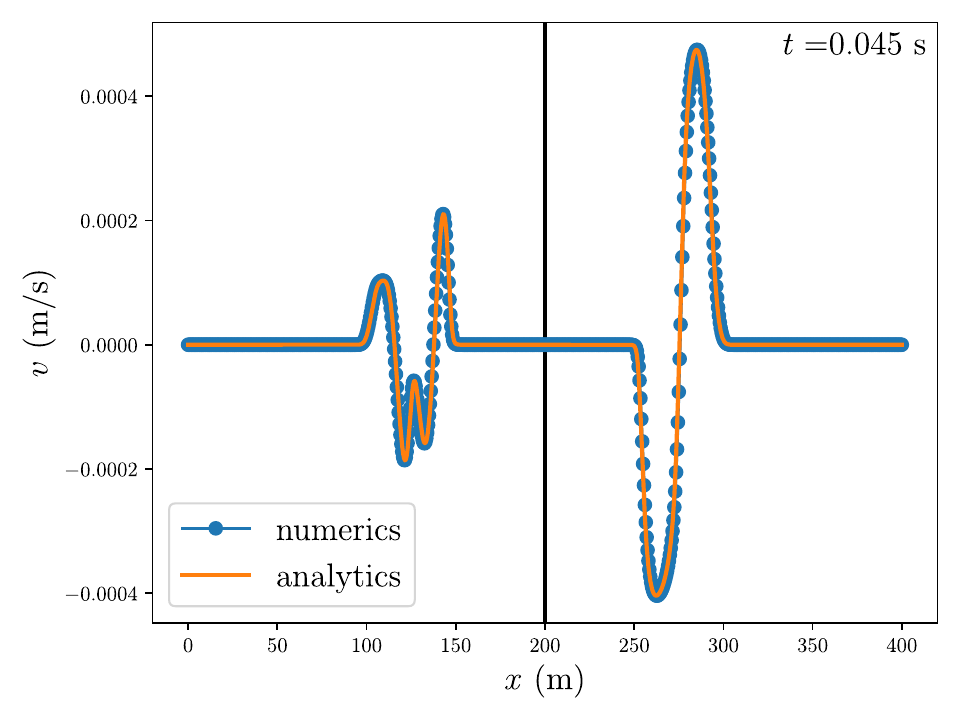} & 
\hspace{-0.3cm}
\includegraphics[width=0.48\linewidth]{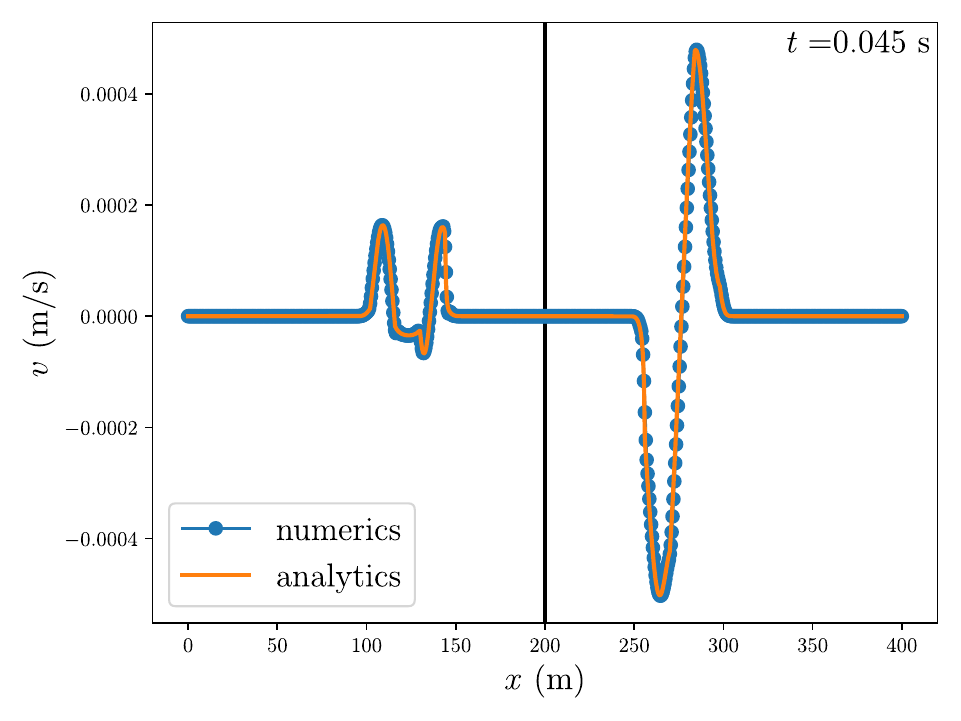}\\
\end{tabular}
\end{center}
\vspace{-0.5cm}
\caption{\label{FigModulK-NH}Scattering by an interface with modulated stiffness, for non-harmonic modulation at the frequency $f_m=100$ Hz at $t=0.045$ s. (a): quasi-periodic modulation, (b): square modulation.}
\end{figure}

Figure \ref{FigModulK-NH} illustrates the two other modulation types given by \eqref{FuncModul} at the frequency $f_m=100$ Hz. Figure \ref{FigModulK-NH}-(a) displays the results obtained in the quasi-periodic modulation, while Figure \ref{FigModulK-NH}-(b) investigates the rectangular modulation ($\nu=0.65$ in \eqref{FuncModul}). Here again, there is an excellent agreement between the exact and numerical solution. In the case of a square modulation, the terms corresponding to the derivatives of the modulation in \eqref{Esim4} are set to zero.

Measures of convergence have been performed to verify the conclusions obtained in \cite{lombardSJSC2003}. The Figure \ref{figConvergence} presents the evolution of the error $\varepsilon_v$ committed on the velocity field $v$ at $t=0.045$ s, and defined by:
\begin{equation}
    \varepsilon_v=\sqrt{\Delta x\sum_{i=0}^{N_x}\, \left(v(x_i,t)-v_i^n\right)^2}.
\end{equation}
On the left side, convergence measurement are shown for a sinusoidal modulation of the compliance $\mathscr{C}$ using $\mathscr{K}_0=2.45$ GPa/m and $\mathscr{M}=0$. On the right side, the analog case is presented (\textit{i.e.\ } $\mathscr{M}_0=2\,\times 10^4$ kg/m$^2$ and $\mathscr{C}=0$). The other numerical values of the parameters are $\varepsilon_{C,M}=0.75$, $f_c=45$ Hz and two different frequencies of a sine modulation are tested: $f_m=100$ Hz (plain blue line) and $f_m=10$ Hz (orange dashed line). 
\begin{figure}[htpb]
\begin{center}
\begin{tabular}{cc}
(a) & (b) \\
\hspace{-0.3cm}
\includegraphics[width=0.48\linewidth]{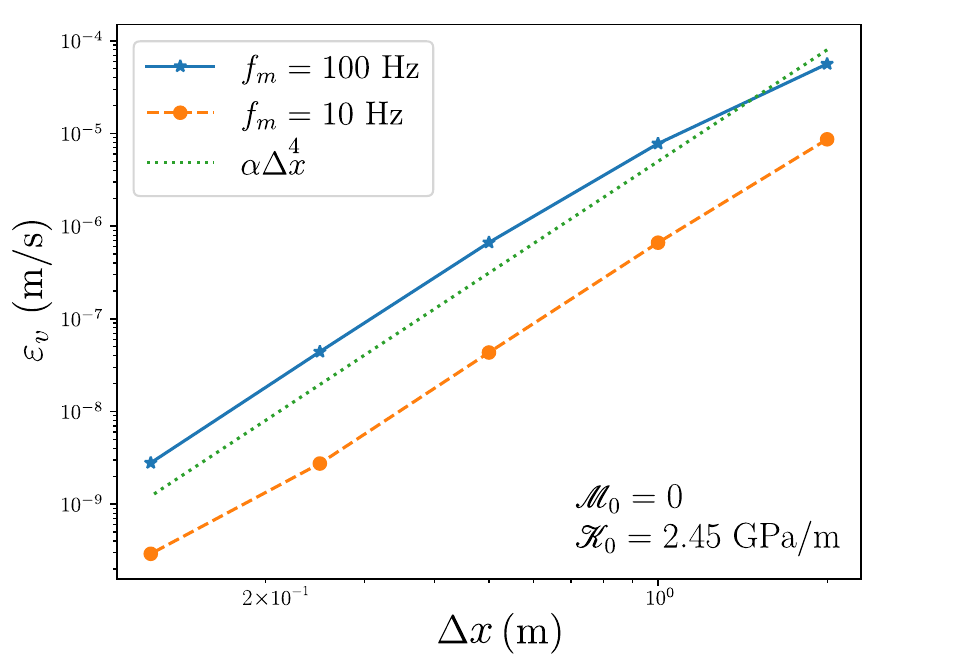} & 
\hspace{-0.3cm}
\includegraphics[width=0.48\linewidth]{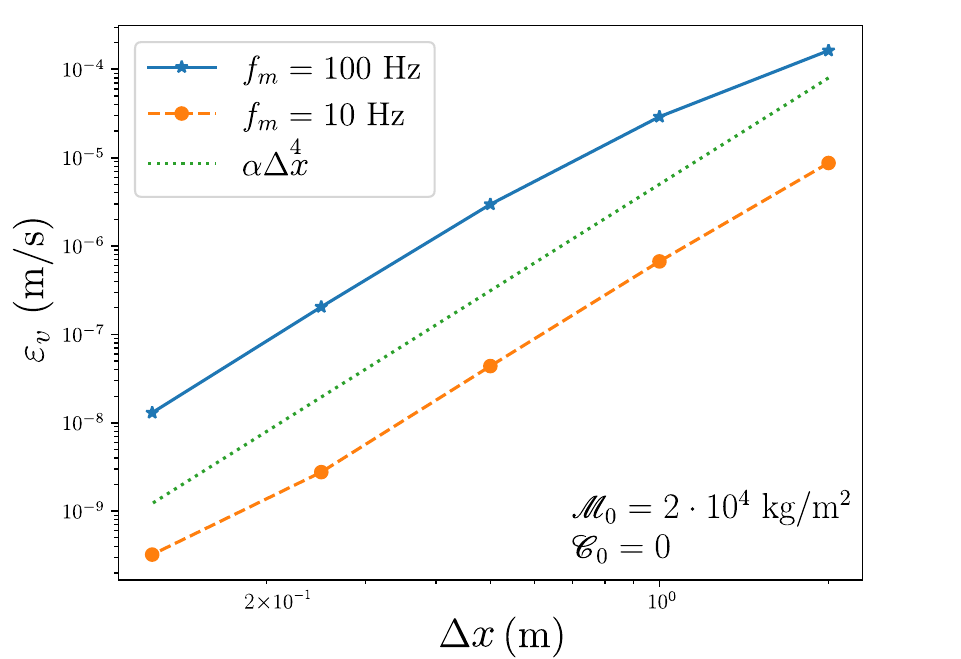}\\
\end{tabular}
\end{center}
\vspace{-0.5cm}
\caption{\label{figConvergence}Convergence measurements for $f_m =100$ Hz in blue line and $f_m=10$ Hz in orange dashed line in the case where (a) $\mathscr{C}(t)\neq 0$ and $\mathscr{M}(t)=0$, and (b)  $\mathscr{C}(t)=0$ and $\mathscr{M}(t)\neq 0$. The green dotted line shows a slope of order 4 (in log-log scale).}
\end{figure}
These measurements show a convergence of order $4$ for both modulations and both frequencies. The two curves corresponding to $f_m=100$ Hz present a slope discontinuity for the higher value of $\Delta x$ but we only consider the limit for small values of $\Delta x$. Moreover, due to the generation of harmonics with larger frequencies, the value of the error is more significant for a modulation frequency at $f_m=100$ Hz than at $f_m=10$ Hz.


\section{Numerical experiments}\label{SecSimus}

Here we illustrate the various properties analyzed in Section \ref{SecContinu}. We conclude with an illustration of non-reciprocity. 

\subsection{Modulation of energy}\label{SecSimusNRJ}

\begin{figure}[htpb]
\begin{center}
\begin{tabular}{cc}
(a) & (b) \\
\includegraphics[width=0.48\linewidth]{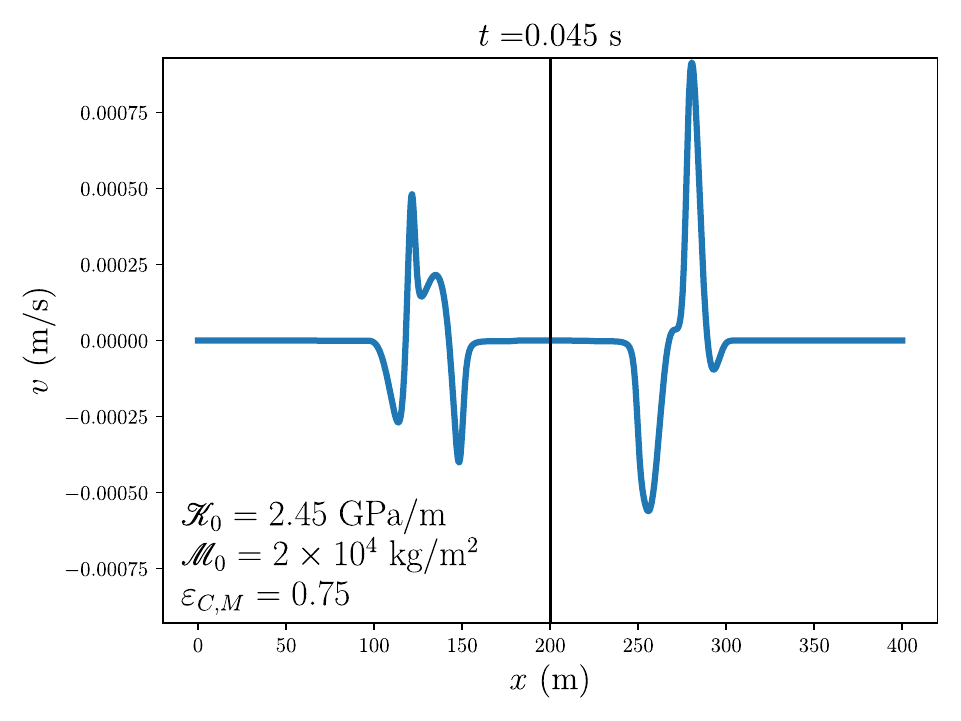} & 
\includegraphics[width=0.48\linewidth]{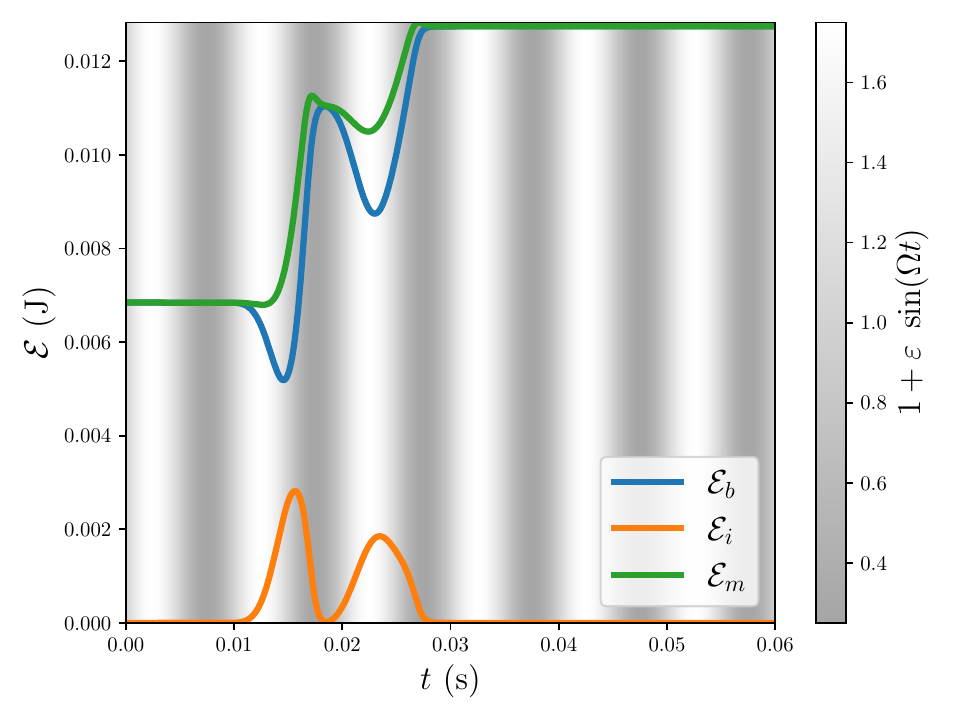}\\ 
(c) & (d) \\
\includegraphics[width=0.48\linewidth]{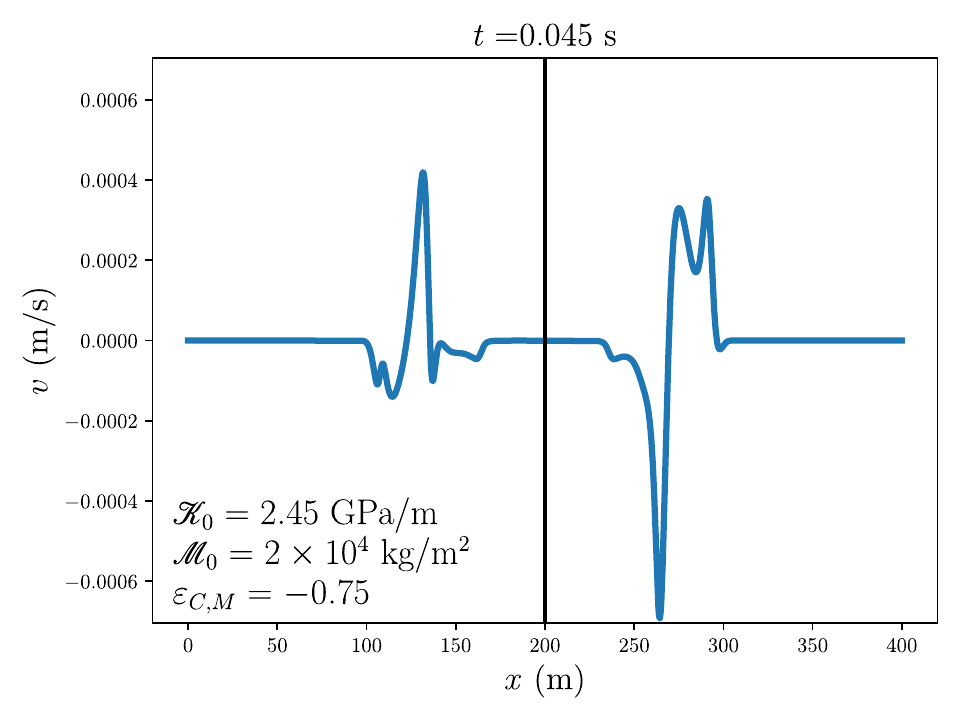} &
\includegraphics[width=0.48\linewidth]{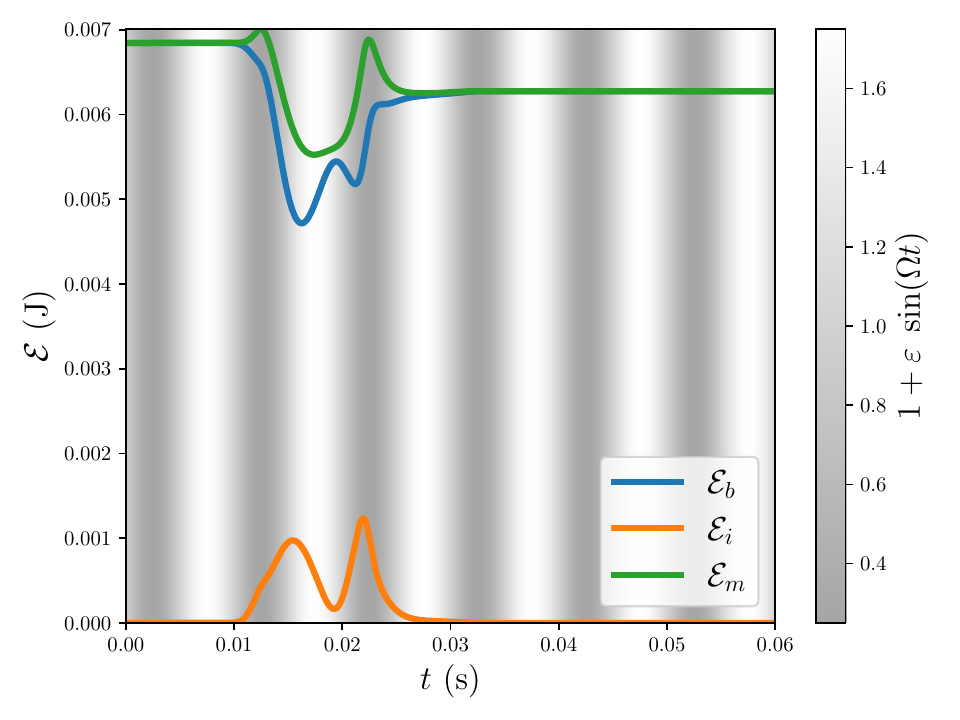}\\
(e) & (f) \\
\includegraphics[width=0.48\linewidth]{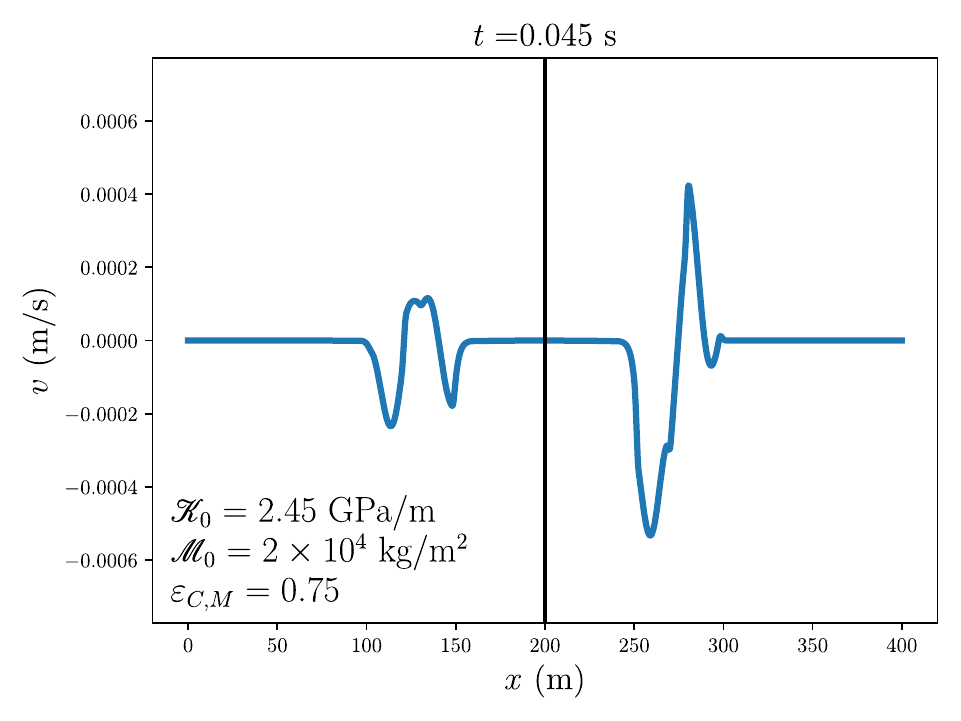} & 
\includegraphics[width=0.48\linewidth]{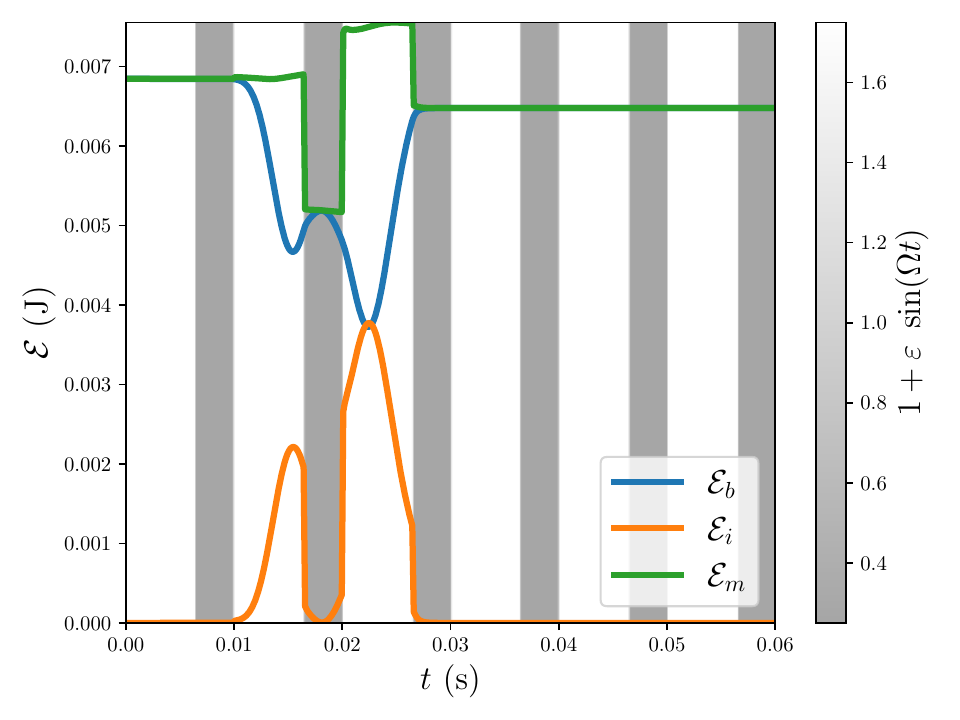}\\ 
\end{tabular}
\end{center}
\vspace{-0.5cm}
\caption{\label{FigNinterfAmpli}Interaction of a pulse centered on $f_c=45$ Hz with a modulated interface. Left row: snapshot of the scattered velocity $v$ at $t=0.045$ s. Right row: time evolution of the energy, where the gray shaded areas denote the evolution of the interface parameters with modulation frequency $f_m=100$ Hz. Sinusoidal modulation using (a-b): $\varepsilon_{C,M}=0.75$  ; (c-d): $\varepsilon_{C,M}=-0.75$ and a rectangular modulation (e-f) with $\varepsilon_{C,M}=0.75$ and $\nu=0.65$.}
\end{figure}

A pulse of the form \eqref{Cauchy} with central frequency $f_c=45$ Hz is initially located on the left of the interface. The interface parameters are $\mathscr{K}_0=2.45$ GPa/m and $\mathscr{M}_0=2\,\times 10^4$ kg/m$^2$ with a modulation frequency $f_m=100$ Hz. Figure \ref{FigNinterfAmpli} investigates the scattered waves for different modulations. The left column shows a snapshot of the velocity $v$. The right column shows the temporal evolution of the energies \eqref{NRJv}-\eqref{NRJi} as a function of time. The background gray scale represents the time evolution of the interface parameters. 

The total energy ${\mathcal E}_m$ is constant until $t\approx 0.01$ s, when the incident wave impacts the interface. The initially zero interface energy ${\mathcal E}_i$ then increases and remains non-zero as long as the wave/interface interaction lasts. As a consequence ${\mathcal E}_m$ also varies. After the wave has passed through the interface, the interface energy is radiated and reconverted into bulk energy, which is then conserved. Depending on the sign of $\varepsilon_{C,M}$, different evolutions of energy are observed;
\begin{itemize}
\item In the first case $\varepsilon_{C,M}=0.75$ (a-b), the final total energy value ($=0.0013$) has almost doubled compared with the initial energy ($=0.007$);
\item In the second case, the modulation is inverted by choosing $\varepsilon_{C,M}=-0.75$  (c-d). The total amount of energy decreases from $=0.007$ to $=0.006$;
\item In the third case, a rectangular modulation is applied with $\varepsilon_{C,M}=0.75$ and $\nu=0.65$ (e-f). The interface energy (and consequently the total energy) is then discontinuous.
\end{itemize}
In the first case, the interface is reached when the value of $\varepsilon_{C,M} \sin(\Omega t)$ is at its maximum value, whereas these values are negative in the second case. In conclusion, the work used to modulate the interface parameters can logically increase or decrease the total energy, as expected and predicted in our remarks to Proposition \ref{PropNRJ}.

\begin{figure}
    \centering
    \begin{tabular}{cc}
    (a) & (b) \\
    \includegraphics[width=0.48\linewidth]{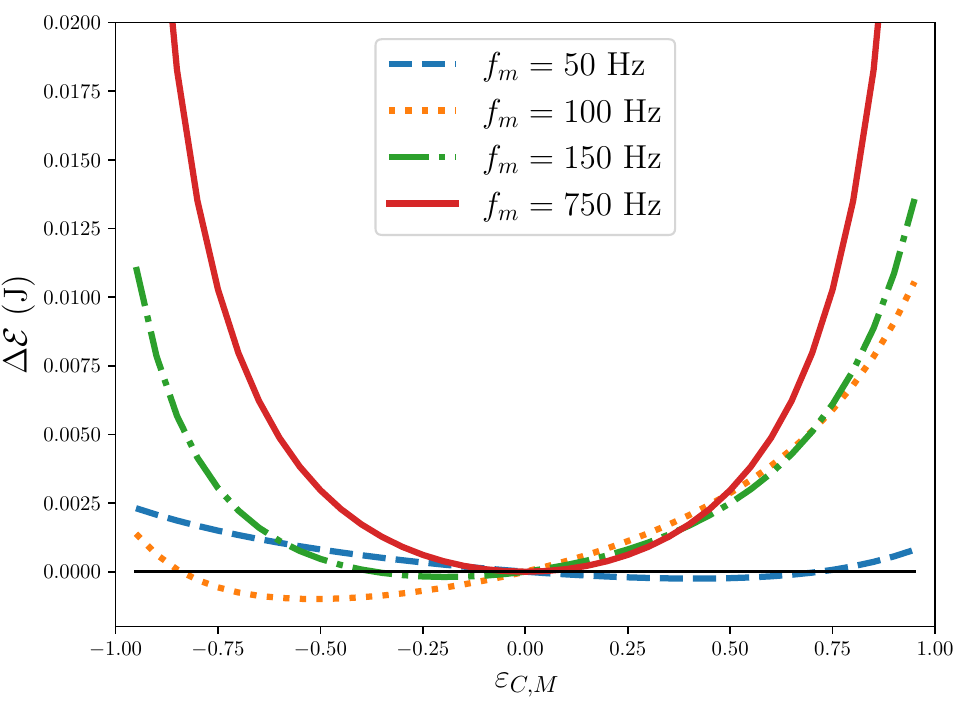} & \includegraphics[width=0.48\linewidth]{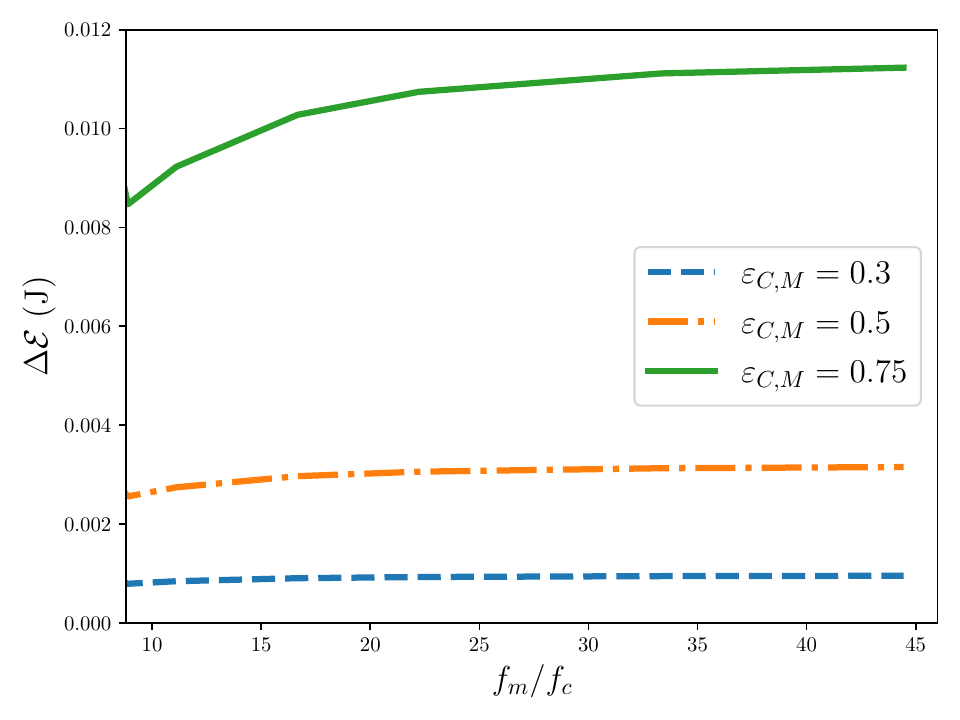}\\
    \end{tabular}
    \caption{Influence of the modulation parameters on the variation of energy after crossing the interface. (a) influence on the amplitude of the modulation $\varepsilon_{C,M}$ and (b) influence of the ratio $f_m/f_c$.} 
    \label{fig:deltaener}
\end{figure}

Figure \ref{fig:deltaener}(a) presents the variation of the energy $\Delta{\mathcal E}=\mathcal{E}_m(t)-\mathcal{E}_m(0)$ between $t=0$~s and $t=0.04$~s (time after which the wavelet has completely crossed the interface) for different values of the amplitude of modulation $\varepsilon_{C,M}\in [-0.95,0.95]$, and for four different modulation frequencies ($f_m=50$ Hz, $f_m=100$ Hz, $f_m=150$ Hz and $f_m=750$ Hz). This figure shows that the gain of energy is not linear with respect to $\varepsilon_{C,M}$. Moreover, for low values of $f_m$, there exists a range of modulation for which energy is lost, but this range depends also on the choice of the modulation frequency. For the highest frequencies, the variation of energy is always postive and is symmetric because the phase shift due to the modification of sign of $\varepsilon_{C,M}$ does not have a significant effect. Figure \ref{fig:deltaener}(b) shows the variation of the energy $\Delta{\mathcal E}$ as a function of $\frac{f_m}{f_c}$ for three values of $\varepsilon_{C,M}$. At lower frequencies (until $f_m=400$ Hz, \textit{i.e.} $\frac{f_m}{f_c}\approx 9$, not presented here), the variation of energy $\Delta{\mathcal E}$ depends strongly on the modulation frequency $f_m$. This is because, at such low modulation frequency, the state of the interface when the pulse crosses it varies a lot with $f_m$. On the contrary, at higher frequencies (presented here), $\Delta{\mathcal E}$ appears to tend monotonously to a limit value, that could be interpreted as an effective (averaged in time) behaviour.


\subsection{Generation of harmonics}\label{SecSimusHarmonics}

Here we illustrate the calculations carried out in Section \ref{SecContinuHarm}. The reflection and transmission coefficients $R_k$ and $T_k$ in \eqref{AnsatzRT} are determined by solving \eqref{RkTk} with $N=12$ Fourier modes. This number has been chosen experimentally to ensure the convergence. The interface parameters are $\mathscr{K}_0=2.45$ GPa/m and $\mathscr{M}_0=2\,\times 10^4$ kg/m$^2$, with modulation amplitudes $\varepsilon_C=\varepsilon_M=0.75$. 

\begin{figure}[htbp]
\begin{center}
\begin{tabular}{cc}
(a) & (b) \\
\hspace{-0.3cm}
\includegraphics[width=0.48\linewidth]{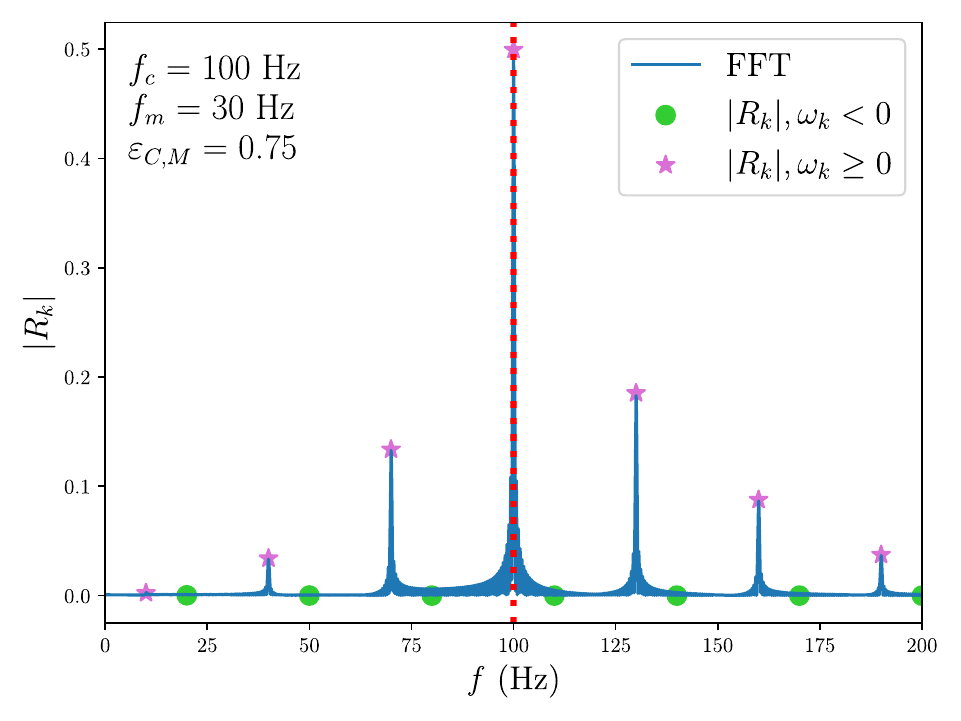} & 
\hspace{-0.3cm}
\includegraphics[width=0.48\linewidth]{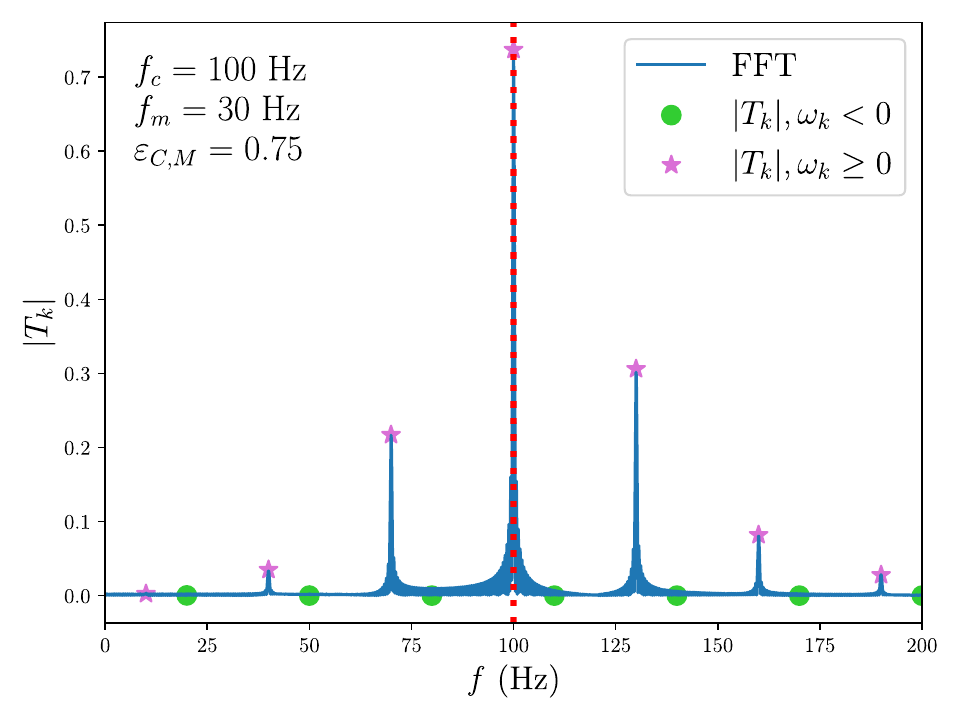} \\
(c) & (d) \\
\hspace{-0.3cm}
\includegraphics[width=0.48\linewidth]{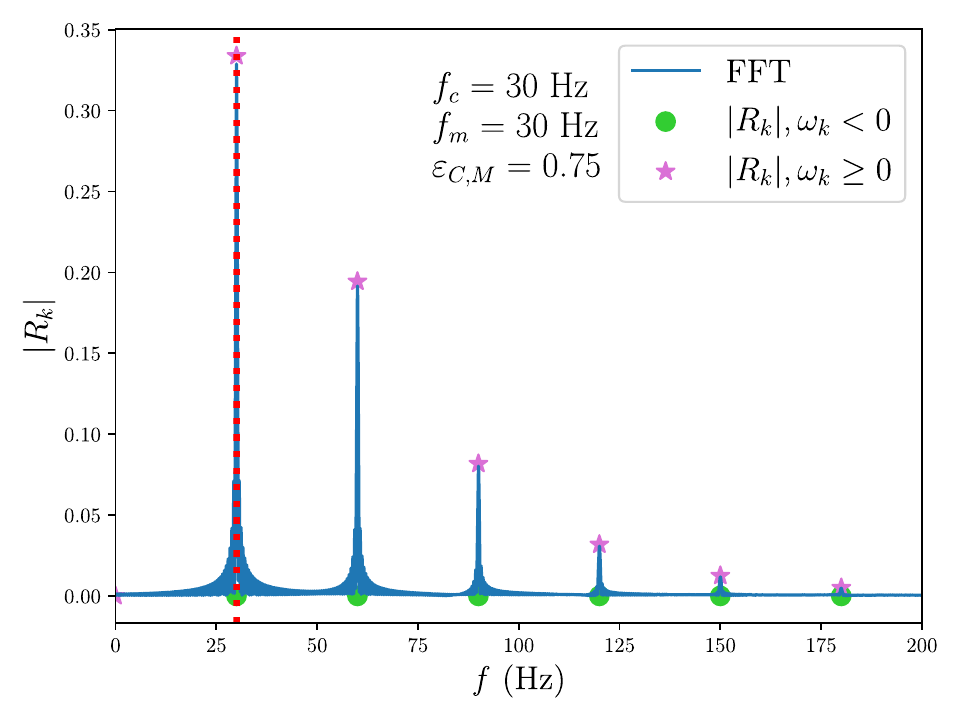} & 
\hspace{-0.3cm}
\includegraphics[width=0.48\linewidth]{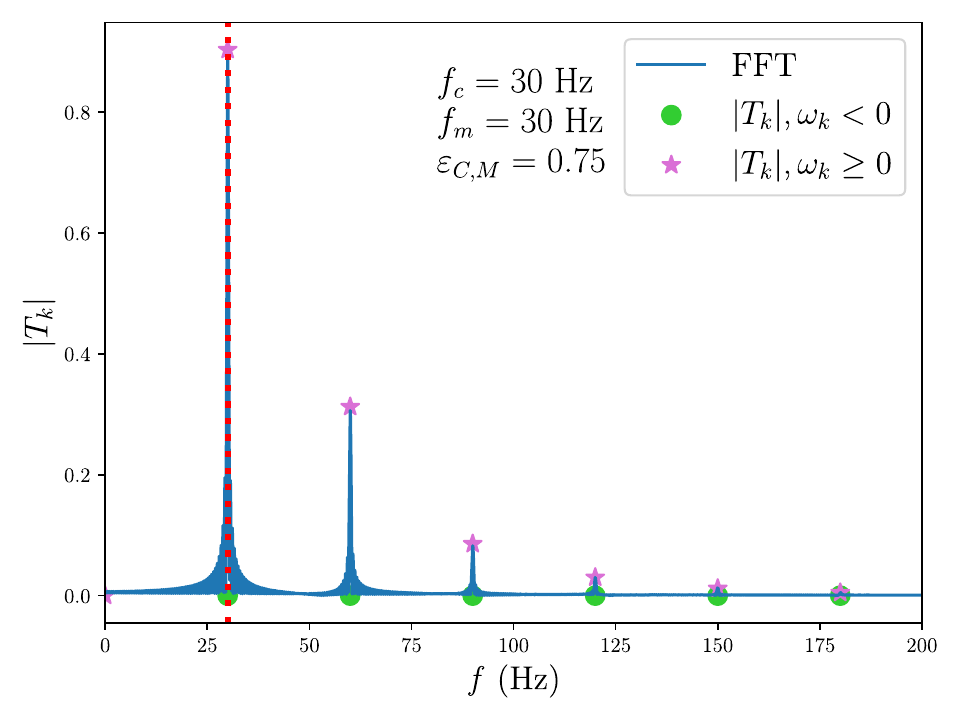} \\
(e) & (f) \\
\hspace{-0.3cm}
\includegraphics[width=0.48\linewidth]{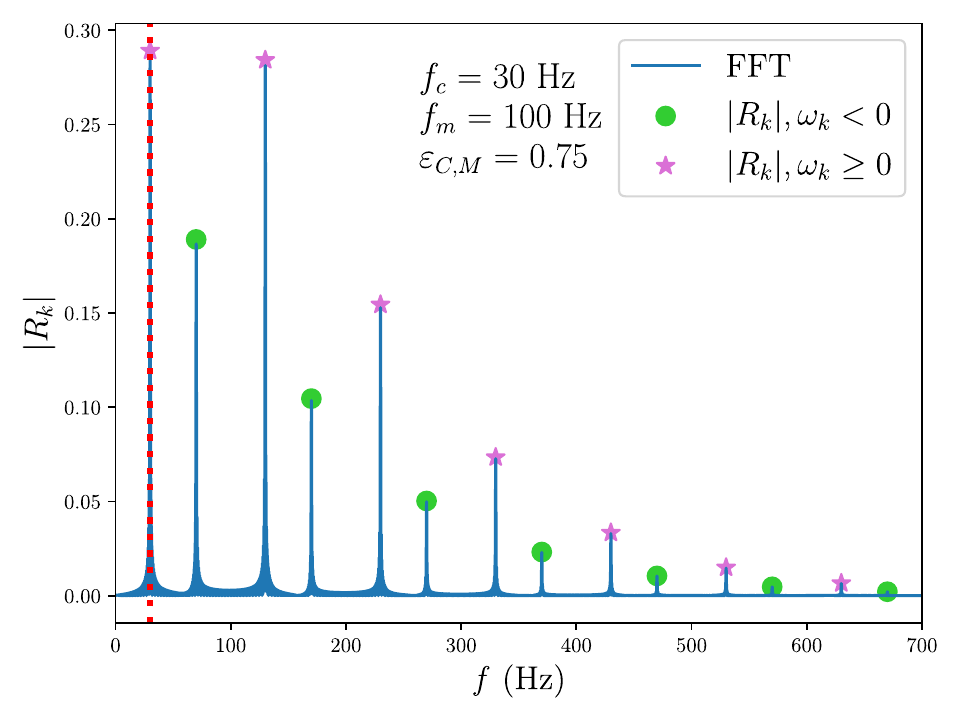} & 
\hspace{-0.3cm}
\includegraphics[width=0.48\linewidth]{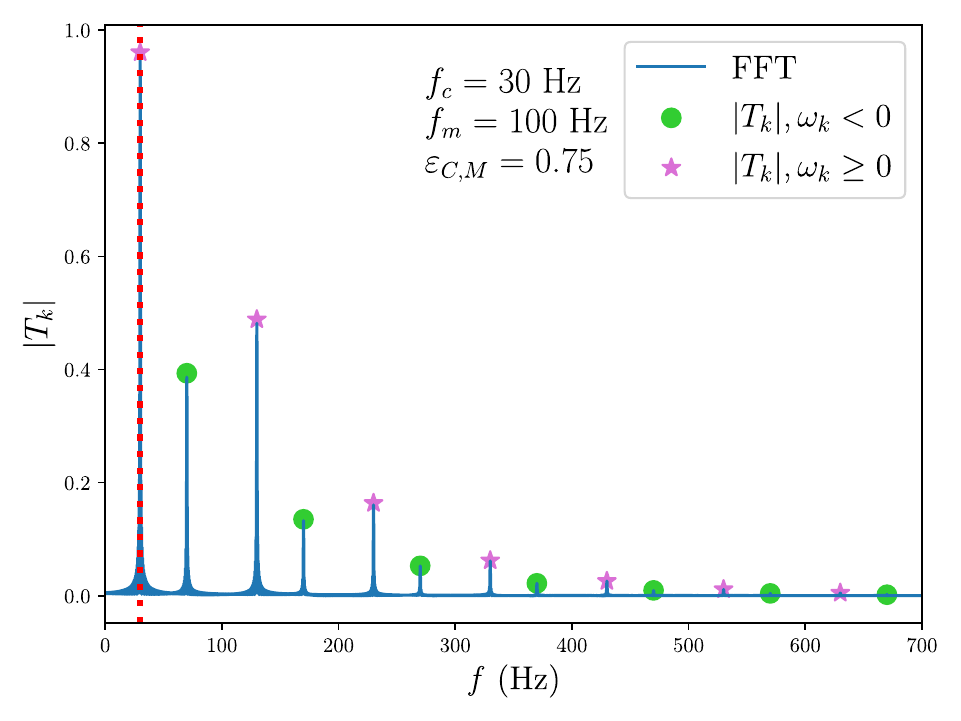}
\end{tabular}
\end{center}
\vspace{-0.5cm}
\caption{\label{FigHarmonics}Generation of harmonics. Coefficients of reflection $R_k$ (left row) and transmission $T_k$ (right row) in \eqref{AnsatzRT}. The vertical dotted line denotes the  frequency $f_c$ of the incident wave. (a-b): $f_c=100$ and $f_m=30$; (c-d): $f_c=f_m=30$; (e-f): $f_c=30$ and $f_m=100$. Blue lines denote the FFT of direct simulations, whereas plain circles and plain stars denote the results of harmonic balance respectively for $\omega_k<0$ and $\omega_k\geq0$.}
\end{figure}

The forcing is a Dirac source point \eqref{Dirac} on the left of the interface $(x_s<x_0$). The velocity field is recorded at a receiver on the right of the interface ($x_r>x_0$), and then a FFT is performed. Three cases are studied considering different frequencies of the source ($f_c=\omega/(2\,\pi)$) and of the modulation ($f_m=\Omega/(2\,\pi)$). Their values are i) $f_c=100$ Hz and $f_m=30$ Hz in the first case, ii) $f_c=f_m=30$ Hz in the second one, and iii) $f_c=30$ Hz and $f_m=100$ Hz in the third one.

Figure \ref{FigHarmonics} compares the Fourier spectra of the reflected and transmitted signals obtained by FFT of direct simulations (blue lines) with the theoretical values obtained by harmonic balance (plain circles and stars). We represent here only the positive part of the frequencies as the spectrum is symmetric. The theoretical coefficient corresponding to negative (resp. positive) values of $\omega_k$ are then plotted in green circles (resp. purple stars). A vertical dotted line is set at the frequency $f_c$. In the three cases, the higher coefficient corresponds to the fundamental mode ($k=0$) and the coefficient values decrease when $|k|$ increase. In each case, an excellent agreement is obtained between results of simulation and of harmonic balance, even for a small number of Fourier modes. The following observations are done:
\begin{itemize}
\item when $f_c>f_m$, the influence of negative frequencies is weak (almost zero) because it corresponds to higher modes. Moreover, the values of the coefficients decrease rapidly. The coefficients of the first harmonics are smaller than the half of the value at the frequency $\omega$;
\item when $f_c=f_m$, then $\omega_{-1}=0$ and the calculations show that $R_{-1}= T_{-1}=0$. In this case, green circles and purple stars are located at the same frequency because $\omega_{-k}=\omega_{k-2}$ and in practice, the coefficients computed for $\omega_k\leq 0$ are strictly zero  and the same observation applies if $f_c$ is a multiple of $f_m$. This is an expected behaviour in the light of Section 4 of \cite{KOUKOURAKI2025103530};
\item when $f_c<f_m$, there is a stronger influence of the modes at negative frequencies. It can be explained by the fact that, as visible on the Figure \ref{FigModulK}, the signal is strongly modified with a high-frequency modulation. The first harmonic for the reflection coefficient is for example almost equal to the fundamental frequency ($\approx0.29$), and there are twelve modes that can be noticed while only five or six were visible for a smaller modulation frequency.
\end{itemize} 


\subsection{Impedance matching}\label{SecSimusImpedance}

The interface is modulated sinusoidally at $f_m=100$ Hz and its parameters are $\mathscr{M}_0=12\,\times 10^3$ kg/m$^2$, $\mathscr{K}_0=941$ MPa/m, so that the properties \eqref{Z_KMT} are verified. Both parameters are modulated at the same frequency and with the same amplitude $\varepsilon_{C,M}=0.75$. Bulk parameters are given in Section \ref{SecSimusSetup}.

\begin{figure}[htbp]
\begin{center}
\begin{tabular}{cc}
(a) & (b) \\
\includegraphics[width=0.48\linewidth]{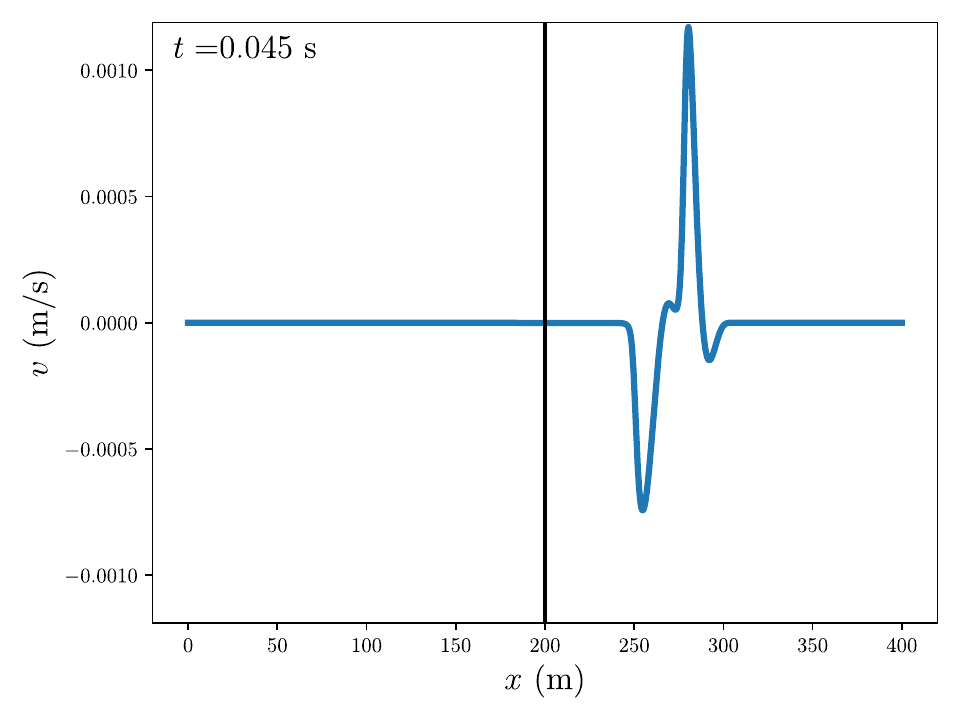} & 
\includegraphics[width=0.48\linewidth]{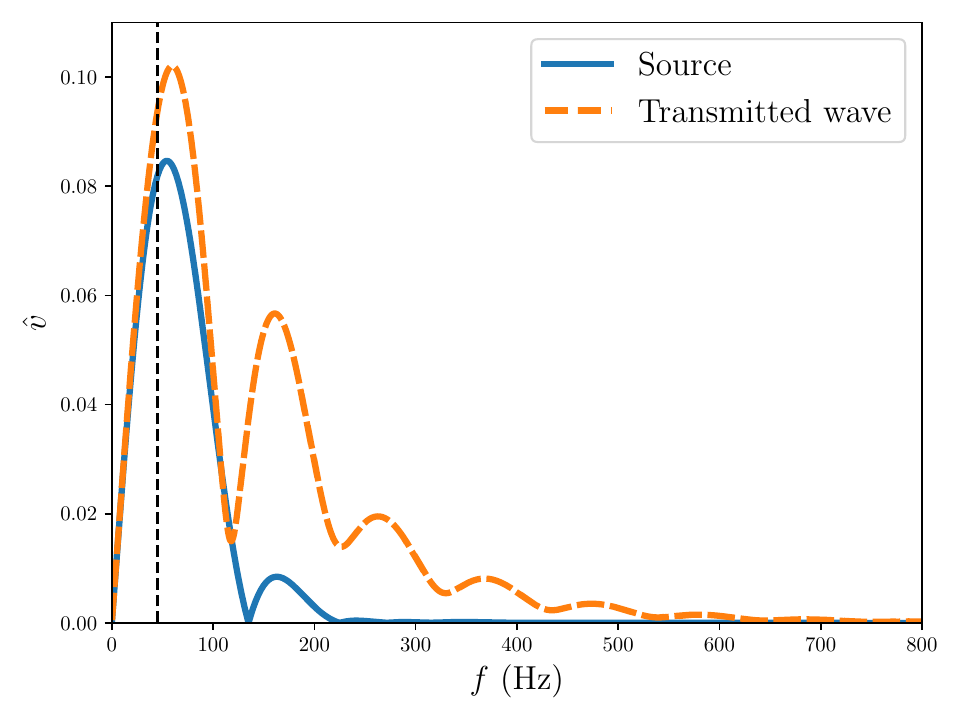}\\ 
\end{tabular}
(c)\\
\includegraphics[width=0.48\linewidth]{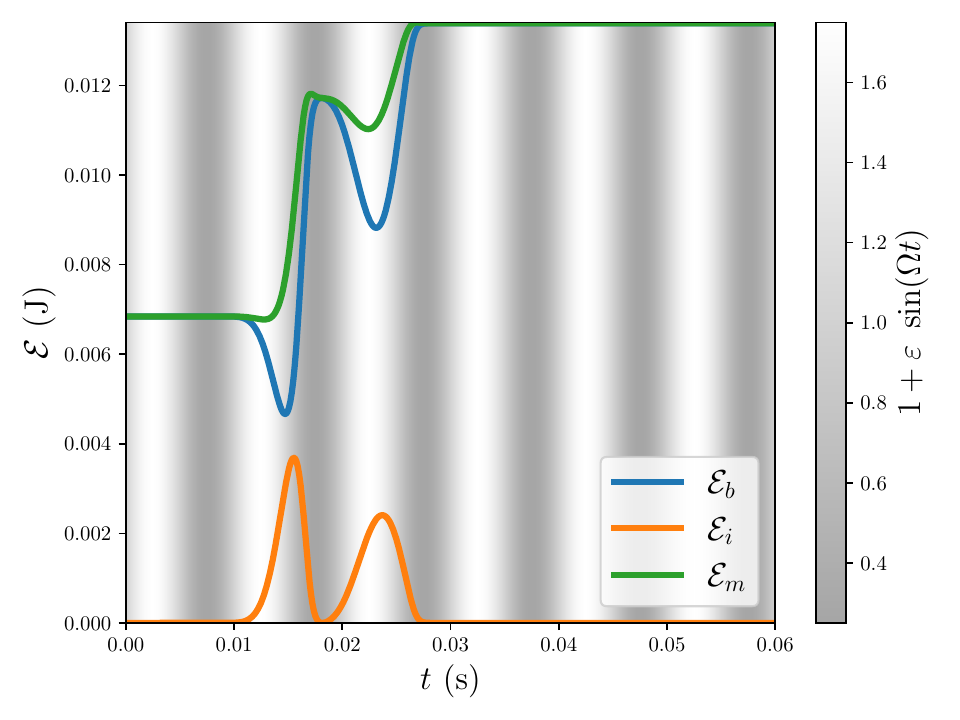}
\end{center}
\vspace{-0.5cm}
\caption{\label{FigNinterfRL}Reflectionless interaction of a pulse centered on $f_c=45$ Hz with a modulated interface due to impedance adaptation. (a): snapshot of the scattered velocity $v$ at $t=0.045$ s. (b): spectra of the incident wave and of the transmitted wave with harmonic generation; the vertical dotted line denotes $f_c$.} (c): time evolution of the energy, where the gray shaded areas denote the evolution of the interface parameters. The modulation function is sinusoidal. 
\end{figure}

Figure \ref{FigNinterfRL} illustrates the behavior of a wave at this modulated interface. On the bottom left side, a snapshot of $v$ is presented. At the left of the interface, no reflected wave propagates because of the impedance matching. The shape of the transmitted signal is similar to that obtained with the first configuration in Figure \ref{FigNinterfAmpli}(a). The spectrum of the transmitted wave is presented in Figure \ref{FigNinterfRL}(b) and compared to the spectrum of the source. The harmonic generation is then visible. As previously, the total amount of energy is modified when the pulse crosses the interface (Figure \ref{FigNinterfRL}(c)). 

Proposition \ref{PropImpedance} about the impedance matching has only been proven for the case of sinusoidal modulation. Nevertheless, Figure \ref{FigNinterfRLNP} shows that the results can be extended to other functions defined by (\ref{FuncModul}) if $\phi_C(t)=\phi_M(t)$. Two cases are tested: a quasi-periodic modulation with $f_m=100$ Hz (a), and a rectangular modulation with $f_m=100$ Hz and $\nu=0.65$ (b). The latter case corresponds to the framework of Remark \ref{RqueImpedance}. In both cases, no reflection is observed.

\begin{figure}[htbp]
\begin{center}
\begin{tabular}{cc} 
(a) & (b) \\
\includegraphics[width=0.48\linewidth]{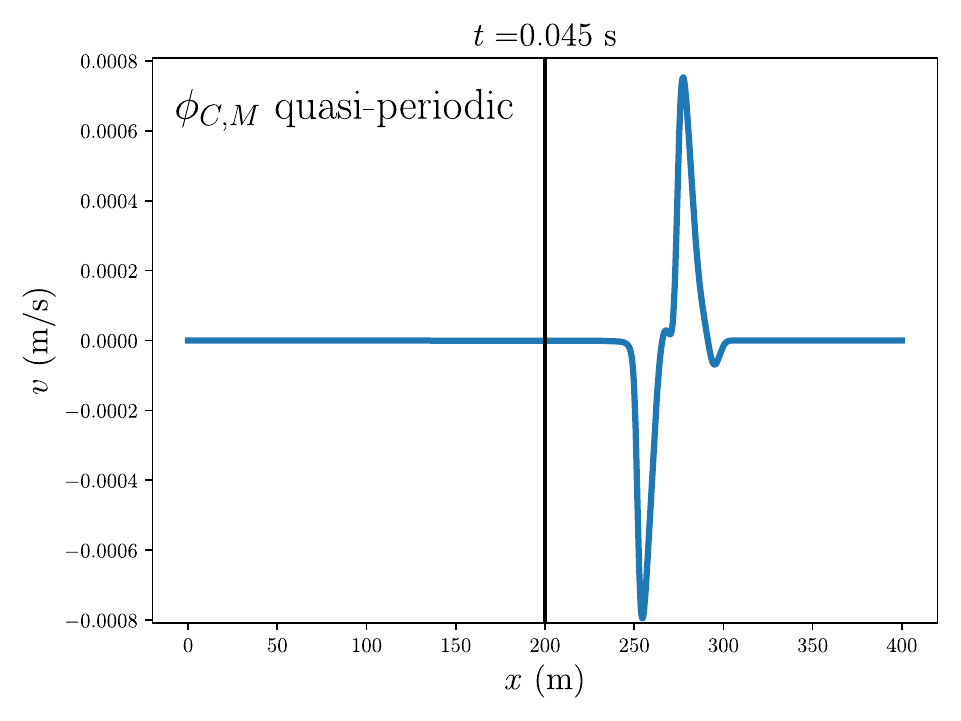} & 
\includegraphics[width=0.48\linewidth]{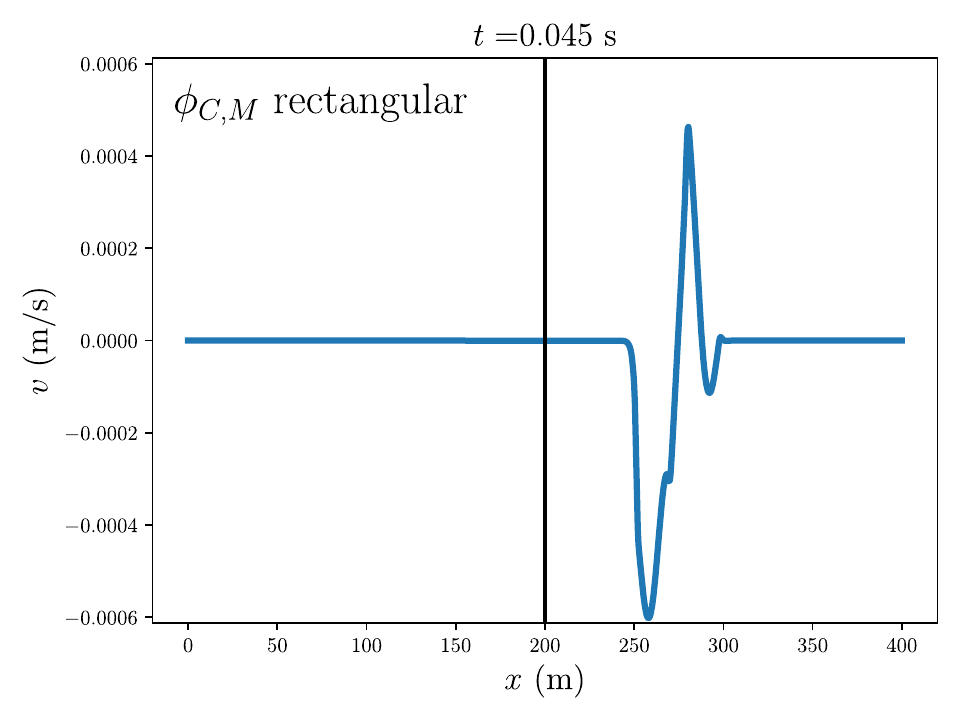}\\ 
\end{tabular}
\end{center}
\vspace{-0.5cm}
\caption{\label{FigNinterfRLNP}Snapshots of reflectionless interaction of a pulse centered on $f_c=45$ Hz with a modulated interface at $t=0.045$ s. Impedance adaptation (a) with a quasi-periodic modulation and (b) with a rectangular modulation.}
\end{figure}


\subsection{Non-reciprocity}\label{SecQSimusNR}

Lastly, we investigate numerically the non-reciprocity induced by time-modulating the parameters of the interface. The configuration is the same as in the previous Section. Figure \ref{FigNonReciprocity} shows a comparison of temporal signals measured when we swap source and sensor (Figure \ref{FigNonReciprocity}(a)) for different values of $f_c$ and $f_m$. The positions of the sources and sensors are $x_s=150$ m and $x_r=225$ m, respectively at a distance $\ell_s$ and $\ell_r$ from the interface. 

\begin{figure}[htpb]
\begin{center}
\begin{tabular}{cc}
(a) & (b) \\
\includegraphics[width=0.3\linewidth]{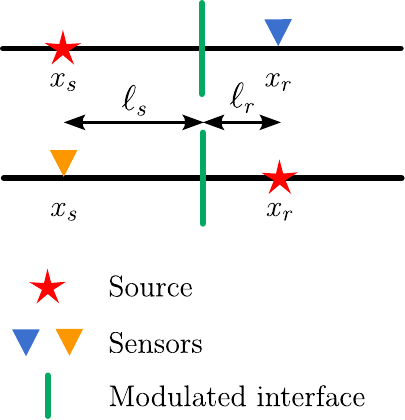} & 
\includegraphics[width=0.48\linewidth]{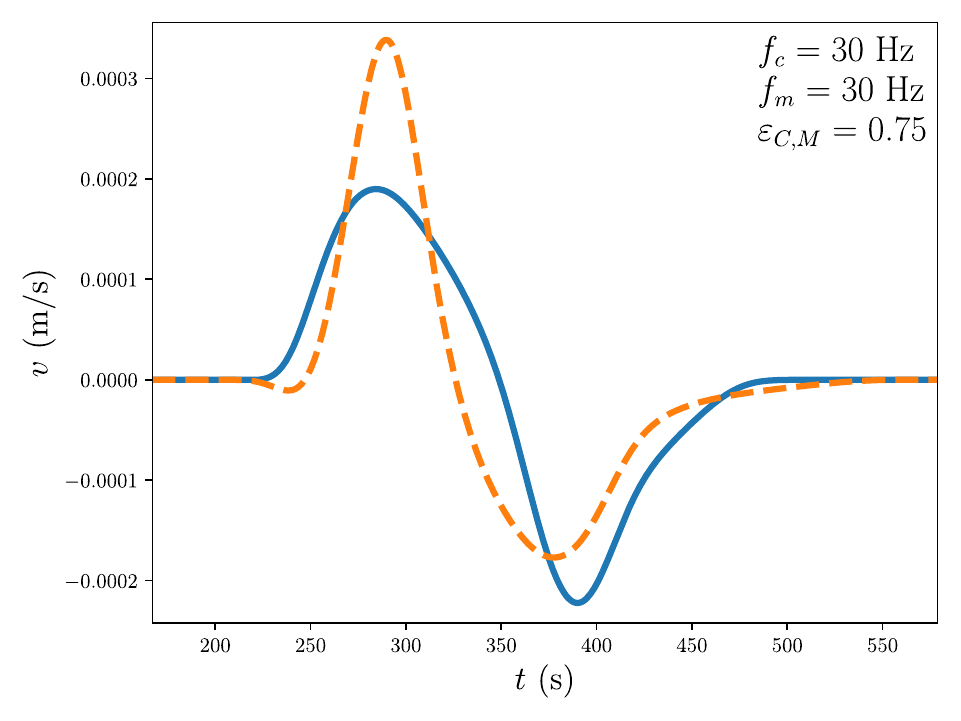}\\ 
(c) & (d) \\
\includegraphics[width=0.48\linewidth]{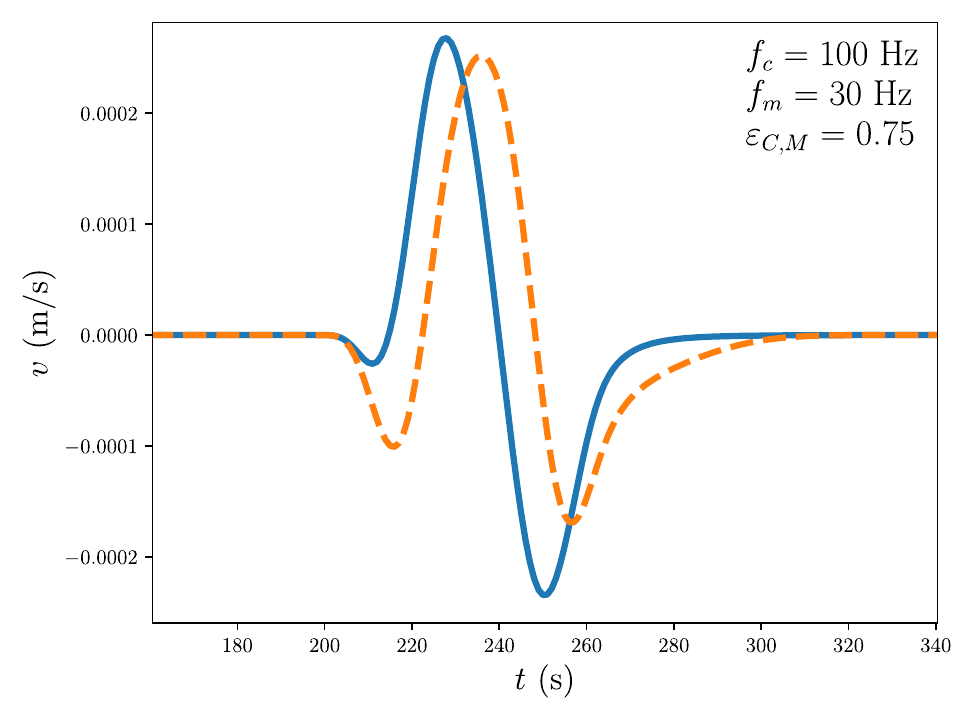} & 
\includegraphics[width=0.48\linewidth]{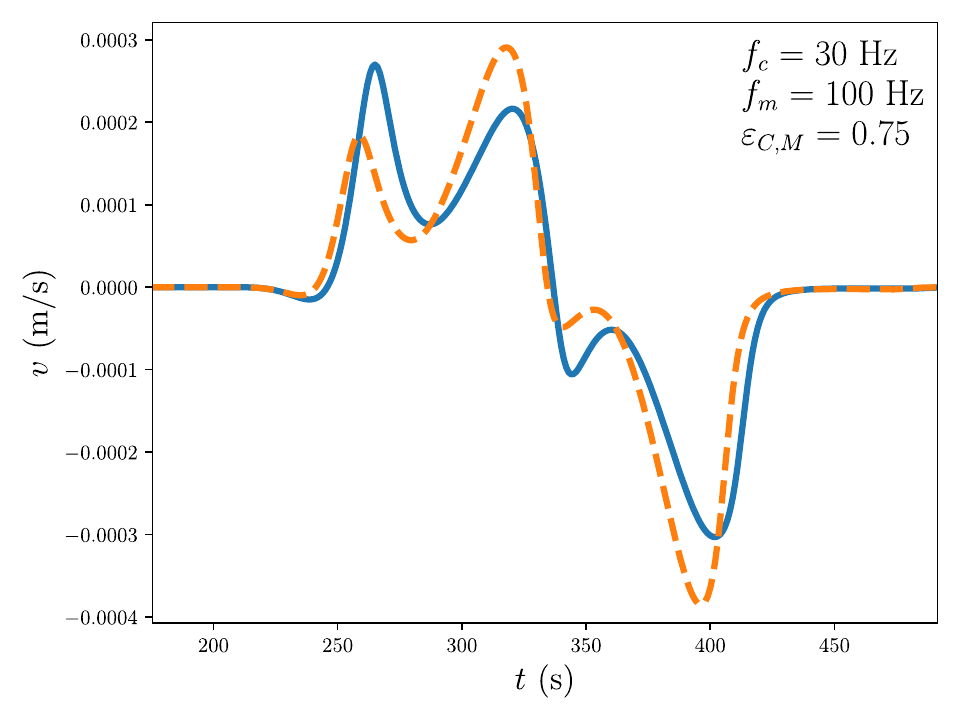}\\ 
\end{tabular}
\end{center}
\vspace{-0.5cm}
\caption{\label{FigNonReciprocity}Illustration of non-reciprocity. (a) Scheme of the numerical experiment where the positions of the sensor and the source is inverted. The signals measured by the sensor are represented for (b)$f_m=f_c=30$ Hz, (c) $f_c=100$ Hz and $f_m=30$ Hz and (d) $f_c=30$ Hz and $f_m=100$ Hz.}
\end{figure}

The transmission of waves through the interface is altered because of the change of state of the interface, when the wave reaches it. It is due to the difference in travel time from the source to the interface. In both cases, the measured signals are different, which illustrates the fact that the reciprocity theorem is no longer verified. 

We propose in this Section to quantify the non-reciprocity by computing the difference between signals. Figure \ref{FigQuantNonReciprocity} presents a measure of the difference between the two signals $v_s$, measured by the sensor at $x=x_r$ due to the source in $x=x_s$, and $v_r$, measured by the sensor at $x=x_s$ due to the source in $x=x_r$ at $t=0.07$ s. This quantity is defined by \begin{equation}
    \vartheta_v=f_c\sqrt{\sum_{i=0}^{N_x}(v_s-v_r)^2}.
\end{equation}
Two types of modulation are studied, a sinusoidal one (\ref{ModulSinus}) and a quasi-periodic one (\ref{FuncModul}) for 100 values of the modulation frequency $f_m$ from $0$ to $400$ Hz. For both the periodic and quasi-periodic modulations, two different frequencies of the pulse are tested: $f_c=30$ Hz represented with the plain blue line and $f_c=50$ Hz with orange dashed line. As the wavelength differs in the two cases, $\vartheta_v$ is normalized using the source frequency.
\begin{figure}[htpb]
\begin{center}
\begin{tabular}{cc}
(a) & (b) \\
\includegraphics[width=0.48\linewidth]{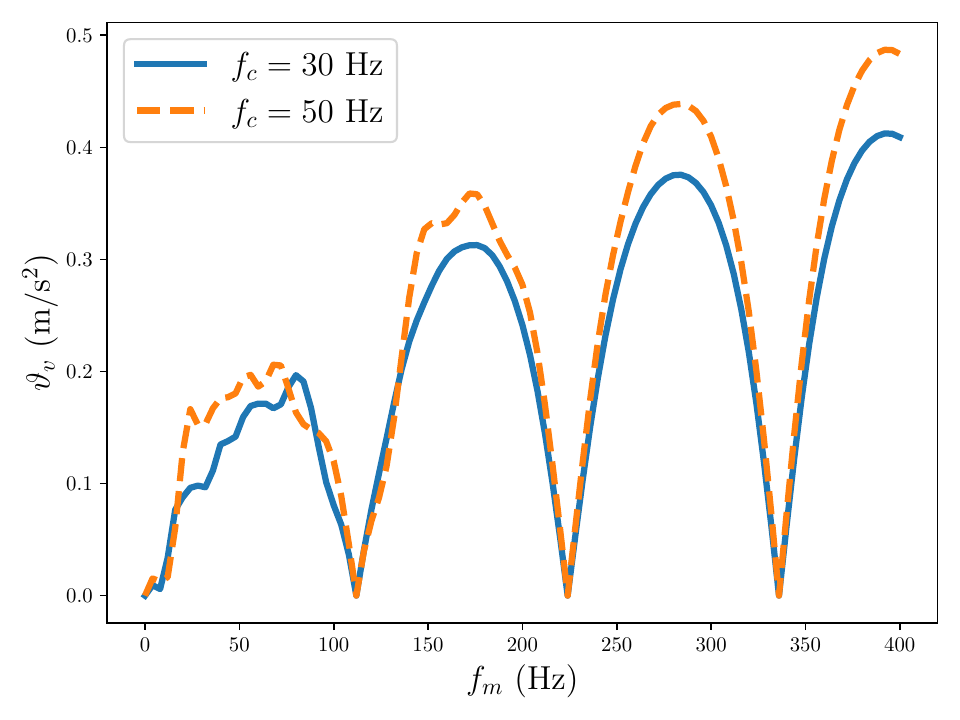} & 
\includegraphics[width=0.48\linewidth]{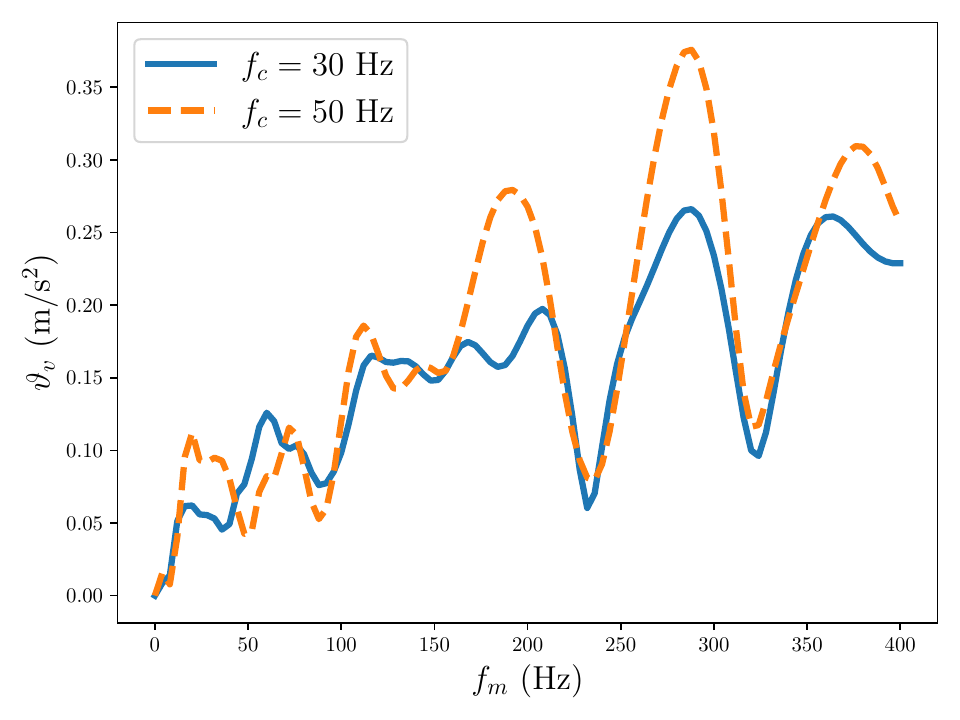}
\end{tabular}
\end{center}
\vspace{-0.5cm}
\caption{\label{FigQuantNonReciprocity}Quantification of non-reciprocity for various modulation frequencies in \eqref{FuncModul}. (a) periodic modulation with two different values of $f_c$: blue line for 30 Hz and orange dashed line for 50 Hz. (b) quasi-periodic modulation with $f_c=30$ Hz and $f_c=50$ Hz.
}
\end{figure}

Measurement with a periodic modulation, as the sinusoidal one treated on Figure \ref{FigQuantNonReciprocity}-(a), shows that some frequencies exist, for a couple source/sensor positions, for which $\vartheta_v=0$. The link between these values is $\ell_s-\ell_r=p \frac{c}{f_m}$, with $p\in\mathbb{Z}$. For these frequencies, the interface is in the same state (due to being phase-shifted by a temporal period) when the incident wave crosses the interface. In particular, the case $p=0$ corresponds to a symmetric case and the measures of non-reciprocity are always zero. It is important to notice that the modulation frequencies with zero non-reciprocity do not depend on the source frequency. One can notice that the non-reciprocity tends to increase with the ratio $\frac{f_m}{f_c}$.

Figure \ref{FigQuantNonReciprocity}-(b) corresponds to a quasi-periodic modulation. Except for the case $\ell_s=\ell_r$, there are no longer values of $f_m\neq 0$ for which $\vartheta_v=0$, so the measures of non-reciprocity are always strictly positive and the signals are always different. Similarly to the case of a periodic modulation in Figure \ref{FigQuantNonReciprocity}-(a), the non-reciprocity tends to increase with the ratio $\frac{f_m}{f_c}$, albeit at a slower rate and with a more complex pattern.


\section{Conclusion}\label{SecConclu}

In this study, we investigated the interaction of waves with an interface whose jump conditions are modulated in time. 
These jump conditions can be experimentally implemented by considering e.g.\ membranes, spring-mass systems or vibrating plates whose thickness is small compared to the wavelength.
Numerical methods for time-domain simulations have been developed. Numerical experiments have demonstrated and validated theoretically predicted properties: amplification of the energy, harmonic generation, impedance matching, and non-reciprocity.

Under the assumptions of Proposition \ref{PropAmpli} (only one sinusoidally modulated interface parameter, the other parameter being zero), we proved that no parametric amplification resonance can occur. That said, we were not able to conclude on this when these assumptions were relaxed. A more complete theoretical study of the energy evolution remains to be carried out, especially in order to better understand the high-frequency regime.

While mentioning resonance, it is worth noting here that we have only considered non-resonant jump conditions. A relaxation of this assumption could be envisaged, and would allow, for example, to study the case of Helmholtz resonators whose physical or geometrical properties vary with time \cite{ammariJoMP2023,mallejacPRA2023}.

Another natural follow-up to this work is to consider a set of modulated interfaces, in phase or not, and periodically distributed in space. In the case where the wavelength is large compared with the spacing between interfaces, asymptotic homogenization can be applied. The work carried out in \cite{bellisJotMaPoS2021} for the case of non-linear but static interfaces is currently generalized to the case of modulated interfaces. This work, currently under finalization, will be the subject of a forthcoming publication. It could also be interesting to study the effective properties induced by a network of time-modulated resonators, as considered in \cite{ammari_effective_2024,tachet2024effective}.

\section*{Conflicts of interest}

The authors declare no competing financial interest.

\section*{Dedication}

The manuscript was written through contributions of all authors. All
authors have given approval to the final version of the manuscript.

\section*{Acknowledgments}
M. D. was funded as a post-doctoral researcher by the Institut M\'ecanique et Ing\'enierie (Marseille, France). S. G. was funded by UK Research and Innvovation (UKRI) under the UK government's Horizon Europe funding guarantee (grant number 10033143). We also would like to thank the reviewers for their useful comments, in particular for the suggestion to study the evolution of energy. Lastly, our thanks go to Vincent Pagneux for the valuable discussions that contributed to this work.



\appendix

\section{Scattering coefficients in the static case}\label{AppRT}

In the case of static jump conditions ($\mathscr{C}(t)=\mathscr{C}_0$,  $\mathscr{M}(t)=\mathscr{M}_0$, $\mathscr{Q}_{C}(t)=\mathscr{Q}_{C_0}$ and $\mathscr{Q}_{M}(t)=\mathscr{Q}_{M_0}$), the scattering coefficients can be calculated analytically. The physical parameters are assumed to be piecewise constant:
$$
(\rho,\,c)=\left\{
\begin{array}{l}
\ds (\rho_0,\,c_0)\quad \mbox{in  }\Omega_0\quad(x<x_0),\\ [8pt]
\ds (\rho_1,\,c_1)\quad \mbox{in  }\Omega_1\quad(x>x_0).
\end{array}
\right.
$$ Upon introducing the following quantities
\begin{equation}
\begin{array}{c}
\ds Z_0=\rho_0\,c_0,\quad Z_1=\rho_1\,c_1,\quad Z_2=\rho_0\,c_1,\\ [8pt]
\ds Y_0=\mathscr{C}_0\,Z_0\,Z_1-\mathscr{M}_0,\quad Y_1=\mathscr{C}_0\,Z_0\,Z_1+\mathscr{M}_0,\quad \omega_\sharp^2=\frac{\mathscr{C}_0\mathscr{M}_0}{4},
\end{array}
\label{NotaImpedance}
\end{equation}
the reflection coefficient $R$ and the transmission coefficient $T$ at frequency $\omega$ are given by
\begin{equation}
\begin{array}{l}
\ds R(\omega)=\frac{\ds (Z_1-Z_0)\,\left(1-\left(\frac{\omega}{\omega_\sharp}\right)^2+\frac{\mathscr{Q}_{C_0}\,\mathscr{Q}_{M_0}}{4}\right)-\mathscr{Q}_{C_0}\,Z_0\,Z_1+\mathscr{Q}_{M_0}-i\,\omega\left(Y_0+\frac{1}{4}\left(\mathscr{Q}_{C_0}\,\mathscr{M}_0+\mathscr{Q}_{M_0}\,\mathscr{C}_0\right)\,(Z_1-Z_0)\right)}{\ds (Z_1+Z_0)\,\left(1-\left(\frac{\omega}{\omega_\sharp}\right)^2+\frac{\mathscr{Q}_{C_0}\,\mathscr{Q}_M}{4}\right)+\mathscr{Q}_{C_0}\,Z_0\,Z_1+\mathscr{Q}_{M_0}+i\,\omega\left(Y_1+\frac{1}{4}\left(\mathscr{Q}_{C_0}\,\mathscr{M}_0+\mathscr{Q}_{M_0}\,\mathscr{C}_0\right)\,(Z_1+Z_0)\right)},\\ [14pt]
\ds T(\omega)=\frac{\ds 2\,Z_2\,\left(1+\left(\frac{\omega}{\omega_\sharp}\right)^2-\frac{\mathscr{Q}_{C_0}\,\mathscr{Q}_{M_0}}{4}-i\,\omega\frac{1}{4}\left(\mathscr{Q}_{C_0}\,\mathscr{M}_0+\mathscr{Q}_{M_0}\,\mathscr{C}_0\right)\right)}{\ds (Z_1+Z_0)\,\left(1-\left(\frac{\omega}{\omega_\sharp}\right)^2+\frac{\mathscr{Q}_{C_0}\,\mathscr{Q}_{M_0}}{4}\right)+\mathscr{Q}_{C_0}\,Z_0\,Z_1+\mathscr{Q}_{M_0}+i\,\omega\left(Y_1+\frac{1}{4}\left(\mathscr{Q}_{C_0}\,\mathscr{M}_0+\mathscr{Q}_{M_0}\,\mathscr{C}_0\right)\,(Z_1+Z_0)\right)}.
\end{array}
\label{RT}
\end{equation}

In the particular case where the density and the Young's modulus are constant, the impedance $Z=\rho\,c=\sqrt{\rho\,E}$ is constant. Additionally, let us assume that the two following conditions are satisfied
\begin{equation}
\mathscr{M}_0=Z^2\,\mathscr{C}_0,\hspace{1cm}
\mathscr{Q}_{M_0}=Z^2\,\mathscr{Q}_{C_0}.
\end{equation}
Then $R(\omega)=0$: no reflected wave is generated at a static imperfect interface, at any frequency. One recovers the impedance matching condition given in Proposition \ref{PropImpedance}. Moreover, these coefficients can also be obtained by calculating $R_0$ and $T_0$ using (\ref{RkTk}) when the interface parameters are static.


\section{Analytical solution of scattered fields}\label{AppAnalytic}

Here we determine the exact solution of \eqref{BVP} with initial data \eqref{Cauchy}. As in Appendix \ref{AppRT}, the physical parameters are assumed to be piecewise constant.
We consider the case where only $\mathscr{C}(t)$ and $\mathscr{Q}_C(t)$ are and the interface stiffness is modulated, whereas the interface inertia and the corresponding dissipation term is taken to be zero ($\mathscr{M}(t)=0$ and $\mathscr{Q}_M(t)=0$). The reciprocal case $\mathscr{M}(t)\neq 0$, $\mathscr{Q}_C(t)=0$ and $\mathscr{C}(t)=0$ can be treated in a similar manner. In both cases, the calculation relies on the method of characteristics, illustrated on Figure \ref{fig:Riemann}. The proof follows the lines of \cite{lombardJoCaAM2007}.

\begin{figure}[htbp]
    \centering
    \includegraphics[width=0.6\linewidth]{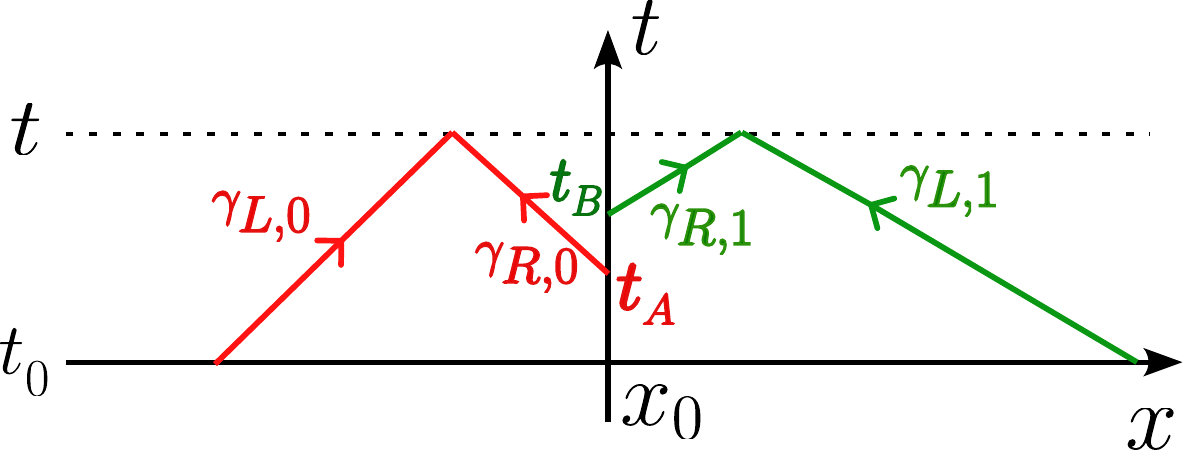}
    \caption{Sketch of the method of characteristics for computing semi-analytical solution using Riemann invariants.}
    \label{fig:Riemann}
\end{figure}
To express ${\bf U}(x,\,t)$ in terms of limit-values of the fields at $x_0$, we use the {\it Riemann invariants} $J^{R,L}$ that are constant along the {\it characteristics} $\gamma_{R,L}$. The invariants for linear Partial Differential Equations with constant coefficients satisfy
\begin{equation}
\left\{
\begin{array}{l}
\ds \gamma_R:\,\frac{\ds dx}{\ds dt}=+c\,\Rightarrow\,\left.\frac{\ts dJ^R}{\ts dt}\right|_{\gamma_R}=0,\quad \mbox{ with }J^R(x,\,t)=\frac{\ds 1}{\ds 2}\left(v-\frac{\ds 1}{\ds \rho\,c}\sigma\right)(x,\,t),\\
[10pt]
\ds \gamma_L:\,\frac{\ds dx}{\ds dt}=-c\,\Rightarrow\,\left.\frac{\ts dJ^L}{\ts dt}\right|_{\gamma_L}=0,\quad \mbox{ with }J^L(x,\,t)=\frac{\ds 1}{\ds 2}\left(v+\frac{\ts 1}{\ts \rho\,c}\sigma\right)(x,\,t).
\end{array}
\right.
\label{Jrl}
\end{equation}
Using the continuity of stress \eqref{JC1interf-M} and the invariant \eqref{Jrl}, along with the initial data condition of compact support in $\Omega_0$, one obtains for $t\geq t_0$
\begin{equation}
\begin{array}{l}
\sigma^\pm(t)=-\rho_1\,c_1\,v^+(t),\\[8pt]
\ds
v^-(t)=-\frac{\ds \rho_1\,c_1}{\ts \rho_0\,c_0}\,v^+(t)+2\,J_0^R\left(x_0-c_0\,t,\,0\right),
\end{array}
\label{SUpm} 
\end{equation}
where the subscript $i$ on $J^{R,L}_i$ refers to $\Omega_i$. The solution ${\bf U}(x,\,t)$ can be expressed in terms of $v^+(s)$, with $t_0\leq s \leq t$, and of the initial values of the Riemann invariants. Introducing the travel times
$$
t_A =t-\frac{\textstyle 1}{\textstyle c_0}(x_0-x), \qquad
t_B=t-\frac{\textstyle 1}{\textstyle c_1}(x-x_0),
$$
the solution ${\bf U}(x,\,t)$ is given in $\Omega_0$ by
\begin{equation}
\begin{array}{l}
{\bf U}(x,\,t)=
\left(
\begin{array}{cc}
1           &  1\\
[4pt]
-\rho_0\,c_0 & \rho_0\,c_0
\end{array}
\right)\,
\left(
\begin{array}{c}
J_0^R\left(x-c_0\,t,\,0\right)\\
[4pt]
\Delta_A(x,\,t)
\end{array}
\right),\\
\\
\ds \mbox{with } \Delta_A(x,\,t)=
\left\{
\begin{array}{l}
\ds
-\frac{\ts \rho_1\,c_1}{\textstyle \rho_0\,c_0}\,v^+(t_A)+J_0^R(x_0-c_0\,t_A,\,0)
\,\mbox{ if }\,t_A \geq 0,\\
[8pt]
J_0^L(x+c_0\,t,\,0) \, \mbox{ otherwise},
\end{array}
\right.
\end{array}
\end{equation}
and in $\Omega_1$ by
\begin{equation}
\begin{array}{l}
{\bf U}(x,\,t)=
\left(
\begin{array}{c}
1\\
[4pt]
-\rho_1\,c_1
\end{array}
\right)
\Delta_B(x,\,t),\\
\\
\ds \mbox{with } \Delta_B(x,\,t)=
\left\{
\begin{array}{l}
\ds
v^+(t_B)\, \mbox{ if }\,t_B \geq t_0,\\
[4pt]
0 \, \mbox{ otherwise}.
\end{array}
\right.
\end{array}
\label{U(x,t)}
\end{equation}
To complete the analytical solution, it remains to determine $v^+(t)$. For this purpose, the jump of displacement \eqref{JC1interf-K} is differentiated, leading to
\begin{equation}
v^+(t)-v^-(t)=\left(\mathscr{C}(t)\right)'\sigma^+(t)+\mathscr{C}(t)\,\partial_t\sigma^+(t)+\mathscr{Q}_C(t)\,\sigma^+(t).
\end{equation}
From \eqref{SUpm}, it follows 
\begin{equation}
\left(1+\frac{\ts \rho_1\,c_1}{\rho_0\,c_0}+\rho_1c_1\left(\mathscr{C}'(t)+\mathscr{Q}_C(t)\right)\right)\,v^+(t)-2\,J_0^R\left(x_0-c_0\,t,\,0\right)=-\rho_1\,c_1\,\mathscr{C}(t)\,\partial_tv^+(t).
\label{eq:eqdify}
\end{equation}
Setting
\begin{equation}
y(t)=v^+(t),\qquad h(t)=2\,J_0^R\left(x_0-c_0\,t,\,0\right),
\end{equation}
one obtains the non-autonomous Ordinary Differential Equation
if $\mathscr{C}(t)\neq0$:
\begin{equation}
y'(t)=\frac{\ts 1}{\ts \rho_1\,c_1\,\mathscr{C}(t)}\left( h(t)-\left(1+\frac{\ts \rho_1\,c_1}{\rho_0\,c_0}+\rho_1\,c_1\left(\mathscr{C}'(t)+\mathscr{Q}_C(t)\right)\right)\,y(t)\right),
\label{ODE}
\end{equation}
which is integrated by the usual Runge-Kutta 4 method.
If $\mathscr{C}(t)=0$ and $\mathscr{Q}_C\neq 0$, (\ref{eq:eqdify}) gives:
\begin{equation}
    v^+(t)=\frac{2\,J_0^R\left(x_0-c_0\,t,\,0\right)}{\ds 1+\frac{ \rho_1\,c_1}{\rho_0\,c_0}+\rho_1c_1\mathscr{Q}_C(t)}.
\end{equation}
This solution lies within the bounds defined by the solutions associated with the extreme (minimum and maximum) values of $\mathscr{Q}_C(t)$.

The reciprocal case $\mathscr{M}(t)\neq 0$, $\mathscr{Q}_M\neq 0$ and $\mathscr{C}(t)=0$, $\mathscr{Q}_C(t)=0$ can be treated using the following expressions:
\begin{equation}
\begin{array}{l}
 \Delta_A(x,\,t)=
\left\{
\begin{array}{l}
\ds
-\frac{\ts 1}{\textstyle \rho_1\,c_1}\,\sigma^+(t_A)-J_0^R(x_0-c_0\,t_A,\,0)
\,\mbox{ if }\,t_A \geq t_0,\\
[8pt]
J_0^L(x+c_0\,t,\,0) \, \mbox{ otherwise},
\end{array}
\right.
\end{array}
\end{equation}
\begin{equation}
\begin{array}{l}
 \Delta_B(x,\,t)=
\left\{
\begin{array}{l}
\ds
-\frac{\sigma^+(t_B)}{\rho_1\,c_1}\, \mbox{ if }\,t_B \geq t_0,\\
[4pt]
0 \, \mbox{ otherwise}.
\end{array}
\right.
\end{array}
\label{V(x,t)}
\end{equation}
and by solving 
\begin{equation}
z'(t)=-\frac{\ts \rho_1\,c_1}{\ts \mathscr{M}(t)}\left( g(t)+\left(1+\frac{\ts \rho_0\,c_0}{\rho_1\,c_1}+\frac{\ts \left(\mathscr{M}'(t) + \mathscr{Q}_M(t)\right)}{\ts \rho_1\,c_1}\right)\,z(t)\right),
\label{ODEM}
\end{equation}
where
\begin{equation}
z(t)=\sigma^+(t),\qquad g(t)=2\,\rho_0c_0J_0^R\left(x_0-c_0\,t,\,0\right).
\end{equation}
As previously, the particular case where $\mathscr{M}(t)=0$ and $\mathscr{Q}_M(t)\neq 0$ leads to
\begin{equation}
    \sigma^+(t)=-\frac{2\,\rho_0c_0J_0^R\left(x_0-c_0\,t,\,0\right)}{\ds 1+\frac{ \rho_0\,c_0}{\rho_1\,c_1}+\frac{ \mathscr{Q}_M(t)}{ \rho_1\,c_1}},
\end{equation}
which is bounded by the static solutions obtained with the extreme values of $\mathscr{Q}_M(t)$.


\section{Bounded amplification}\label{AppAmpli}

Here we prove Proposition \ref{PropAmpli}. The assumptions are:
\begin{itemize}
    \item homogeneous medium, with $\rho_0\,c_0=\rho_1\,c_1=Z$;
    \item source point at $x_s<x_0$, with a bounded source function $S$;
    \item $T$-periodic modulation of the compliance and of the corresponding dissipation parameter \eqref{KM-T}, with $\phi_X(t+T)=\phi_X(t)=\varphi_X(\Omega t)$, $T=2\,\pi/\Omega$, and $X=C,Q_C$;
    \item continuity of stress ($\mathscr{M}_0=0$ and $\mathscr{Q}_M=0$). Similar conclusions hold with a modulated inertia and continuity of velocity ($\mathscr{C}_0=0$ and $\mathscr{Q}_C=0$).
\end{itemize}

The mean of a function $w(t)$ over $T$ is denoted by
\begin{equation}
\overline{w}=\frac{1}{T}\int_0^Tw(t)\,dt.
\end{equation}
We need the following Lemmas.

\begin{Lemma}
Consider the Ordinary Differential Equation
\begin{equation}
    y'(t)=\alpha(t)\,y+\beta(t),
    \label{ODE_Duhamel}
\end{equation}
where $\alpha(t)$ is a $T$-periodic function, and $\beta(t)$ is bounded. If $\overline{\alpha}\leq 0$, then $y(t)$ is bounded.
\label{LemmaDuhamel}
\end{Lemma}

\begin{proof}
The integrating factor method yields
\begin{equation}
y(t)=e^{A(t)}\left(y_0+\int_0^t e^{-A(s)}\,\beta(s)\,ds\right),\qquad \mbox{with } A(t)=\int_0^t\alpha(s)\,ds.
    \label{Duhamel}
\end{equation}
If $\overline{\alpha}\leq 0$, then $e^{A(t)}$ does not grow exponentially and $y(t)$ is bounded.
\end{proof}

\begin{Lemma}
Let the function
\begin{equation}
\Ng(\Na,\Nb)=1+\frac{\Na}{\Nb}\left(\sqrt{1-\Nb^2}-1\right)
\label{FuncG}
\end{equation}
defined on the domain ${\mathcal{D}}=\{(\Na,\Nb)\in\mathbb{R}^2/\,|\Na|<1,\,|\Nb|<1,\,\Nb\neq 0\}$. Then $\Ng(\Na,\Nb)>0$ for all $(\Na,\Nb)\in {\mathcal{D}}$.
\label{LemmaG}
\end{Lemma}

\begin{proof}
One introduces
$$
\Nh(\Na,\Nb)=-\frac{\Na}{\Nb}\left(1-\sqrt{1-\Nb^2}\right).
$$
For $(\Na,\Nb)\in{\mathcal{D}}$, one has
$$
1-\sqrt{1-\Nb^2}\in(0,1),\quad \left|\frac{\Na}{\Nb}\right|<\frac{1}{|\Nb|},
$$
so that
$$
|\Nh(\Na,\Nb)|<\frac{1}{|\Nb|}\left(1-\sqrt{1-\Nb^2} \right):= H(\Nb).
$$
The real function $H$ satisfies
$$
H(\Nb)=\frac{|\Nb|}{2}+o(\Nb)\mathop{\rightarrow}\limits_{\Nb\rightarrow 0}0,\quad H(\Nb)\mathop{\rightarrow}\limits_{\Nb\rightarrow \pm 1}1.
$$
For $\Nb>0$, its derivative writes
$$
H'(\Nb)=\frac{1-\sqrt{1-\Nb^2}}{\Nb^2 \sqrt{1-\Nb^2}}>0,
$$
so that $H$ is increasing from 0 (at $\Nb=0$) to 1 (at $\Nb=1$). Lastly, $H$ is a continuous even function. For all $\Nb\in(-1,1)$, one deduces that
$$
0<H(\Nb)<1.
$$

It follows $|\Nh(\Na,\Nb)|<1$, and hence
$$
\Ng(\Na,\Nb)=1+\Nh(\Na,\Nb)>1-1=0,
$$
which concludes the proof.
\end{proof}

Following the same lines as in Appendix \ref{AppAnalytic}, it can be proven that $y(t)=v^+(t)$ satisfies the Ordinary Differential Equation
\begin{equation}
    y'=a(t)\,y+b(t),
    \label{ODE-bounded}
\end{equation}
with
\begin{equation}
        a(t)=-\frac{1}{Z}\,\frac{2+Z\,\left(\mathscr{C}'(t)+\mathscr{Q}_C(t)\right)}{\mathscr{C}(t)},\qquad b(t)=-\frac{1}{Z}\,\frac{S(t)}{\mathscr{C}(t)}.
    \label{ABt}
\end{equation}
The function $a(t)$ is $T$-periodic, and $b(t)$ is bounded. If $\overline{a}\leq 0$, then the conclusion of Lemma \ref{LemmaDuhamel} holds and $y(t)$ is bounded.
Let us consider again the modulation $\varphi_X(t)=\sin(t)$, with $X=C,Q_C$. From \eqref{ABt}, one obtains
\begin{equation}
    a(t)=-\frac{1}{Z\,\mathscr{C}_0}\,f(t),
    \label{FuncA}
\end{equation}
with
\begin{equation}
    f(t)=\frac{2+Z\mathscr{Q}_{C_0}+\gamma\,\cos(\Omega t)+Z\mathscr{Q}_{C_0}\varepsilon_{Q_C}\sin(\Omega t)}{1+\varepsilon_C\,\sin(\Omega t)},\qquad \gamma=Z\,\mathscr{C}_0\,\varepsilon_C\,\Omega.
\end{equation}

The mean value of $f$ writes
\begin{equation}
    \begin{array}{lll}
         \ds \overline{f} &=& \ds \frac{1}{T}\int_0^T\frac{2+Z\mathscr{Q}_{C_0}+\gamma\,\cos(\Omega t)+Z\mathscr{Q}_{C_0}\varepsilon_{Q_C}\sin(\Omega t)}{1+\varepsilon_C\,\sin(\Omega t)}\,dt, \\ [10pt]
         &=& \ds \frac{1}{2\,\pi}\left((2+Z\mathscr{Q}_{C_0})\int_0^{2\pi}\frac{ds}{1+\varepsilon_C\,\sin(s)}+\gamma\int_0^{2\pi}\frac{\cos(s)}{1+\varepsilon_C\,\sin(s)}\,ds\,+\,Z\mathscr{Q}_{C_0}\varepsilon_{Q_C}\int_0^{2\pi}\frac{\sin(s)}{1+\varepsilon_C\,\sin(s)}\,ds \right).
    \end{array}
\end{equation}
Classical change of variables, and integration of a periodic function over one period, gives
\begin{equation}
\begin{array}{l}
\ds   \int_0^{2\pi}\frac{ds}{1+\varepsilon_C\,\sin(s)}=\frac{2\,\pi}{\sqrt{1-\varepsilon_C^2}},\qquad  \int_0^{2\pi}\frac{\cos(s)}{1+\varepsilon_C\,\sin(s)}\,ds=0, \\ [8pt] \ds \int_0^{2\pi}\frac{\sin(s)}{1+\varepsilon_C\,\sin(s)}\,ds=\frac{2\pi}{\varepsilon_C}\left(1-\frac{1}{\sqrt{1-\varepsilon_C^2}}\right),
\end{array}
\end{equation}

Straightforward calculations give
\begin{equation}
\overline{f}=\frac{1}{\sqrt{1-\varepsilon_C^2}}\left(2+Z\mathscr{Q}_{C_0}\,\Ng(\varepsilon_C,\varepsilon_{Q_C})\right),
\end{equation}
where $\Ng(.,.)$ has been introduced in \eqref{FuncG}. Equation \eqref{FuncA} and Lemma \ref{LemmaG} yield
\begin{equation}
    \overline{a}=-\frac{1}{Z\,\mathscr{C}_0}\,\frac{2+Z\mathscr{Q}_{C_0}\,\Ng(\varepsilon_C,\varepsilon_{Q_C})}{\sqrt{1-\varepsilon_C^2}}<0.
\end{equation}
Lemma \ref{LemmaDuhamel} implies that the solution of \eqref{ODE-bounded} is bounded, which concludes the proof of Proposition~\ref{PropAmpli}.

%

\bibliographystyle{crunsrt}
%
%
\bibliography{CREM_2024}

\def\bysame{\leavevmode ---------\thinspace}
\makeatletter\if@francais\providecommand{\og}{<<~}\providecommand{\fg}{~>>}
\else\gdef\og{``}\gdef\fg{''}\fi\makeatother
\def\cdrandname{\&}
\providecommand\cdrnumero{no.~}
\providecommand{\cdredsname}{eds.}
\providecommand{\cdredname}{ed.}
\providecommand{\cdrchapname}{chap.}
\providecommand{\cdrmastersthesisname}{Memoir}
\providecommand{\cdrphdthesisname}{PhD Thesis}
\begin{thebibliography}{10}

\bibitem{craster2013}
R.~V. Craster, S.~Guenneau, \emph{Acoustic {{Metamaterials}: Absorption,
  Cloaking, Imaging, Time-Modulated Media, and Topological Crystals}}, Springer
  Series in Material Science (2nd Edition), 2024.

\bibitem{caloz2019spacetime1}
C.~Caloz, Z.-L. Deck-L{\'e}ger, {\og Spacetime metamaterials—part I: general
  concepts\fg}, \emph{IEEE Transactions on Antennas and Propagation}
  \textbf{68} (2019), \cdrnumero 3, p.~1569-1582.

\bibitem{caloz2019spacetime2}
C.~Caloz, Z.-L. Deck-Leger, {\og Spacetime metamaterials—Part II: Theory and
  applications\fg}, \emph{IEEE Transactions on Antennas and Propagation}
  \textbf{68} (2019), \cdrnumero 3, p.~1583-1598.

\bibitem{galiffiA2022}
E.~Galiffi, R.~Tirole, S.~Yin, H.~Li, S.~Vezzoli, P.~A. Huidobro, M.~G.
  Silveirinha, R.~Sapienza, A.~Al{\`u}, J.~B. Pendry, {\og Photonics of
  Time-Varying Media\fg}, \emph{Advanced Photonics} \textbf{4} (2022),
  \cdrnumero 1, p.~014002.

\bibitem{wang_temporal_2025}
S.~Wang, N.~Shao, H.~Chen, J.~Chen, H.~Qian, Q.~Wu, H.~Duan, A.~Alu, G.~Huang,
  {\og Temporal refraction and reflection in modulated mechanical metabeams:
  theory and physical observation\fg},  (2025),
  \url{https://arxiv.org/abs/2501.09989}.

\bibitem{Swinteck2015}
N.~Swinteck, S.~Matsuo, K.~Runge, J.~O. Vasseur, P.~Lucas, P.~A. Deymier, {\og
  Bulk elastic waves with unidirectional backscattering-immune topological
  states in a time-dependent superlattice\fg}, \emph{Journal of Applied
  Physics} \textbf{118} (2015), \cdrnumero 6, p.~063103.

\bibitem{Lustig2018}
E.~Lustig, Y.~Sharabi, M.~Segev, {\og Topological aspects of photonic time
  crystals\fg}, \emph{Optica} \textbf{5} (2018), \cdrnumero 11, p.~1390.

\bibitem{Salehi2022}
M.~Salehi, P.~Rahmatian, M.~Memarian, K.~Mehrany, {\og Frequency conversion in
  time-varying graphene microribbon arrays\fg}, \emph{Optics Express}
  \textbf{30} (2022), \cdrnumero 18, p.~32061.

\bibitem{yi_frequency_2017}
K.~Yi, M.~Collet, S.~Karkar, {\og Frequency conversion induced by time-space
  modulated media\fg}, \emph{Physical Review B} \textbf{96} (2017), \cdrnumero
  10, p.~104110.

\bibitem{torrentPRB2018}
D.~Torrent, W.~J. Parnell, A.~N. Norris, {\og Loss Compensation in
  Time-Dependent Elastic Metamaterials\fg}, \emph{Physical Review B}
  \textbf{97} (2018), \cdrnumero 1, p.~014105.

\bibitem{kimPRE2023}
B.~L. Kim, C.~Chong, S.~Hajarolasvadi, Y.~Wang, C.~Daraio, {\og Dynamics of
  Time-Modulated, Nonlinear Phononic Lattices\fg}, \emph{Physical Review E}
  \textbf{107} (2023), \cdrnumero 3, p.~034211.

\bibitem{kiorpelidisPRB2024}
I.~Kiorpelidis, F.~K. Diakonos, G.~Theocharis, V.~Pagneux, {\og Transient
  Amplification in Stable {{Floquet}} Media\fg}, \emph{Physical Review B}
  \textbf{110} (2024), \cdrnumero 13, p.~134315.

\bibitem{cullen1958travelling}
A.~Cullen, {\og A travelling-wave parametric amplifier\fg}, \emph{Nature}
  \textbf{181} (1958), \cdrnumero 4605, p.~332-332.

\bibitem{morgenthaler1958velocity}
F.~R. Morgenthaler, {\og Velocity modulation of electromagnetic waves\fg},
  \emph{IRE Transactions on microwave theory and techniques} \textbf{6} (1958),
  \cdrnumero 2, p.~167-172.

\bibitem{tien1958parametric}
P.~K. Tien, {\og Parametric amplification and frequency mixing in propagating
  circuits\fg}, \emph{Journal of Applied Physics} \textbf{29} (1958),
  \cdrnumero 9, p.~1347-1357.

\bibitem{Cassedy1963}
E.~Cassedy, A.~Oliner, {\og Dispersion relations in time-space periodic media:
  Part I{\textemdash}Stable interactions\fg}, \emph{Proceedings of the {IEEE}}
  \textbf{51} (1963), \cdrnumero 10, p.~1342-1359.

\bibitem{Cassedy1967}
E.~Cassedy, {\og Dispersion relations in time-space periodic media part
  {II}{\textemdash}Unstable interactions\fg}, \emph{Proceedings of the {IEEE}}
  \textbf{55} (1967), \cdrnumero 7, p.~1154-1168.

\bibitem{fante1971transmission}
R.~Fante, {\og Transmission of electromagnetic waves into time-varying
  media\fg}, \emph{IEEE Transactions on Antennas and Propagation} \textbf{19}
  (1971), \cdrnumero 3, p.~417-424.

\bibitem{chu1972wave}
R.~Chu, T.~Tamir, {\og Wave propagation and dispersion in space-time periodic
  media\fg}, in \emph{Proceedings of the Institution of Electrical Engineers},
  vol. 119, IET, 1972, p.~797-806.

\bibitem{weekes2001numerical}
S.~L. Weekes, {\og Numerical computation of wave propagation in dynamic
  materials\fg}, \emph{Applied numerical mathematics} \textbf{37} (2001),
  \cdrnumero 4, p.~417-440.

\bibitem{maestre2007spatio}
F.~Maestre, A.~M{\"u}nch, P.~Pedregal, {\og A spatio-temporal design problem
  for a damped wave equation\fg}, \emph{SIAM Journal on Applied Mathematics}
  \textbf{68} (2007), \cdrnumero 1, p.~109-132.

\bibitem{lurie2007introduction}
K.~A. Lurie, \emph{An introduction to the mathematical theory of dynamic
  materials}, vol.~15, Springer, 2007.

\bibitem{maestre2008dynamic}
F.~Maestre, P.~Pedregal, {\og Dynamic materials for an optimal design problem
  under the two-dimensional waveequation\fg}, \emph{Discrete and Continuous
  Dynamical Systems} \textbf{23} (2008), \cdrnumero 3, p.~973-990.

\bibitem{jensen2009space}
J.~S. Jensen, {\og Space--time topology optimization for one-dimensional wave
  propagation\fg}, \emph{Computer Methods in Applied Mechanics and Engineering}
  \textbf{198} (2009), \cdrnumero 5-8, p.~705-715.

\bibitem{to2009homogenization}
H.~T. To, {\og Homogenization of dynamic laminates\fg}, \emph{Journal of
  mathematical analysis and applications} \textbf{354} (2009), \cdrnumero 2,
  p.~518-538.

\bibitem{jensen2010optimization}
J.~S. Jensen, {\og Optimization of space-time material layout for 1D wave
  propagation with varying mass and stiffness parameters\fg}, \emph{Control and
  Cybernetics} \textbf{39} (2010), \cdrnumero 3, p.~599-614.

\bibitem{sanguinet2011homogenized}
W.~Sanguinet, {\og The homogenized equations of motion for an activated elastic
  lamination in plane strain\fg}, \emph{ZAMM-Journal of Applied Mathematics and
  Mechanics/Zeitschrift f{\"u}r Angewandte Mathematik und Mechanik} \textbf{91}
  (2011), \cdrnumero 12, p.~944-956.

\bibitem{fang2012realizing}
K.~Fang, Z.~Yu, S.~Fan, {\og Realizing effective magnetic field for photons by
  controlling the phase of dynamic modulation\fg}, \emph{Nature photonics}
  \textbf{6} (2012), \cdrnumero 11, p.~782-787.

\bibitem{yuan2016photonic}
L.~Yuan, Y.~Shi, S.~Fan, {\og Photonic gauge potential in a system with a
  synthetic frequency dimension\fg}, \emph{Optics letters} \textbf{41} (2016),
  \cdrnumero 4, p.~741-744.

\bibitem{milton2017field}
G.~W. Milton, O.~Mattei, {\og Field patterns: a new mathematical object\fg},
  \emph{Proceedings of the Royal Society A: Mathematical, Physical and
  Engineering Sciences} \textbf{473} (2017), \cdrnumero 2198, p.~20160819.

\bibitem{movchanPTRSMPES2022}
A.~B. Movchan, N.~V. Movchan, I.~S. Jones, G.~W. Milton, H.-M. Nguyen, {\og
  Frontal Waves and Transmissions for Temporal Laminates and Imperfect Chiral
  Interfaces\fg}, \emph{Philosophical Transactions of the Royal Society A:
  Mathematical, Physical and Engineering Sciences} \textbf{380} (2022),
  \cdrnumero 2231, p.~20210385.

\bibitem{ammariJoMP2023}
H.~Ammari, J.~Cao, E.~O. Hiltunen, L.~Rueff, {\og Transmission Properties of
  Time-Dependent One-Dimensional Metamaterials\fg}, \emph{Journal of
  Mathematical Physics} \textbf{64} (2023), \cdrnumero 12, p.~121502.

\bibitem{Ammari2022}
H.~Ammari, J.~Cao, E.~O. Hiltunen, {\og NonReciprocal Wave Propagation in
  Space-Time Modulated Media\fg}, \emph{Multiscale Modeling \& Simulation}
  \textbf{20} (2022), \cdrnumero 4, p.~1228-1250.

\bibitem{Nassar2017}
H.~Nassar, X.~Xu, A.~Norris, G.~Huang, {\og Modulated phononic crystals:
  Non-reciprocal wave propagation and Willis materials\fg}, \emph{Journal of
  the Mechanics and Physics of Solids} \textbf{101} (2017), p.~10-29.

\bibitem{Huidobro2019}
P.~A. Huidobro, E.~Galiffi, S.~Guenneau, R.~V. Craster, J.~B. Pendry, {\og
  Fresnel drag in space{\textendash}time-modulated metamaterials\fg},
  \emph{Proceedings of the National Academy of Sciences} \textbf{116} (2019),
  \cdrnumero 50, p.~24943-24948.

\bibitem{farhat2021spacetime}
M.~Farhat, S.~Guenneau, P.-Y. Chen, Y.~Wu, {\og Spacetime modulation in
  floating thin elastic plates\fg}, \emph{Physical Review B} \textbf{104}
  (2021), \cdrnumero 1, p.~014308.

\bibitem{Huidobro2021}
P.~Huidobro, M.~Silveirinha, E.~Galiffi, J.~Pendry, {\og Homogenization Theory
  of Space-Time Metamaterials\fg}, \emph{Physical Review Applied} \textbf{16}
  (2021), \cdrnumero 1, p.~014044.

\bibitem{Pham2022}
K.~Pham, A.~Maurel, {\og Diffraction grating with space-time modulation\fg},
  \emph{Journal of Computational Physics} \textbf{469} (2022), p.~111528.

\bibitem{li2023space}
Z.~Li, X.~Ma, A.~Bahrami, Z.-L. Deck-L{\'e}ger, C.~Caloz, {\og Space-time
  Fresnel prism\fg}, \emph{Physical Review Applied} \textbf{20} (2023),
  \cdrnumero 5, p.~054029.

\bibitem{Touboul2024}
M.~Touboul, B.~Lombard, R.~C. Assier, S.~Guenneau, R.~V. Craster, {\og
  High-order homogenization of the time-modulated wave equation:
  non-reciprocity for a single varying parameter\fg}, \emph{Proceedings of the
  Royal Society A: Mathematical, Physical and Engineering Sciences}
  \textbf{480} (2024), \cdrnumero 2289.

\bibitem{ashkin1958parametric}
A.~Ashkin, {\og Parametric amplification of space charge waves\fg},
  \emph{Journal of Applied Physics} \textbf{29} (1958), \cdrnumero 12,
  p.~1646-1651.

\bibitem{couder2005walking}
Y.~Couder, S.~Protiere, E.~Fort, A.~Boudaoud, {\og Walking and orbiting
  droplets\fg}, \emph{Nature} \textbf{437} (2005), \cdrnumero 7056, p.~208-208.

\bibitem{lira2012electrically}
H.~Lira, Z.~Yu, S.~Fan, M.~Lipson, {\og Electrically driven nonreciprocity
  induced by interband photonic transition on a silicon chip\fg},
  \emph{Physical review letters} \textbf{109} (2012), \cdrnumero 3, p.~033901.

\bibitem{taravati2016mixer}
S.~Taravati, C.~Caloz, {\og Mixer-duplexer-antenna leaky-wave system based on
  periodic space-time modulation\fg}, \emph{IEEE transactions on antennas and
  propagation} \textbf{65} (2016), \cdrnumero 2, p.~442-452.

\bibitem{fink2016loschmidt}
M.~Fink, {\og From Loschmidt daemons to time-reversed waves\fg},
  \emph{Philosophical Transactions of the Royal Society A: Mathematical,
  Physical and Engineering Sciences} \textbf{374} (2016), \cdrnumero 2069,
  p.~20150156.

\bibitem{mallejac2023scattering}
M.~Mall{\'e}jac, R.~Fleury, {\og Scattering from time-modulated
  transmission-line loads: theory and experiments in acoustics\fg},
  \emph{Physical Review Applied} \textbf{19} (2023), \cdrnumero 6, p.~064012.

\bibitem{tessierbrothelandeAPL2023}
S.~Tessier~Brothelande, C.~Cro{\"e}nne, F.~Allein, J.~O. Vasseur, M.~Amberg,
  F.~Giraud, B.~Dubus, {\og Experimental Evidence of Nonreciprocal Propagation
  in Space-Time Modulated Piezoelectric Phononic Crystals\fg}, \emph{Applied
  Physics Letters} \textbf{123} (2023), \cdrnumero 20, p.~201701.

\bibitem{Tirole2022}
R.~Tirole, E.~Galiffi, J.~Dranczewski, T.~Attavar, B.~Tilmann, Y.-T. Wang,
  P.~A. Huidobro, A.~Al{\'{u}}, J.~B. Pendry, S.~A. Maier, S.~Vezzoli,
  R.~Sapienza, {\og Saturable Time-Varying Mirror Based on an Epsilon-Near-Zero
  Material\fg}, \emph{Physical Review Applied} \textbf{18} (2022), \cdrnumero
  5, p.~054067.

\bibitem{Moussa2023}
H.~Moussa, G.~Xu, S.~Yin, E.~Galiffi, Y.~Ra’di, A.~Alù, {\og Observation of
  temporal reflection and broadband frequency translation at photonic time
  interfaces\fg}, \emph{Nature Physics} \textbf{19} (2023), \cdrnumero 6,
  p.~863-868.

\bibitem{Lustig2023}
E.~Lustig, O.~Segal, S.~Saha, E.~Bordo, S.~N. Chowdhury, Y.~Sharabi,
  A.~Fleischer, A.~Boltasseva, O.~Cohen, V.~M. Shalaev, M.~Segev, {\og
  Time-refraction optics with single cycle modulation\fg}, \emph{Nanophotonics}
  \textbf{12} (2023), \cdrnumero 12, p.~2221-2230.

\bibitem{Harwood2024}
A.~C. Harwood, S.~Vezzoli, T.~V. Raziman, C.~Hooper, R.~Tirole, F.~Wu, S.~A.
  Maier, J.~B. Pendry, S.~A.~R. Horsley, R.~Sapienza, {\og Super-luminal
  Synthetic Motion with a Space-Time Optical Metasurface\fg},  (2024),
  \url{https://arxiv.org/abs/2407.10809}.

\bibitem{Goldsberry2025}
B.~M. Goldsberry, A.~N. Norris, S.~P. Wallen, M.~R. Haberman, {\og Green’s
  function approach to model vibrations of beams with spatio-temporally
  modulated properties\fg}, \emph{Proceedings of the Royal Society A:
  Mathematical, Physical and Engineering Sciences} \textbf{481} (2025),
  \cdrnumero 2311.

\bibitem{Goldsberry2022}
B.~M. Goldsberry, S.~P. Wallen, M.~R. Haberman, {\og Nonreciprocity and mode
  conversion in a spatiotemporally modulated elastic wave circulator\fg},
  \emph{Physical Review Applied} \textbf{17} (2022), \cdrnumero 3, p.~034050.

\bibitem{wenCP2022}
X.~Wen, X.~Zhu, A.~Fan, W.~Y. Tam, J.~Zhu, H.~W. Wu, F.~Lemoult, M.~Fink,
  J.~Li, {\og Unidirectional Amplification with Acoustic Non-{{Hermitian}}
  Space-Time Varying Metamaterial\fg}, \emph{Communications Physics} \textbf{5}
  (2022), \cdrnumero 1, p.~1-7.

\bibitem{Galiffi2024}
E.~Galiffi, A.~C. Harwood, S.~Vezzoli, R.~Tirole, A.~Alù, R.~Sapienza, {\og
  Optical coherent perfect absorption and amplification in a time-varying
  medium\fg},  (2024), \url{https://arxiv.org/abs/2410.16426}.

\bibitem{Croenne2019}
C.~Croënne, J.~O. Vasseur, O.~Bou~Matar, A.-C. Hladky-Hennion, B.~Dubus, {\og
  Non-reciprocal behavior of one-dimensional piezoelectric structures with
  space-time modulated electrical boundary conditions\fg}, \emph{Journal of
  Applied Physics} \textbf{126} (2019), \cdrnumero 14, p.~145108.

\bibitem{Nassar2020}
H.~Nassar, B.~Yousefzadeh, R.~Fleury, M.~Ruzzene, A.~Al{\`{u}}, C.~Daraio,
  A.~N. Norris, G.~Huang, M.~R. Haberman, {\og Nonreciprocity in acoustic and
  elastic materials\fg}, \emph{Nature Reviews Materials} \textbf{5} (2020),
  \cdrnumero 9, p.~667-685.

\bibitem{yi_reflection_2018}
K.~Yi, M.~Collet, S.~Karkar, {\og Reflection and transmission of waves incident
  on time-space modulated media\fg}, \emph{Physical Review B} \textbf{98}
  (2018), \cdrnumero 5, p.~054109.

\bibitem{zhuPRB2020}
X.~Zhu, J.~Li, C.~Shen, G.~Zhang, S.~A. Cummer, L.~Li, {\og Tunable
  Unidirectional Compact Acoustic Amplifier via Space-Time Modulated
  Membranes\fg}, \emph{Physical Review B} \textbf{102} (2020), \cdrnumero 2,
  p.~024309.

\bibitem{zhuAPL2020}
X.~Zhu, J.~Li, C.~Shen, X.~Peng, A.~Song, L.~Li, S.~A. Cummer, {\og
  Non-Reciprocal Acoustic Transmission via Space-Time Modulated Membranes\fg},
  \emph{Applied Physics Letters} \textbf{116} (2020), \cdrnumero 3, p.~034101.

\bibitem{KOUKOURAKI2025103530}
M.~Koukouraki, P.~Petitjeans, A.~Maurel, V.~Pagneux, {\og Floquet scattering of
  shallow water waves by a vertically oscillating plate\fg}, \emph{Wave Motion}
  \textbf{136} (2025), p.~103530.

\bibitem{puJoSaV2024}
X.~Pu, A.~Marzani, A.~Palermo, {\og A Multiple Scattering Formulation for
  Elastic Wave Propagation in Space--Time Modulated Metamaterials\fg},
  \emph{Journal of Sound and Vibration} \textbf{573} (2024), p.~118199.

\bibitem{mallejacPRA2023}
M.~Mall{\'e}jac, R.~Fleury, {\og Scattering from {{Time-Modulated
  Transmission-Line Loads}}: {{Theory}} and {{Experiments}} in
  {{Acoustics}}\fg}, \emph{Physical Review Applied} \textbf{19} (2023),
  \cdrnumero 6, p.~064012.

\bibitem{ammari_scattering_2024}
H.~Ammari, J.~Cao, E.~O. Hiltunen, L.~Rueff, {\og Scattering from
  time-modulated subwavelength resonators\fg}, \emph{Proceedings of the Royal
  Society A: Mathematical, Physical and Engineering Sciences} \textbf{480}
  (2024), \cdrnumero 2289, p.~20240177.

\bibitem{ammari_spacetime_2025}
H.~Ammari, E.~O. Hiltunen, L.~Rueff, {\og Space–time wave localization in
  systems of subwavelength resonators\fg}, \emph{Proceedings of the Royal
  Society A: Mathematical, Physical and Engineering Sciences} \textbf{480}
  (2025), \cdrnumero 2289, p.~20240177.

\bibitem{assierPRSMPES2020}
R.~C. Assier, M.~Touboul, B.~Lombard, C.~Bellis, {\og High-Frequency
  Homogenization in Periodic Media with Imperfect Interfaces\fg},
  \emph{Proceedings of the Royal Society A: Mathematical, Physical and
  Engineering Sciences} \textbf{476} (2020), \cdrnumero 2244, p.~20200402.

\bibitem{lombardSJSC2003}
B.~Lombard, J.~Piraux, {\og How to {{Incorporate}} the {{Spring-Mass
  Conditions}} in {{Finite-Difference Schemes}}\fg}, \emph{SIAM Journal on
  Scientific Computing} \textbf{24} (2003), \cdrnumero 4, p.~1379-1407.

\bibitem{lombardSJSC2006}
B.~Lombard, J.~Piraux, {\og Numerical Modeling of Elastic Waves across
  Imperfect Contacts.\fg}, \emph{SIAM Journal on Scientific Computing}
  \textbf{28} (2006), \cdrnumero 1, p.~172-205.

\bibitem{benjaziaWM2014}
A.~Ben~Jazia, B.~Lombard, C.~Bellis, {\og Wave Propagation in a Fractional
  Viscoelastic {{Andrade}} Medium: {{Diffusive}} Approximation and Numerical
  Modeling\fg}, \emph{Wave Motion} \textbf{51} (2014), \cdrnumero 6,
  p.~994-1010.

\bibitem{touboulJoCP2020}
M.~Touboul, B.~Lombard, C.~Bellis, {\og Time-Domain Simulation of Wave
  Propagation across Resonant Meta-Interfaces\fg}, \emph{Journal of
  Computational Physics} \textbf{414} (2020), p.~109474.

\bibitem{assierJFM2014}
R.~C. Assier, X.~Wu, {\og Linear and Weakly Nonlinear Instability of a Premixed
  Curved Flame under the Influence of Its Spontaneous Acoustic Field\fg},
  \emph{Journal of Fluid Mechanics} \textbf{758} (2014), p.~180-220.

\bibitem{bellisJotMaPoS2021}
C.~Bellis, B.~Lombard, M.~Touboul, R.~Assier, {\og Effective Dynamics for
  Low-Amplitude Transient Elastic Waves in a {{1D}} Periodic Array of
  Non-Linear Interfaces\fg}, \emph{Journal of the Mechanics and Physics of
  Solids} \textbf{149} (2021), p.~104321.

\bibitem{ammari_effective_2024}
H.~Ammari, J.~Cao, E.~O. Hiltunen, L.~Rueff, {\og Effective Medium Theory for
  Time-modulated Subwavelength Resonators\fg},  (2024),
  \url{https://arxiv.org/abs/2412.13545v2}.

\bibitem{tachet2024effective}
S.~Tachet, K.~Pham, A.~Maurel, {\og Effective models of
  metasurface/metainterface made of three-dimensional Helmholtz resonators\fg},
  in \emph{INTER-NOISE and NOISE-CON Congress and Conference Proceedings}, vol.
  270, Institute of Noise Control Engineering, 2024, p.~784-788.

\bibitem{lombardJoCaAM2007}
B.~Lombard, J.~Piraux, {\og Modeling 1-{{D}} Elastic {{P-waves}} in a Fractured
  Rock with Hyperbolic Jump Conditions\fg}, \emph{Journal of Computational and
  Applied Mathematics} \textbf{204} (2007), \cdrnumero 2, p.~292-305.

\end{thebibliography}

\end{document}